\pdfoutput=1
\RequirePackage{fixltx2e}
\documentclass[10pt,letterpaper]{article}  
\usepackage[utf8]{inputenc}
\usepackage{amsmath, amsfonts, amsthm, amssymb}
\usepackage{graphicx, xcolor}
\usepackage{tikz}
\usepackage{accents, verbatim}
\usepackage{microtype}
\usepackage[textheight = 680pt, textwidth = 480pt,footskip=30pt]{geometry}
\usepackage[nottoc]{tocbibind} 
\usepackage{hyperref,footnotebackref}
\hypersetup{linktoc = all, hidelinks, bookmarksnumbered}
\usepackage{enumitem}
\setlist[2]{nosep}

\usepackage[mathcal]{euscript}
\usepackage{dsfont, mathrsfs}
\usepackage{pxfonts}
\renewcommand{\emph}{\textsl}
\makeatletter
\newcommand*\smallbullet{\mathpalette\smallbullet@{.5}}
\newcommand*\smallbullet@[2]{\mathbin{\vcenter{\hbox{\scalebox{#2}{$\m@th#1\bullet$}}}}}
\makeatother
\DeclareSymbolFont{origAMSa}{U}{msa}{m}{n}
\DeclareMathSymbol{\Box}{\mathord}{origAMSa}{3}
\renewcommand{\geq}{\geqslant}
\renewcommand{\leq}{\leqslant}
\theoremstyle{plain}
\newtheorem{theorem}{Theorem}[section]
\newtheorem{definition}[theorem]{Definition}
\newtheorem{proposition}[theorem]{Proposition}
\newtheorem{lemmadefinition}[theorem]{Lemma and Definition}
\newtheorem{lemma}[theorem]{Lemma} 
\newtheorem{remark}[theorem]{Remark}
\newtheorem{example}[theorem]{Example}
\numberwithin{equation}{section}
\numberwithin{figure}{section}
\newcommand{\bei}{\begin{itemize}}
\newcommand{\eei}{\end{itemize}}
\newcommand{\be}{\begin{equation}}
\newcommand{\ee}{\end{equation}} 
\newcommand{\bel}{\be \label}
\newcommand{\bse}{\begin{subequations}}
\newcommand{\ese}{\end{subequations}}
\newcommand{\mathtikz}[2][]{\begin{tikzpicture}[baseline=\the\dimexpr-\fontdimen22\textfont2\relax,#1]#2\end{tikzpicture}}
\DeclareMathOperator{\diag}{diag}
\DeclareMathOperator{\Cl}{\bf Cl}
\DeclareMathOperator{\Tr}{Tr}
\DeclareMathOperator{\sgn}{sgn}
\newcommand{\abs}[1]{\lvert#1\rvert}
\newcommand{\RR}{\mathbb{R}}
\newcommand{\ZZ}{\mathbb{Z}}
\newcommand{\Tbb}{\mathbb{T}}
\newcommand{\del}{\partial}
\newcommand{\Hcal}{\mathcal{H}}
\newcommand{\Lcal}{\mathcal{L}}
\newcommand{\Mcal}{\mathcal{M}}
\newcommand{\Ncal}{\mathcal{N}}
\newcommand{\Pcal}{\mathcal{P}}
\newcommand{\Qcal}{\mathcal{Q}}
\newcommand{\Sbf}{\mathbf S}
\newcommand{\Siso}{\mathbf S^{\textnormal{iso}}}
\newcommand{\Sani}{\mathbf S^{\textnormal{ani}}}
\newcommand{\rmoi}{r_-}
\newcommand{\rpoi}{r_+}
\newcommand{\thetamoi}{\theta_-}

\newcommand{\gmoi}{g_-{}}
\newcommand{\gpoi}{g_+{}}
\newcommand{\Kmoi}{K_-^{}{}}
\newcommand{\Kmoicirc}{\mathring{K}_-}
\newcommand{\Kpoi}{K_+^{}{}}
\newcommand{\Kpoicirc}{\mathring{K}_+}
\newcommand{\gst}{g_*{}}
\newcommand{\Kst}{K_*^{}{}}
\newcommand{\phist}{\phi_*{}}
\newcommand{\Mup}{\Mcal_{\vee}}
\newcommand{\Mplus}{\Mcal_{//}}
\newcommand{\Mminus}{\Mcal_{\backslash\backslash}}
\newcommand{\Mleftfar}{\Mcal_{>}}
\newcommand{\Mrightfar}{\Mcal_{<}}
\newcommand{\Mdown}{\Mcal_{\wedge}}
\newcommand{\Ncalplus}{\mathcal{N}}
\newcommand{\Ncalminus}{\overline{\mathcal{N}}}
\newcommand{\Lscr}{\mathscr{L}}
\newcommand{\Eplus}{E}
\newcommand{\Fplus}{F}
\newcommand{\Iplus}{I}
\newcommand{\vplus}{{v\hskip.02cm}}
\renewcommand{\uplus}{{u\hskip.02cm}}
\newcommand{\lplus}{l}
\newcommand{\mplus}{m}
\renewcommand{\nplus}{n}
\newcommand{\rplus}{r}
\newcommand{\xplus}{x}
\newcommand{\Phiplus}{\Phi}
\newcommand{\Psiplus}{\Psi}
\newcommand{\epsplus}{\epsilon}
\newcommand{\phiplus}{\varphi}
\newcommand{\psiplus}{\psi}
\newcommand{\Eminus}{\overline{E}}
\newcommand{\Fminus}{\overline{F}}
\newcommand{\Iminus}{\overline{I}}
\newcommand{\mminus}{\overline{m}}
\newcommand{\nminus}{\overline{n}}
\newcommand{\rminus}{{\overline{r}}}
\newcommand{\uminus}{{\overline u\hskip.02cm}}
\newcommand{\vminus}{{\overline v\hskip.02cm}}
\newcommand{\Phiminus}{\overline{\Phi}}
\newcommand{\Psiminus}{\overline{\Psi}{}}
\newcommand{\epsminus}{\overline \epsilon}
\newcommand{\phiminus}{{\overline \varphi}}
\newcommand{\psiminus}{{\overline \psi}}
\newcommand{\psiplusG}{\psiplus_{\scriptscriptstyle\textnormal{G}}}
\newcommand{\psiminusG}{\psiminus{}_{\scriptscriptstyle\textnormal{G}}}
\newcommand{\psinot}{\psi_{\circ}}
\newcommand{\Abf}{\mathbf{A}}
\newcommand{\SscrADMfromF}{{\mathscr{S}}_{\!\!\scriptscriptstyle\textnormal{ADM}}^{\ \ \ \textnormal{F}}}
\newcommand{\SscrFfromADM}{{\mathscr{S}}_{\!\textnormal{F}}^{\ \scriptscriptstyle\textnormal{ADM}}}
\makeatletter
\newcommand{\preprint}[1]{%
  \newcommand{\ps@preprintnumber}{\ps@plain
    \renewcommand{\@oddhead}{\hfill\begin{picture}(0,0)\put(0,-40){\llap{#1}}\end{picture}}%
    \let\@evenhead\@oddhead}%
  \thispagestyle{preprintnumber}%
}
\makeatother

\begin{document}

\title{Cyclic spacetimes through singularity scattering maps. 
\\
Plane-symmetric gravitational collisions
\footnotetext{$^1$ 
Philippe Meyer Institute, Physics Department, \'Ecole Normale Sup\'erieure, PSL Research University, 24 rue Lhomond, F-75231 Paris Cedex 05, France.  Email: {\tt bruno@le-floch.fr}.
\\
$^2$ Laboratoire Jacques-Louis Lions \& Centre National de la Recherche Scientifique, 
Sorbonne Universit\'e,  
4 Place Jussieu, 75252 Paris Cedex, France. Email: {\tt contact@philippelefloch.org}. 
\\
$^3$ CERN, Theory Department, CH-1211 Geneva 23, Switzerland. Email: {\tt Gabriele.Veneziano@cern.ch.} 
\\
$^4$ Coll\`ege de France, 11 Place M. Berthelot, 75005 Paris, France. 
}} 
\author{Bruno Le Floch$^1$, Philippe G. LeFloch$^2$, and Gabriele Veneziano$^{3,4}$  
}
\date{}

\maketitle

\begin{abstract}  
We study the plane-symmetric collision of two gravitational waves and describe the global spacetime geometry generated by this collision. To this end, we formulate the characteristic initial value problem for the Einstein equations, when Goursat data describing the incoming waves are prescribed on two null hypersurfaces. We then construct a global solution representing a cyclic spacetime based on junction conditions associated with a prescribed singularity scattering map (a notion recently introduced by the authors). 
From a mathematical analysis standpoint, this amounts to a detailed analysis of the Goursat and Fuchsian initial value problems associated with singular hyperbolic equations, when junction conditions at interfaces are prescribed. 
In our construction, we introduce a partition into monotonicity diamonds (that is, suitable spacetime domains) and we 
construct the solution by concatenating domains across interfaces of timelike, null, or spacelike type. 
\end{abstract}

\preprint{CERN-TH-2021-094}
 
 {\small 
\setcounter{tocdepth}{2} 
\tableofcontents 

}

\

\section{Introduction}

\paragraph{Plane gravitational waves.}

We study the gravitational collision problem which we solve {\sl beyond singularities} in the class of plane-symmetric spacetimes, in a sense that we defined recently in \cite{LLV-2,LLV-1a}. Recall that the global geometry of {\sl plane waves} (without collisions)  was analyzed in Penrose \cite{Penrose0}, who emphasized their relevance as an idealization of more general physical phenomena. He stressed that a challenge arises with the study of plane waves that is not encountered with, for instance, asymptotically flat spacetimes. Namely, in a plane wave spacetime, no spacelike hypersurface exists on which one could pose the initial value problem globally, so that such a plane wave cannot be embedded in a globally hyperbolic spacetime.  The hypersurface necessarily contains trapped surfaces and, in fact, {\sl geometric singularities,} as we explain later in this paper. 
Following the pioneering work by Penrose \cite{Penrose0}, Flores and S\'anchez \cite{FloresSanchez:2003,FloresSanchez:2008} extensively studied this class of plane-wave spacetimes 
and established several properties of their causal domains and conjugate points. 

Despite the presence of singularities, we prove here that in the generalized setup of cyclic spacetimes proposed in \cite{LLV-2,LLV-1a} we can still prescribe initial data on such a hypersurface and solve globally.
Note that while trapped surface are unavoidable only in plane symmetry, they are also a typical phenomenon arising with the Einstein equations even when no symmetry restriction is made~\cite{Christodoulou-book}.
Before beginning the description of our results, let us emphasize that this paper can be read without a priori knowledge of the Einstein equations since the required notions will be defined in this text when needed. 
Only for additional material, the reader will be able to refer to \cite{LLV-1a}, while all 
technical contributions of this paper involves singular solutions to partial differential equations of hyperbolic type 
and, therefore, 
  are accessible to a reader interested in learning the mathematical analysis techniques proposed here. 
 
\paragraph{The collision problem as a characteristic initial value problem.} 

We wish to describe two plane waves propagating from past infinity in opposite directions, initially separated by a flat spacetime region, that come together and interact.
We formulate this collision problem as seeking the Cauchy development of an initial data set consisting of two null hypersurfaces with a two-dimensional intersection and endowed with suitable data sets. 
We are thus given two three-dimensional manifolds $\Ncalplus_0$ and $\Ncalminus_0$ endowed with degenerate Riemannian metrics of signature $(0, +, +)$ and identified along their boundary
(a two-dimensional Riemannian manifold denoted by~$\Pcal_0$), as well as data of incoming gravitational radiation on these hypersurfaces (as specified below).
Our plane-symmetry assumption allows symmetry orbits such as~$\Pcal_0$ to be two-planes~$\RR^2$ or two-tori~$\Tbb^2$; we choose for definiteness the latter, with a rectangular torus\footnote{Our results can be stated in both cases, and choosing torus orbits of symmetry allows one for the energy content of the spacetime to be finite, rather than only its local energy density.  Likewise, our choice of a torus lets us talk about the area of symmetry orbits, rather than the local area element.}. 

We construct geometrically, in Section~\ref{section----6} below,
a set of null coordinates $(\uplus,\uminus,x,y)$ adapted to our problem,
where $x,y\in[0,1)$ are periodic coordinates on the $\Tbb^2$ symmetry orbits, while $\uplus,\uminus\in\RR$ are null coordinates on the quotient of spacetime by the $\Tbb^2$~action.
We express in these coordinates the metric and Einstein equations in terms of geometric quantities.
The metric is
\[
g = - 2 e^{2\omega} \, d\uplus d\uminus + e^{2a+2\psi} dx^2+e^{2a-2\psi}dy^2
\]
so that our problem has four scalar variables, functions of $\uplus,\uminus$: the conformal factor~$\omega$, the area~$e^{2a}$ and modular parameter $ie^{2\psi}$ of torus symmetry orbits, and the scalar field~$\phi$.
We then pose the initial value problem as follows. 
\bei 

\item Initial data can be prescribed at past infinity within the regions we have denoted by $\Mleftfar$ and $\Mrightfar$ in Figure~\ref{Fig:1.1}, describing plane gravitational waves that collide in the region $\Mup=\{\uplus>0,\uminus>0\}$.
Equivalently, we prescribe data along the null hypersurfaces $\Ncalplus_0=\{\uplus>0,\uminus=0\}$ and $\Ncalminus_0=\{\uplus=0,\uminus>0\}$ and we solve globally in the region~$\Mup$.
Without loss of generality we set $\omega=0$ along these hypersurfaces by redefinitions of $\uplus$ and of~$\uminus$; this also provides us with a geometric definition of preferred null coordinates $\uplus,\uminus$, as affine parameters of null geodesics (in the two initial waves) orthogonal to the symmetry orbits.

\item
The initial data for the area function~$e^{2a}$, on the other hand, can be recovered from the data for $\psi,\phi$ 
by using the Einstein equations.
Altogether, the incoming radiation data boils down to the values of $\psi,\phi$ along the two null hypersurfaces, and solving the Einstein-scalar field equations consists of finding the unknown functions $\omega,a,\psi,\phi$ in the region~$\Mup$.

\item It turns out that the area function $A=e^{2a}$ (in its initial data) must vanish at least once on~$\Ncalplus_0$ and once on~$\Ncalminus_0$: this is the infinite-focussing phenomenon found by Penrose to be unavoidable in plane-symmetric gravitational waves.
One of the Einstein evolution equations is a wave equation $A_{\uplus\uminus}=0$.
Starting from initial data where the area function~$A$ is non-negative, this yields a solution without a definite sign, and the area $|A|=e^{2a}$ of symmetry orbits generically vanishes along a collection of singularity hypersurfaces as depicted in Figure~\ref{Fig:1.1}.

\eei 

Moreover, we emphasize that describing the spacetime geometry after the first singularity is necessary if one wants the physical theory to be complete. Indeed,
physically, one may ask  what will be the geometry experienced by observers that enter the system (defined by the two incoming colliding waves) at much later times than those characterizing the formation of the singularity. They should be determined by the incoming energy flux at later times,
but  they also depend on the data specified on the ``other side'' of the above mentioned singular boundary.

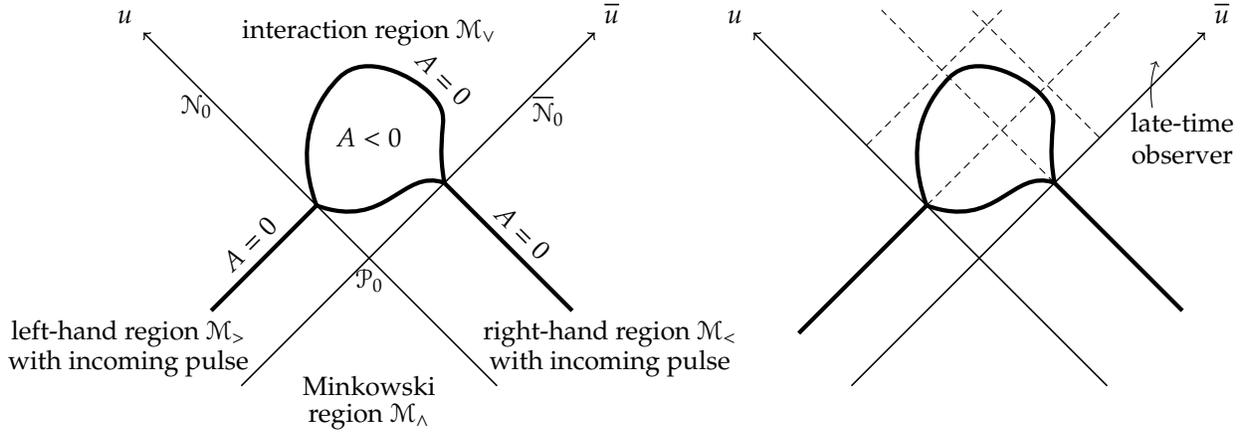
\begin{figure}
  \centering
  \centerline{%
  \begin{tikzpicture}
    \draw[semithick,->] (1.7,-1.7) -- (-3,3) node [pos=.85,below] {$\Ncalplus_0$} node [above left] {$\uplus$};
    \draw[semithick,->] (-1.7,-1.7) -- (3,3) node [pos=.85,below] {$\,\,\,\Ncalminus_0$} node [above right] {$\uminus$};
    \draw[ultra thick] (-.7,.7)
    .. controls +(-20:1) and +(160:.6) .. (1,1)
    .. controls +(100:.8) and +(-45:.5) .. (.8,2.2)
    .. controls +(135:.3) and +(45:.5) .. (-.4,2.4) node [pos=0,above,sloped] {$A=0$}
    .. controls +(-135:.3) and +(110:1) .. (-.7,.7);
    \node at (0,1.6) {$A<0$};
    \draw[ultra thick] (-2.1,-.7) -- (-.7,.7) node [midway,above,sloped] {$A=0$};
    \draw[ultra thick] (2.7,-.7) -- (1,1) node [midway,above,sloped] {$A=0$};
    \node at (0,-.3) {$\Pcal_0$};
    \node at (0,-1.7) {Minkowski};
    \node at (0,-2.1) {region $\Mdown$};
    \node at (-3.2,-1) {left-hand region $\Mleftfar$};
    \node at (-3.2,-1.4) {with incoming pulse};
    \node at (3.2,-1) {right-hand region $\Mrightfar$};
    \node at (3.2,-1.4) {with incoming pulse};
    \node at (0,3) {interaction region $\Mup$};
  \end{tikzpicture}\hspace{-1em}%
  \begin{tikzpicture}
    \draw[semithick,->] (1.7,-1.7) -- (-3,3) node [above left] {$\uplus$};
    \draw[semithick,->] (-1.7,-1.7) -- (3,3) node [above right] {$\uminus$};
    \draw[ultra thick] (-.7,.7)
    .. controls +(-20:1) and +(160:.6) .. (1,1)
    .. controls +(100:.8) and +(-45:.5) .. (.8,2.2)
    .. controls +(135:.3) and +(45:.5) .. (-.4,2.4)
    .. controls +(-135:.3) and +(110:1) .. (-.7,.7);
    \draw[ultra thick] (-2.4,-1) -- (-.7,.7);
    \draw[ultra thick] (2.7,-.7) -- (1,1);
    \draw[densely dashed] (-.7,.7) -- +(2.6,2.6);
    \draw[densely dashed] (-1.5,1.5) -- +(1.8,1.8);
    \draw[densely dashed] (1,1) -- +(-2.3,2.3);
    \draw[densely dashed] (1.6,1.6) -- +(-1.7,1.7);
    \node at (0,-2.1) {\phantom{region $\Mdown$}};
    \draw [->] (2.3,2) arc (190:160:1.2);
    \node[below] at (2.7,2.1) {\begin{tabular}{c}late-time\\observer\end{tabular}};
  \end{tikzpicture}%
  }
  \caption{\label{Fig:1.1}{\bf Left:} \sl
  Example of colliding plane-symmetric gravitational waves.  Each point of the diagram is a plane-symmetry orbit, and 45$^{\circ}$ lines are null rays.  The data can be prescribed either as incoming gravitational pulses in the regions $\Mleftfar$ and~$\Mrightfar$ or simply on the wave-fronts $\Ncalplus_0$ and~$\Ncalminus_0$.  We depict by thick lines the singularity hypersurfaces along which the area function~$A$ vanishes; in the region $\Mup=\{\uplus>0,\uminus>0\}$ these are generically also curvature singularities.  In general this locus $A=0$ can partition spacetime into any number of spacetime domains, as depicted in Figure~\ref{fig:1.2} in the main text.
  {\bf Right:} Decomposition of the same spacetime (along the dashed lines) into monotonicity diamonds in which $A$ depends monotonically on $\uplus$ and on~$\uminus$.
  In a complete theory, an observer (depicted by the curved arrow) entering the interaction region~$\Mup$ at late times must see some geometry that is fully determined from the initial data; we achieve this through our use of singularity scattering maps.}
\end{figure}

\paragraph{A global construction of cyclic spacetimes.}

The curvature of the spacetime generated by such a collision blows up (for generic data) in finite time along future timelike directions, and understanding the global structure of such spacetimes is a challenging problem.
We establish that gravitational collisions generate a cyclic spacetime in which finitely (respectively, infinitely) many contracting/expanding phases successively take place if the incoming pulses have finite duration or sufficient decay (resp.\@, are not decaying).
Such a spacetime consists of a collection of spacetime domains and, more precisely, (anisotropic, cosmological) {\sl monotonicity diamonds} as we call them,
containing at most one spacelike singularity hypersurface connecting a Big Crunch region to a Big Bang region, or else one timelike hypersurface of singularity forming an interface between two regions. 
We refer to Section~\ref{section----6} for terminology on the plane-symmetric collision, and we now state our main result
in a simplified form, while postponing the full statement to Theorem~\ref{theorem-plane} below. 

\begin{theorem}[Cyclic spacetimes generated by the collision of plane waves] 
\label{theo:first}
Fix a causal momentum-preserving ultralocal scattering map\footnote{The relevant notions were introduced first in \cite{LLV-1a} and are defined in Section~\ref{section---8}, below.}. 
Given two generic plane-symmetric gravitational waves propagating in opposite directions and colliding along a two-plane,
the characteristic initial value problem associated with the initial data induced on the two wave fronts admits a global Cauchy development $(\Mcal, g, \phi)$ in a class of spacetimes with singularity hypersurfaces.
This spacetime consists of a concatenation of anisotropic cosmological spacetime domains separated by singularity hypersurfaces of spacelike or timelike type.
The Einstein field equations are satisfied away from the singularity hypersurfaces,
and the junction conditions associated 
with the given scattering map hold across all singularity hypersurfaces.
The curvature of $(\Mcal, g)$ generically blows up as one approaches the singularity hypersurfaces. 
\end{theorem}

Importantly, in our construction we need to cross singularity hypersurfaces, for which we rely on our junction conditions.
This theorem provides us with a large class of cyclic spacetimes governed by the Einstein equations, for which we can arbitrarily prescribe radiation data at past infinity. In such a Universe, starting from a flat background in which initially non-interacting plane-symmetric waves propagate, successive periods of contraction and expansion take place in an oscillating pattern.  There are finitely (respectively infinitely) many cycles when the incoming waves have finite (resp.\@ infinite) duration. 

In our context, we could allow the metric to have only weak regularity with locally finite energy, so that away from the singularity hypersurfaces Einstein's field equations must be understood in the distributional sense. Namely, the collision problem can be naturally posed within a class of incoming radiation data (and spacetimes) that need not be regular. The relevant mathematical techniques to construct such weakly regular spacetimes were introduced in \cite{PLFMardare}--\cite{PLFStewart} and are developed in the series of papers \cite{LeFlochLeFloch-1}--\cite{LeFlochLeFloch-4}.

\paragraph{Causality of singularity scattering maps.}

The junction conditions were introduced in~\cite{LLV-1a} as relations between suitable rescalings of the scalar field~$\phi$, spatial metric~$g$ and extrinsic curvature~$K$ of a Gaussian foliation on both sides of the singularity.
Denoting by~$s$ the proper time (or distance, for a timelike singularity) along geodesics normal to the singularity hypersurface located at $s=0$, the singularity data is
\be\label{singularity-data-set}
\aligned
(\gpoi, \Kpoi, \phi_{0+}, \phi_{1+})
& \coloneqq \lim_{\substack{s \to 0\\s > 0}} \Big(|s|^{2 s K} g,\; -s K,\; s \del_s \phi,\; \phi - s \log |s|  \del_s \phi \Big)(s), 
\\
(\gmoi, \Kmoi, \phi_{0-}, \phi_{1-}) 
& \coloneqq \lim_{\substack{s \to 0\\s < 0}} \Big(|s|^{2 s K} g,\; -s K,\; s \del_s \phi,\; \phi - s \log |s|  \del_s \phi \Big)(s).
\endaligned
\ee
Junction conditions of the form $(\gpoi, \Kpoi, \phi_{0+}, \phi_{1+}) = \Sbf(\gmoi, \Kmoi, \phi_{0-}, \phi_{1-})$ at every point on the singularity, for some map~$\Sbf$ dubbed a singularity scattering map.
An important result in \cite{LLV-1a} is a classification of such pointwise (or ultralocal) relations between singularity data sets that are compatible with Einstein constraint equations.
Our global existence theory in plane symmetry treats both sides of timelike singularity hypersurfaces on an equal footing, so that either singularity data 
 can be written as a function of the other.  This requires the scattering map, denoted by~$\Sbf$ and describing the junction,
 to be {\sl invertible.}

For definiteness, in this paper we focus on the class of momentum-preserving maps ($K_+=K_-$), classified in~\cite{LLV-1a}.
They are best expressed in terms of a parametrization of the rescaled intrinsic curvatures~$K_{\pm}$, namely ($\delta$~is the identity matrix)
\be
K_{\pm} = \frac{1}{3}\delta+\frac{2}{3}r_{\pm}\cos\Theta_{\pm} , \qquad
r_{\pm} = r(\phi_{0\pm}) = \sqrt{1-12\pi\phi_{0\pm}^2} , \qquad
\Theta_{\pm} = \diag\biggl(\theta_{\pm}, \theta_{\pm} + \frac{2\pi}{3} , \theta_{\pm} + \frac{4\pi}{3}\biggr) ,
\ee
in a suitable basis orthonormal with respect to the rescaled metric~$g_{\pm}$.
A momentum-preserving scattering map then gives the junction condition
\bel{junction-for-planes}
\aligned
\gpoi & = c^2 \exp\Big(
16\pi\xi\cos\Theta_-
- 16\pi \del_{\thetamoi} \bigl( \xi + \frac{\phi_{0-}}{r(\phi_{0-})} f \bigr) \sin\Theta_-
\Big) \, \gmoi , \qquad
&&\Kpoi = \Kmoi , \qquad
\\
\phi_{0+} & = a \phi_{0-} , \qquad
&&\phi_{1+} = a f(\thetamoi,\phi_{0-})+\phi_{1-} ,
\endaligned
\ee
which depends on an unimportant sign $a=\pm$ and unimportant constant $c>0$, as well as an essentially arbitrary function $f=f(\thetamoi,\phi_{0-})$, from which we construct the auxiliary function
$\xi = \xi(\thetamoi,\phi_{0-}) \coloneqq - \int_{-1/\sqrt{12\pi}}^{\phi_{0-}} r(y)^{-1}y\del_y f(\thetamoi,y) dy$.
The expression of~$g_+$ involves a matrix exponential that rescales~$g_-$ differently along different eigenvectors of~$K_-$.
The constant~$c$ can be normalized away by changing units on one side of the singularity.

In Section~\ref{section---8}, we shall translate these junction conditions from the Gaussian foliation used in~\cite{LLV-1a} to the double null foliation used in plane symmetry.
When solving the evolution problem in the presence of a timelike singularity hypersurface, we learn that causality imposes a further invertibility condition on~$f$, stated as Definition~\ref{def:causal}.
This condition precludes for instance taking $f=0$, but allows taking $f(\thetamoi,\phi_{0-})=b\phi_{0-}$ for any constant $b\neq 0$.

\paragraph{Relevance of the global collision problem.} 

With the observation of gravitational waves under way and given the extensive development of numerical relativity, the collision problem provides
a simplified yet physically interesting set-up, even within the class of plane-symmetric waves. 
We make here global predictions about the propagation and interaction of gravitational waves,
and investigate the effect of an arbitrarily large number of colliding gravitational waves. This picture 
may be relevant in cosmology for the study of the large-scale structure of the Universe. Insights obtained in the present paper are also relevant in order to extend our conclusions beyond the case of plane-symmetric collisions.  

The colliding gravitational wave problem has a long history in the physics literature. Examples were first constructed by Penrose~\cite{Penrose0,Penrose} and Khan and Penrose \cite{KhanPenrose}. 
In particular, recall that LeFloch and Stewart \cite{PLFStewart} established the nonlinear stability of the Khan-Penrose solution within the class of weakly regular spacetimes. The present paper will establish that the spacetimes in \cite{PLFStewart} make sense beyond the ``first'' singularity hypersurface and we will describe how the extension should be made.  In particular, our construction applies to the Khan-Penrose spacetime and extends it to a cyclic spacetime.

The plane collision problem provides a simplified description of the actual collision of two high energy beams. To further investigate pre-Big Bang scenarios, and following earlier work by Eardley and Giddings~\cite{EardleyGiddings}, Kohlprath and Veneziano~\cite{KohlprathVeneziano} analyzed the high-energy collision of two beams of massless particles, represented by two axi-symmetric colliding waves, and established that marginally trapped surfaces arise after the collision.
In the past ten years, there has been a renewed interest in the gravitational collision problem, motivated by high-energy physics and in connection with the pre-Big Bang scenario in string cosmology, proposed in~\cite{BV,BDV,FKV,GasperiniVeneziano2,VenezianoSFD}. 

\paragraph{Outline of this paper.}

The resolution of the gravitational collision problem is tackled as follows.
The interaction of plane gravitational waves unavoidably produces singularities, and one must face the question of continuing the spacetime beyond singularities in order to determine the global geometry generated by the collision.
Our definition and study of singularity scattering maps in \cite{LLV-1a}
is essential for this purpose.
Conversely, applying our general junction conditions to the plane-symmetric setting gives a more concrete handle on our scattering maps.
We begin with a few observations.
\bei

\item {\bf The first collision region.}
Relying on (null) coordinates that are suitably adapted to the plane symmetry, we will state and analyze the essential part of the Einstein-matter field equations and, in fact, exhibit an explicit formula (the so-called Abel representation formula) for the essential metric and matter fields, denoted by 
$\psi$ and $\phi$ below. The remaining metric coefficients, that is the area $a$ and the conformal factor $\omega$,
 are obtained by solving a differential equation along the null directions.  This closed-form formula allows us to establish a well-posedness result for the characteristic initial value problem within a ``first'' region of the interaction.

\item{\bf Behavior near the singularity hypersurface.}
Next, by performing a suitable expansion on the ``first'' singularity hypersurface, we will analyze the behavior of the solution and, in particular, observe that the spacetime curvature generically blows up as one approaches the future boundary of this first collision region.

\item {\bf Crossing the first singularity hypersurface.} 
We will then cross the singularity and start the construction of the global spacetime structure. 
Recall that the standard junction condition introduced by Israel does not apply to our problem, and we must rely here on the 
general junction conditions~\eqref{junction-for-planes},
based on singularity scattering maps. 
A special case could be to impose suitable ``continuity conditions'' across singularity hypersurfaces, but we prefer to keep 
our setting sufficiently general in order to accommodate a variety of physical models.

\item {\bf The global cyclic spacetime geometry.}
Based on the chosen singularity scattering map we can determine the global spacetime geometry, and the constructed spacetime can be interpreted as a physically meaningful cyclic Universe.

\eei

In Section~\ref{section----6} we begin with the formulation of the global collision problem. In Section~\ref{section---7} the geometry of such spacetimes is studied, while the construction is initiated in Section~\ref{section---8} (field equations and junction conditions) and fully described in Section~\ref{section---9} (actual construction based on a prescribed scattering map). 

\


\section{Formulation of the global plane collision problem}
\label{section----6}

\subsection{Geometry of the colliding spacetime} 
\label{section--31}

\paragraph{The double foliation by null hypersurfaces.}

From now on, we assume that the collision involves plane-symmetric gravitational waves that propagate in opposite directions and we solve the problem globally.
The interaction region denoted by $\Mup$ is defined as the future of the two-plane $\Pcal_0$ of intersection between the two null hypersurfaces.
We are going to construct a foliation of this spacetime domain
\[
\Mup = \bigcup_{\uplus >0} \Ncalminus_{\uplus}
= \bigcup_{\uminus >0} \Ncalplus_{\uminus}
\]
by two families of null hypersurfaces
$
\Ncalminus_\uplus = \bigl\{ (\uplus, \uminus) \bigm| \uminus >0 \bigr\} 
$ and $
\Ncalplus_\uminus = \bigl\{ (\uplus, \uminus) \bigm| \uplus >0 \bigr\},
$
along which $\uplus$ and $\uminus$ are constants, respectively.  In particular, the initial hypersurfaces on which incoming wave data are prescribed are $\Ncalminus_0 = \big\{ \uplus = 0; \, \, \uminus > 0 \big\}$ and $\Ncalplus_0 = \big\{\uplus > 0; \, \, \uminus = 0 \big\}$, and they intersect along $\Pcal_0 = \big\{\uplus = \uminus = 0 \big\}$.
The fact that null coordinates $(\uplus,\uminus)$ can be globally defined despite the presence of singularities
relies on the fact that light rays meaningfully traverse singularities, as we will observe.

We introduce coordinates $(x,y)$ in the $2$-torus $\Tbb^2$ (or $\RR^2$) on the orbits of plane symmetry and supplement them with two null coordinates $(\uplus, \uminus)$ geometrically defined as follows.

\bei 

\item First of all, it is convenient to introduce the quotient manifold $\Qcal\coloneqq \Mcal/\Tbb^2$ of the spacetime by the two-dimensional translation group of the plane.
  Then, at the intersection $\Pcal_0$ of the two incoming fronts, we choose two null vectors $l_{\Pcal_0},n_{\Pcal_0}$ in the tangent spaces to $\Ncalplus_0, \Ncalminus_0$, respectively, and orthogonal to the group orbits, normalized so that 
\[
g(l_{\Pcal_0},n_{\Pcal_0}) = -1.
\]
Each timelike surface orthogonal to the group orbits represents the quotient manifold $\Qcal$. 

\item We then extend the vectors $l_{\Pcal_0}, n_{\Pcal_0}$ toward the future as geodesic fields $\lplus, \nminus$, whose integral curves in $\Qcal$ are $\Ncalplus_0/\Tbb^2, \Ncalminus_0/\Tbb^2$, respectively. We define $\uplus$ and $\uminus$ on these curves to be the affine parameters of $\lplus$ and $\nminus$, normalized so that
$
\uplus = \uminus = 0$ at $\Pcal_0$. 
By definition, the quotient $\Ncalplus_0/\Tbb^2$ (for instance) is identified with a geodesic $\xplus^\alpha= \xplus^\alpha(\uplus)$ satisfying the equation $\nabla_{\dot \xplus} \dot \xplus = 0$.

\item Next, we define the hypersurfaces $\Ncalplus_\uminus$, $\Ncalminus_\uplus$ for every $\uplus,\uminus$ by requiring that $\Ncalplus_\uminus/\Tbb^2$ and $\Ncalminus_\uplus/\Tbb^2$ are right-moving and left-moving null curves in the quotient manifold~$\Qcal$, respectively.  We assign coordinates $(\uplus,\uminus)$ to the intersection point of $\Ncalplus_\uminus/\Tbb^2$ and $\Ncalminus_\uplus/\Tbb^2$ when it exists and is unique.  We show later that in a maximal development of the initial data all values $(\uplus,\uminus)$ correspond to a point in~$\Qcal$ (namely $(\Ncalplus_\uminus\cap\Ncalminus_\uplus)/\Tbb^2$ is a point) and conversely that this coordinate patch covers all of~$\Qcal$ and, therefore, gives coordinates on the whole spacetime~$\Mcal$.
In particular, we show that $\Ncalplus_\uminus$ and $\Ncalminus_\uplus$ extend naturally beyond singularities.
Finally, we introduce the following vector fields in~$\Mcal$: 
\[
l \coloneqq {\del \over \del \uplus}, \qquad 
n \coloneqq {\del \over \del \uminus}.
\]
\eei
This completes the geometric construction of the coordinates $(\uplus, \uminus, x, y)$ and the associated frame $\big(l,n, {\del \over \del x}, {\del \over \del y}\big)$. We emphasize the freedom to scale the original vectors $l_{\Pcal_0} \mapsto \lambda l_{\Pcal_0}$ and $n_{\Pcal_0} \mapsto n_{\Pcal_0}/\lambda$ by an arbitrary factor $\lambda>0$, hence to scale the coordinates $\uplus\mapsto\lambda^{-1}\uplus$ and $\uminus\mapsto\lambda\uminus$ globally.  This is further discussed in Remark~\ref{rem:invariance}.

\begin{figure}[t]\centering
\begin{tikzpicture}
\fill[opacity=.5,color=black!20!white] (0,0) -- (-3.4,3.4)
.. controls +(30:2) and +(150:2) .. (3.4,3.4) -- cycle;
\filldraw[opacity=.5,color=black!20!white] (-4,-4) -- (0,0) -- (-1.2,1.2) -- (-5.2,-2.8) -- cycle;
\filldraw[opacity=.5,color=black!20!white] (4,-4) -- (0,0) -- (1.2,1.2) -- (5.2,-2.8) -- cycle;
\node[below] at (0.,-0.2) {$\Pcal_{0}$};
\node[below] at (0,3.25) {$\Mup$};
\node[below] at (0,2.75) {(interaction region)};
\node[below] at (-2.85,2.75) {$\Ncalplus_0$};
\node[below] at (2.85,2.75) {$\Ncalminus_0$};
\node[below] at (-2.6,-2.75) {$\Ncalminus^{(0)}_0$};
\node[below] at (2.55,-2.75) {$\Ncalplus^{(0)}_0$};
\draw[->] (-4,-4) -- (3.5,3.5);
\node[right] at (3.5,3.5) {$\uminus$};
\draw[->] (4,-4) -- (-3.5,3.5);
\node[left] at (-3.5,3.5) {$\uplus$};
\draw[loosely dashed] (-1,1) -- (-5,-3);
\draw[loosely dashed] (-0.8,0.8) -- (- 4.8,-3.2);
\draw[loosely dashed] (-0.6,0.6) -- (- 4.6,-3.4);
\draw[loosely dashed] (-0.4,0.4) -- (- 4.4,-3.6);
\draw[loosely dashed] (-0.2,0.2) -- (- 4.2,-3.8);
\node at (-5.9,-3.4) {$\uplus \in  [\uplus_\star^-, \uplus_\star^+]$};
\node at (-5.9,-3.9) {$\Mplus$ (incoming pulse)};
\node at (-5,0) {$\Mleftfar$ (left-hand region)};
\draw[loosely dashed] (1,1) -- (5,-3);
\draw[loosely dashed] (0.8,0.8) -- (4.8,-3.2);
\draw[loosely dashed] (0.6,0.6) -- (4.6,-3.4);
\draw[loosely dashed] (0.4,0.4) -- (4.4,-3.6);
\draw[loosely dashed] (0.2,0.2) -- (4.2,-3.8);
\node at (5.9,-3.4) {$\uminus \in  [\uminus_\star^-, \uminus_\star^+]$};
\node at (5.9,-3.9) {$\Mminus$ (incoming pulse)};
\node at (5,0) {$\Mrightfar$ (right-hand region)};
\node[below] at (0,-2.75) {$\Mdown$};
\node[below] at (0,-3.25) {(Minkowski region)};
\end{tikzpicture}
\caption{\label{fig:fourregions}{\bf Collision of two plane waves of finite duration.}  Spacetime is exactly flat away from the shaded regions.}
\end{figure}
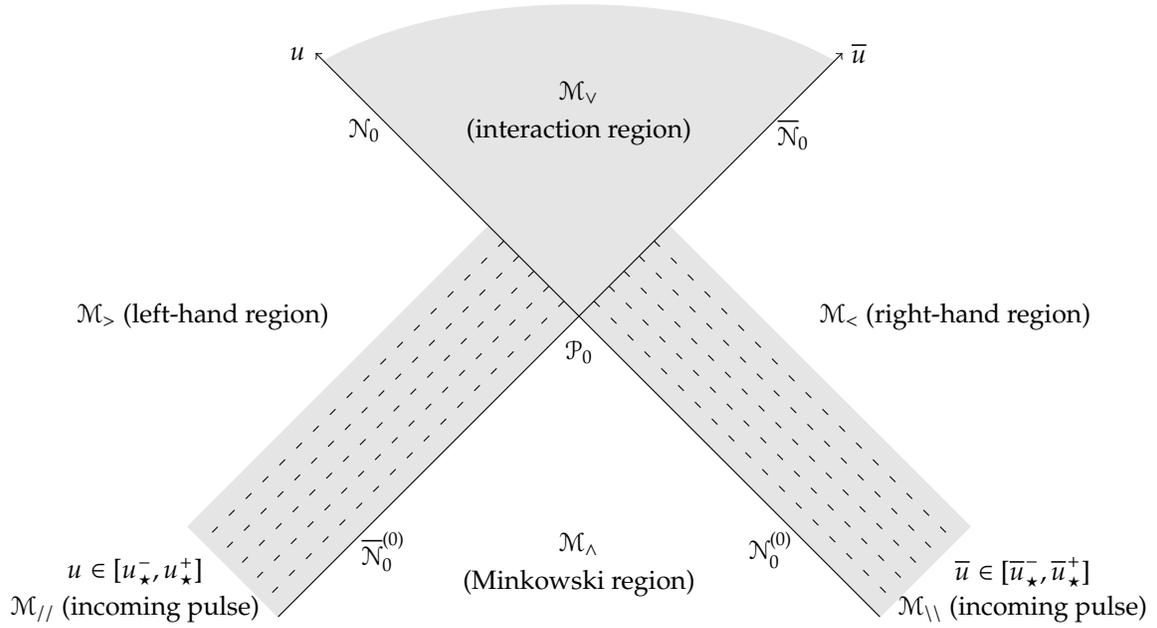

\paragraph{Decomposition of the metric.} 

Three fundamental geometric notions are now introduced. 

\bei 

\item {\bf Conformal quotient factor.}
In the coordinates $(\uplus, \uminus)$ that we just introduced, the quotient metric $g_\Qcal$ induced on $\Qcal$ can be written in the form  
\[
g_\Qcal  = -2 \Omega\, d\uplus d\uminus,
\]
where the conformal factor $\Omega = e^{2\omega}$ depends upon the null variables $(\uplus,\uminus)$, only. In our construction, 
generically, $\Omega$ tends to zero or blows up when singularity hypersurfaces in $\Mcal$ are approached.

\item {\bf Area function.} 
The (signed) area of the surface of $\Tbb^2$ symmetry denoted here by $A = \pm e^{2a}$ (when $A \neq 0$) 
also depends upon $(\uplus,\uminus)$ only and may vanish at geometric singularities. 
As we will show later, the evolution of the signed area function in each null direction is given by the {\bf Raychaudhuri equation} (cf.~Section~\ref{sec:81}), that is, 
\bel{equa-Rayc}
a_{\uplus\uplus}  + (a_\uplus)^2 - 2 \omega_\uplus \, a_\uplus = - \Eplus, 
\qquad
a_{\uminus\uminus} + (a_\uminus)^2 - 2 \omega_\uminus \, a_\uminus = -  \Eminus,  
\ee
in which $\Eplus, \Eminus \geq 0$ denote the (gravitational and matter) {\bf energy fluxes} in the null directions $\uplus,\uminus$, respectively. The expression of $\Eplus,\Eminus$ in terms of the matter field~$\phi$ and the metric coefficient $\psi$ (defined in the next paragraph) is given in \eqref{eq:219} below. (See also  Section~\ref{sec:81}.)  

\item {\bf Modular parameter.} One more geometric coefficient is required in order to fully describe the spacetime geometry, that is, the so-called {\bf modular parameter} $\psi$ associated with the two Killing fields spanning $\Tbb^2$. Interestingly, this coefficient satisfies the same wave equation as the matter field~$\phi$:
\[
\Box_g \psi = 0. 
\]

\eei 

\paragraph{The four spacetime domains.}

We decompose the spacetime in four main regions, that is, 
$
\Mleftfar \cup \Mdown \cup \Mup \cup \Mrightfar$
(cf.~Figure~\ref{fig:fourregions}), 
which together with the corresponding boundaries determines a partition of $\Mcal$. 
\bei
\item {\bf Past domain.}
A region containing past infinity is flat and unperturbed by the gravitational radiation, i.e.
\[
\Mdown   = \big\{ \uplus < 0; \quad \uminus < 0 \big\}   \qquad \mbox{ (flat region),}
\]
and its boundary components are the two null hypersurfaces $\Ncalplus^{(0)}_0=\{\uplus<0,\uminus=0\}$ and $\Ncalminus^{(0)}_0=\{\uplus=0,\uminus<0\}$ defining the incoming wave fronts, which themselves share a two-dimensional boundary $\Pcal_0=\{\uplus=0,\uminus=0\}$.
 
\item {\bf Two incoming wave domains.}
The regions 
\[
\aligned
\Mleftfar & = \big\{ 0 < \uplus; \quad \uminus < 0 \big\}    && \qquad \mbox{(left-hand wave domain),}
\\
\Mrightfar & = \big\{ \uplus < 0; \quad 0 < \uminus \big\} && \qquad \mbox{(right-hand wave domain),}
\endaligned
\]
contain the incoming gravitational waves, whose geometry will be described in Section~\ref{sec:63}. 

\item {\bf Interaction domain.}
Finally, the region after the collision is denoted by
\[
  \Mup = \big\{ \uplus >0; \quad \uminus >0 \big\}  \qquad \mbox{ (interaction domain),}
\]
and its boundary components are the two null hypersurfaces $\Ncalplus_0=\{\uplus>0,\uminus=0\}$ and $\Ncalminus_0=\{\uplus=0,\uminus>0\}$, with common boundary~$\Pcal_0$ along which the waves begin interacting.
We will pose the gravitational collision problem by prescribing data on these two hypersurfaces.
In our construction below, we will need to decompose $\Mup$ further, by introducing singularity hypersurfaces of spacelike, null, or timelike type.

\eei

For the interpretation of the problem as the collision of gravitational pulses
it is convenient to introduce further subsets of $\Mleftfar$ and $\Mrightfar$
defined by
\[
\aligned
\Mplus & = \big\{ 
\uplus_\star^- < \uplus < \uplus_\star^+; \quad \uminus < 0 \big\}   && \qquad \mbox{ (left-incoming pulse),}
\\
    \Mminus & = \big\{ \uplus < 0; \quad 
    \uminus_\star^- < \uminus < \uminus_\star^+ \big\}  && \qquad \mbox{ (right-incoming pulse),}
\endaligned
\]
which are thus supported in some given intervals $(\uplus_\star^-,\uplus_\star^+) \subset (0, +\infty)$ and $(\uminus_\star^-, \uminus_\star^+)  \subset (0, +\infty)$ of finite or infinite duration. 
The parameters $\uplus_\star^+ - \uplus_\star^-$ and $\uminus_\star^+ - \uminus_\star^-$ represent (up to some normalization) the duration of the two incoming pulses.
Two regimes may be of particular interest:
the {\bf impulsive limit} $\uplus_\star^\pm, \uminus_\star^\pm \to 0$ and
the {\bf eternal collision} corresponding to $\uplus_\star^-= \uminus_\star^- = 0$ and $\uplus_\star^+= \uminus_\star^+ = + \infty$.  

\subsection{Definitions: plane waves and gravitational pulses}
\label{sec:63}

\paragraph{Raychaudhuri equation for plane waves.}

To specify the initial data for our collision problem we begin with the description of a single right-moving plane wave. The metric coefficients for a plane wave can be chosen to depend on a single null variable, say $\uplus$ in our notation, so that from the first equation in \eqref{equa-Rayc} satisfied for the function $a=a(\uplus)$, we have  
\bel{eq:bequa}
a_{\uplus\uplus}  + (a_\uplus)^2 
=   
2 \omega_\uplus \, a_\uplus - E,  
\ee
while all remaining Einstein equations are trivially satisfied for solutions depending on $\uplus$ only (as can  be checked from  in \eqref{equa-EinEq} below). 

We now regard \eqref{eq:bequa} as an {\sl equation for the metric coefficient $a$} and we solve it as follows. 
\bei 

\item We note that our geometric construction of the coordinates sets $\Omega=e^{2\omega}=1$ in the initial data set, hence in the whole plane wave.
  Alternatively, had we started from some other coordinate system $(\vplus,\vminus)$ we could have taken advantage that $\omega$ depends upon $\vplus$, only, and changed $\vplus$ into $\uplus(\vplus)$ with $\uplus'(\vplus) = \Omega(\vplus)$ (which is assumed to be positive). Then, $2 \Omega d\vplus d\vminus = 2 d\uplus d\vminus$, and without loss of generality we can thus assume that $\Omega= e^{2\omega} = 1$.
  In either approaches, the range of the variable $\uplus$ may be {\sl bounded}, in which case geodesics with fixed $x,y$ are incomplete as their proper time is proportional to~$\uplus$.
  To avoid this situation we therefore assume an {\sl unbounded} range for the variable~$\uplus$.

\item It should be pointed out that geodesics with non-constant $x,y$ reach $\uminus\to+\infty$ and infinite $x$ or~$y$ in finite proper time (or affine parameter), thus the region described by coordinates $(\uplus,\uminus)$ does not cover the whole single-wave spacetime.  However, in the full problem of colliding plane gravitational waves these geodesics simply enter the interaction region $\uminus>0$ in finite affine parameter and a deeper analysis, done later, is necessary to determine geodesic completeness.

\item Consequently, for a plane wave, the Raychaudhuri equation is the Riccati differential equation in the variable $\uplus$: 
\bel{eq:Ray}
a_{\uplus\uplus}  + (a_\uplus)^2 = -E \leq 0.
\ee 

\eei

\paragraph{The notion of gravitational pulse.}

Since \eqref{eq:Ray} is a Riccati equation, the coefficient $a$ typically blows up in a finite time (in the parameter $\uplus$).
This is readily seen by introducing the new coefficient $F$, which we also call the {\bf area function}, defined by 
\be
F^2 = A = e^{2a} \quad \text{ (when $A \neq 0$),} 
\ee
which, in view of \eqref{eq:Ray}, satisfies the (now linear) equation. 
\bel{equa:Flinear}
F_{\uplus\uplus}  = -E \, F. 
\ee
In principle, two initial data should be specified, say at $\uplus = 0$. However, by rescaling the coordinates $(x,y)$ if necessary, without loss of generality we can arrange that $F(0) = 1$. Thus, prescribing some data $F_0'$ we impose the initial conditions 
\bel{eq:ic}
F(0) = 1, \qquad F_\uplus(0) = F_0', 
\ee 
and the solution $F = F(\uplus)$ to \eqref{equa:Flinear} is then uniquely and globally defined.
The incoming radiation energy $E=E(\uplus)$ is a prescribed data of the problem (see next section) and, for the collision problem, it is natural to introduce the notion of a {\bf plane gravitational pulse}\footnote{Sometimes simply referred to as a gravitational wave in what follows}, defined as follows: 

\bei 

\item {\bf Trivial past.} 
We assume that, for all sufficiently large negative values $\uplus$ and, in fact, for all $\uplus<0$, 
the spacetime is flat and empty, so that
\[
F(\uplus) = 1, \qquad E(\uplus) =0, \qquad \text{for $\uplus < 0$}.
\]
Observe that such a trivial past imposes the initial data $\Fplus_0'=0$, as otherwise $R_{\uplus\uplus}$ would suffer a Dirac-type singularity, as seen in the first equation in~\eqref{equa:Ric1} below.

\item {\bf Finite pulse or half-infinite pulse.} Two cases of interest will arise whether the gravitational wave is compactly supported on an interval included in $[0, +\infty)$ or else is of infinite duration. 
\eei

\paragraph{Several kinds of pulses.}

Specifically, in a final section of this paper it will be useful to have the notation $[\uplus_\star^-, \uplus_\star^+]$ for the support of the pulse, with 
$0 \leq \uplus_\star^- \leq \uplus_\star^+ \leq + \infty$, and to write 
\[
E(\uplus) = 0 \qquad \text{for $\uplus < \uplus_\star^-$ or $\uplus > \uplus_\star^+$}.
\]
We can distinguish between several regimes: (1) A {\bf finite pulse} if $\uplus_\star^+$ is finite and then the geometry eventually approaches Minkowski geometry. (2) A {\bf short pulse} if $\uplus_\star^+ - \uplus_\star^-$ is small and, especially, the limit $\uplus_\star^+ - \uplus_\star^- \to 0$ is relevant. 
(3) An {\bf infinite pulse} if $\uplus_\star^+$ is infinite, and in this regime the geometry evolves forever.

\subsection{Prescribing the two incoming radiation data} 
\label{sec:64}

\paragraph{Incoming energy radiation.} 

We now return to the global collision problem and we are in a position to complete the description by introducing now two gravitational pulses, that is, a right-moving plane gravitational pulse defined for $\uplus \geq 0$ and a left-moving pulse defined for $0\leq \uminus$ ---the latter being
defined by replacing $\uplus$ by $\uminus$ throughout the construction in Section~\ref{sec:63}.
We are thus given two {\sl incoming energy functions}, denoted by $\Eplus, \Eminus$ and representing the energy fluxes of two pulses moving in opposite directions. 
These fluxes are determined
via the identities
\bel{eq:219} 
\Eplus = \psiplus_\uplus^2 + 4 \pi \phiplus_\uplus^2 ,
\qquad 
\Eminus = \psiminus_\uminus^2 + 4 \pi \phiminus_\uminus^2,
\ee
in which the functions $(\psiplus, \phiplus) = (\psiplus, \phiplus)(\uplus)$ and $(\psiminus, \phiminus) = (\psiminus, \phiminus)(\uminus)$ are prescribed data for the characteristic initial value problem posed on the null hypersurfaces $\Ncalplus_0$ and $\Ncalminus_0$, respectively.

\paragraph{Area functions of the incoming waves.} 

From these data, we determine the corresponding two area functions $\Fplus = \Fplus(\uplus)$ and $\Fminus = \Fminus(\uminus)$ by solving the Einstein equations \eqref{eq-a1-baru} and \eqref{eq-a1-ubar} restricted to the two initial hypersurfaces, i.e.
\bel{equa:Ffunctions}
\Fplus_{\uplus \uplus}
 = -  \Eplus \, \Fplus
\, \text{ in } \Mleftfar =  \big\{ 0 < \uplus; \, \uminus < 0 \big\};
\quad \qquad
\Fminus_{\uminus\uminus}
 =- \Eminus \, \Fminus
\, \text{ in } \Mrightfar = \big\{ \uplus < 0; \, 0 < \uminus \big\}.
\ee
The conditions \eqref{equa:Ffunctions} are second-order differential equations:
\bei 

\item The past of the two incoming waves is assumed to be trivial, that is,  
\[
\Fplus(\uplus) = 1 \, \text{ for all }  \uplus < 0; 
\qquad
\Fminus(\uminus) = 1 \, \text{ for all }   \uminus < 0.
\]

\item We supplement the equations with initial conditions of the form \eqref{eq:ic} for each wave with $\Fplus_0'=0$ and $\Fminus_0'=0$ due to the waves having a trivial past, i.e.
\bel{init-cond}
\Fplus(0) = \Fminus(0) = 1,
\qquad  
\Fplus_\uplus(0) = 0;
\qquad
\Fminus_\uminus(0) = 0.
\ee
\eei

\paragraph{Geometry of the past of the interaction domain.}

By construction, the areal coefficient $A$ is identically $1$ in the domain $\Mdown$ and we can also choose the metric coefficients $\Omega, \psi$ and matter field~$\phi$ within this domain to ensure that the spacetime is empty and flat: 
\[
\Omega=A=1,\quad \psi=\phi=0
\quad \text{ in the region } \Mdown.
\]

\paragraph{Initial data for the metric and matter.} 

In both domains $\Mleftfar$ and $\Mrightfar$, the geometry and matter content consist of plane waves determined by the prescribed data 
\[
\aligned
&
\psiplus= \psiplus(\uplus), \qquad \phiplus= \phiplus(\uplus), \qquad 
&& \text{ on the hypersurface } \Ncalplus_0, 
\\
&
\psiminus= \psiminus(\uminus), \qquad \phiminus= \phiminus(\uminus), \qquad 
&& \text{ on the hypersurface } \Ncalminus_0, 
\endaligned
\]
which also determines \eqref{eq:219}. 
Hence, the unknowns $\omega, a,\psi, \phi$ for the system \eqref{equa-EinEq} below are prescribed in the two incoming wave domains 
as 
\be
\aligned
& \Omega= 1, \quad A = \Fplus(\uplus)^2, \quad \psi = \psiplus(\uplus),
\quad \phi= \phiplus(\uplus)
\qquad \mbox{ in the domain } \Mrightfar,
\\
& \Omega= 1, \quad A = \Fminus(\uminus)^2, \quad \psi = \psiminus(\uminus),
\quad \phi= \phiminus(\uminus) 
\qquad \mbox{ in the domain } \Mleftfar. 
\endaligned
\ee
This completes the description of the global collision problem.

\subsection{Analysis of the Raychaudhuri equation}
\label{sec:Ray-analysis}

Before turning in Section~\ref{section---7} to the geometry and matter content in the interaction spacetime domain~$\Mup$, we summarize here the properties of the solutions to the differential equations~\eqref{equa:Ffunctions} satisfied by the area functions $\Fplus, \Fminus$.

\begin{proposition}[Properties of the global area functions]
 \label{prop:AreaFuncSm}
  Given sufficiently regular incoming radiation data prescribing energy fluxes $\Eplus, \Eminus$,
  there exist unique globally defined regular functions $\Fplus, \Fminus$ that satisfy the differential equations~\eqref{equa:Ffunctions} and the initial conditions~\eqref{init-cond}.
  These area functions $\Fplus$ and~$\Fminus$ have the following properties.
\bei 

\item {\bf Zeros.}
  They vanish at locally finitely many points, denoted in increasing order by
\[
  \uplus_0^{(i)}\in (\uplus_\star^-,+\infty) \text{ \ for } 0\leq i<\nplus \quad \text{and}\quad
  \uminus_0^{(i)}\in(\uminus_\star^-,+\infty) \text{ \ for } 0\leq i<\nminus,
\]
  respectively, where the numbers of zeros $\nplus,\nminus\in\ZZ_{\geq 0}\cup\{\infty\}$ can be infinite.
  One only has $\nplus=0$ or $\nminus=0$ (no zeros) if the corresponding energy flux $\Eplus$ or $\Eminus$ is identically vanishing.

\item {\bf Sign.}  They change sign at each zero.  For $\Fplus$, this means $(-1)^j\Fplus>0$ on
  each interval $(\uplus_0^{(j-1)},\uplus_0^{(j)})$, with the extended notations $\uplus_0^{(-1)}\coloneqq-\infty$ and $\uplus_0^{(\nplus)}\coloneqq+\infty$.

\item {\bf Convexity.} In each of these intervals $(\uplus_0^{(j-1)},\uplus_0^{(j)})$ their derivative is monotonic: non-increasing on intervals where the function is positive, and non-decreasing on the others.

\item {\bf Vacuum.}
  Finally, $\Fplus$, resp.~$\Fminus$, is an affine function in each interval on which one has vanishing energy flux~$\Eplus$, resp.~$\Eminus$.  In particular, it is constant equal to~$1$ in the initial
  interval $(-\infty,\uplus_\star^-]$, resp.\ $(-\infty,\uminus_\star^-]$.
  Generically this is the only interval with constant~$\Fplus$, resp.~$\Fminus$.
\eei
\end{proposition}

\begin{proof}
  The global existence of~$\Fplus$ (and likewise~$\Fminus$) can be seen from the (otherwise useless) explicit formula
  \[
  \begin{pmatrix} \Fplus(\uplus) \\ \Fplus_\uplus(\uplus) \end{pmatrix}
  = \left( \sum_{n\geq 0} \int_{0\leq\uplus_n\leq\dots\leq\uplus_1\leq\uplus} \prod_{i=1}^n
  \begin{pmatrix} 0 & 1 \\ -\Eplus(\uplus_i) & 0 \end{pmatrix} \, d^nu \right)
  \begin{pmatrix} \Fplus(0) \\ \Fplus_\uplus(0) \end{pmatrix} , \qquad
  \uplus > 0 .
  \]
  All components of the product of matrices are bounded by $\prod_{i=1}^n(\Eplus(\uplus_i)+1)$ hence the $n$-th term is bounded by $(\uplus+\lVert \Eplus\rVert_{L^1})^n/n!$, with a $1/n!$ factor coming from the domain of integration.  This ensures that the series converges for all~$\uplus$.  Checking it is a solution is straightforward.
  We now prove the properties in turn, concentrating on~$\Fplus$ for definiteness.

  {\sl Zeros.} If the set of zeros of~$\Fplus$ has an accumulation point, then $\Fplus=\Fplus_\uplus=0$ at that point.  Solving the second order differential equation of~$\Fplus$ starting from this point (and toward the past) gives $\Fplus=0$ identically, which contradicts the initial data $\Fplus(\uplus)=1$ for $\uplus<0$.
  If $\Fplus$~has no zero, then $\Fplus(\uplus)>0$ for all~$\uplus$ by continuity, so $\Fplus_{\uplus\uplus}=-\Eplus\Fplus\leq 0$ because $\Eplus\geq 0$.  Then $\Fplus$ is a positive concave function, hence a constant, so $\Fplus_{\uplus\uplus}=-\Eplus\Fplus=0$, which requires $\Eplus=0$ identically.

  {\sl Sign.} The derivative $\Fplus_\uplus$ may not vanish at a zero of~$\Fplus$ otherwise the solution~$\Fplus$ would vanish identically (by the same argument as point~1).  Thus $\Fplus$~changes sign at each zero.

  {\sl Convexity.} The second derivative $\Fplus_{\uplus\uplus}=-\Eplus\Fplus$ is zero or has a sign opposite to that of~$\Fplus$, because $\Eplus\geq 0$.

  {\sl Vacuum.} When the energy flux~$\Eplus$ vanishes the area function~$\Fplus$ is a solution of $\Fplus_{\uplus\uplus}=0$, namely is affine.  This occurs in particular on the intervals $(-\infty,\uplus_\star^-]$ and $[\uplus_\star^+, + \infty)$ (when non-empty), and for the first one the initial conditions imply that $\Fplus=1$.
  Conversely, having a constant~$\Fplus$ on any other interval requires $\Eplus=0$ on that interval (vacuum), together with a non-generic fine-tuning of the value of~$\Fplus_\uplus$ at the start of the interval.
\end{proof}


\section{Global spacetime geometry: singularities, islands, and diamonds} 
\label{section---7}

\subsection{The partition in spacetime islands} 
\label{sec:singularities}

\paragraph{The global area function.}

Throughout our analysis, we consider (piecewise) regular spacetimes.
The area function $A$ obeys a wave equation $A_{\uplus\uminus}=0$ on each domain of regularity,
and, later on, we select junction conditions that ensure this wave equation holds everywhere including at singularity hypersurfaces.
The solution~$A$ is thus explicitly given as the sum of a function of $\uplus$ and a function of~$\uminus$ .
From the characteristic initial data we obtain 
\bel{equa:Bexpli}
A(\uplus, \uminus) = \Fplus(\uplus)^2 + \Fminus(\uminus)^2 - 1, 
\ee
since it must coincide with $\Fplus(\uplus)^2$ or with $\Fminus(\uminus)^2$ along each
initial hypersurface.  In particular, $A$~vanishes along this hypersurface at the zeroes of the functions 
$\Fplus$ and $\Fminus$. We rely on properties of these functions determined in Section~\ref{sec:Ray-analysis} to partition the spacetime domain~$\Mup$.

\paragraph{Singularity locus.}

The area~$A$ of symmetry orbits
arises as a coefficient in the principal part of the system of reduced Einstein equations~\eqref{equa-EinEq} below,
hence the behavior of their solutions can be expected to be singular when $A$~vanishes.
We define the {\bf singular locus} of a plane-symmetric colliding spacetime as the set 
\bel{Lscr}
\Lscr 
\coloneqq \big\{ \uplus, \uminus \in [0, +\infty) \, \big/ \, A(\uplus,\uminus) = 0 \big\}
\subset \Cl(\Mup), 
\ee
where $\Cl(\Mup)$ denotes the closure, that is, the union of $\Mup$ with its boundary $\Ncalplus_0\cup\Pcal_0\cup\Ncalminus_0$. 
As we show below, the set~$\Lscr$ consists of the union of locally finitely many open sets (generically none), and of
locally finitely many curves in $(\uplus, \uminus)$ which correspond to
singularity hypersurfaces within $\Mcal$ and across which $A$~changes sign. 
We emphasize that while the coefficient $A$ was defined to be non-negative on the initial hypersurfaces, it {\sl changes sign} in the future of these hypersurfaces. Geometrically, $A$~represents (a signed version of) the {\sl area density of the surfaces of plane symmetry,} which degenerates along~$\Lscr$. We will show that, at least for ``generic'' initial data, the singularity hypersurfaces are genuine curvature singularities.

\paragraph{A first partition of interest.}

Partitioning the spacetime domain $\Mup$ along the singular locus~$\Lscr$ gives a natural decomposition into spacetime domains in which $A$ keeps a constant sign or vanishes identically (which is a non-generic situation). 
We refer to such domains as {\bf spacetime islands.}
These domains are limited by singularity hypersurfaces that may continuously change type from spacelike to null and to timelike.
Their structure is rather intricate, so for practical purposes a less natural partition defined below is more helpful.

\paragraph{The local coordinate functions.} 

From the initial area functions $\Fplus, \Fplus$, we define 
\be
\rplus(\uplus) = 2 \Fplus(\uplus)^2 - 1, \qquad \rminus(\uminus) = 2 \Fminus(\uminus)^2 - 1,
\ee
which we refer to as the {\bf local coordinates}. Observe that they satisfy $\rplus(0) = \rminus(0) = A(0,0) = 1$ on the initial plane $\Pcal_0$.
Using \eqref{equa:Bexpli}, the coefficient $A$ takes the particularly simple form
\be
A(\uplus, \uminus)  
= \frac12 \big( \rplus (\uplus)+ \rminus (\uminus) \big). 
\ee

The function $\rplus$ can only be used as a local coordinate in intervals where it is
monotonic.  By Proposition~\ref{prop:AreaFuncSm}, its derivative $\rplus'$ changes sign
locally finitely many times, only, so we can decompose~$\RR$ into intervals
$\Iplus_j=(\uplus_j,\uplus_{j+1})$, $0\leq j<\mplus$ (with $\uplus_0=-\infty$) in which
$\sgn(\rplus')\in\{+1,-1,0\}$ remains the same throughout the interval.  Denote
$\epsplus_j$ this sign.  The number $\mplus$ of intervals may be finite or infinite.  By
Proposition~\ref{prop:AreaFuncSm}, $\rplus$~can only remain constant on an interval (so
$\epsplus_j=0$) if that value is a local maximum, namely if $\epsplus_{j-1}=+1$ and $\epsplus_{j+1}=-1$
(with the convention that $\epsplus_{-1}=+1$).  Another
consequence of Proposition~\ref{prop:AreaFuncSm} is that $\rplus=-1$ at its local minima.
Thus, in an interval $\Iplus_j$ with $\rplus'>0$, resp.\ $\rplus'<0$, the range of values
of~$\rplus$ is $(-1,\rplus(\uplus_j))$ resp.\ $(-1,\rplus(\uplus_{j+1}))$.
The same comments and properties hold for $\rminus$ and we define $\Iminus_j$ and
$\epsminus_j$, $0\leq j<\mminus$ in the same way according to the sign of~$\rminus'$.
These observations are depicted in Figure~\ref{fig:rmonotonicity}.

\begin{figure}[t]\centering
\begin{tikzpicture}
    \draw[->,thick] (-1,0) -- (8,0) node [xshift=1ex] {$\uplus$};
    \draw[->,thick] (0,-1.2) -- (0,2) node [yshift=1ex]{$\rplus$};
    \draw[dotted] (-.2,-1) node [left=-.5ex] {\scriptsize $\rplus=-1$} -- (7,-1);
    \draw (-1,1) -- (0,1) node [above left=-.5ex] {\scriptsize $\rplus=1$}
    -- (0.8,1)
    .. controls +(0:.4) and +(180:.4) .. (1.8,-1)
    .. controls +(0:.4) and +(180:.4) .. (3.1,-.2)
    -- (4.3,-.2)
    .. controls +(0:.4) and +(180:.4) .. (5.1,-1)
    .. controls +(0:.4) and +(180:.4) .. (5.6,1.4)
    .. controls +(0:.4) and +(180:.4) .. (6.3,-1)
    parabola (7.6,2)
    ;
    \foreach\i/\x in {1/.8, 2/1.8, 3/3.1, 4/4.3, 5/5.1, 6/5.6, 7/6.3} {
      \draw[dashed] (\x,-2.15) -- (\x,.05);
      \draw[dashed] (\x,.35) -- (\x,1.7);
      \node at (\x,.2) {\scriptsize $\uplus_{\i}$};
    }
    \foreach \i/\x/\e in {0/0/0, 1/1.3/-1, 2/2.45/+1, 3/3.7/0, 4/4.7/-1, 5/5.35/+1, 6/5.95/-1, 7/6.9/+1} {
      \node at (\x,-1.5) {$\Iplus_{\i}$};
      \node at (\x,-1.9) {$\e$};
    }
    \node at (-.5,-1.9) {\vphantom{0}\smash{$\epsplus_j\colon$}};
  \end{tikzpicture}
  \caption{\label{fig:rmonotonicity}{\bf Local coordinate~$\rplus$ for a particular example of plane wave.}
    The intervals $\Iplus_j$ and signs $\epsplus_j$ describe monotonicity properties of~$\rplus$.
    In intervals where $\rplus$ is constant the initial data must be $\Eplus=0$, namely
    there is no incoming wave in these regions of~$\Mleftfar$, but the converse is not true.
    As depicted, $\rplus$ is
    always equal to~$-1$ at its local minima, but local maxima can take any value.
    If the pulse has a finite duration, then in the final interval either $\rplus$ is
    constant or it grows quadratically.}
\end{figure}
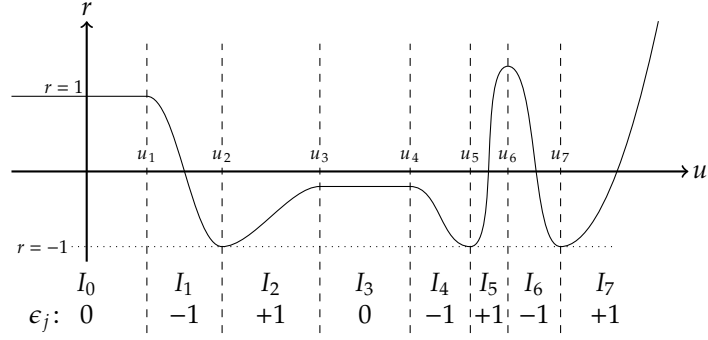

\subsection{The partition in monotonicity diamonds} 
\label{sec:diamonds}

\paragraph{A second partition of interest.}

We partition spacetime $\Mup$ into a locally finite collection of non-intersecting cells
$\Delta_{ij}=\{\uplus\in\Iplus_i,\uminus\in\Iminus_j\}$ such that each derivative
$\rplus'(\uplus)$ and $\rminus'(\uminus)$ has a constant sign $\pm 1$ or~$0$ in~$\Delta_{ij}$,
\[
\Mup = \bigcup_{i, j \geq 1} \Delta_{ij}.
\]
We refer to such a cell $\Delta_{ij}$ as a {\bf monotonicity diamond} because the area function $A$ depends monotonically on~$\uplus$ and on~$\uminus$.
As we will see, the function $A=(\rplus(\uplus) + \rminus(\uminus))/2$ may change sign within one diamond.
Note that the index $i$ (respectively~$j$) has a finite range if $\rplus$ (or~$\rminus$) is eventually monotonic.  This occurs if the incoming radiation is compactly supported or decays fast enough.
In that case some of the diamonds $\Delta_{ij}$ extend to infinity in (at least) one null direction.

\begin{figure}[t]\centering
  \begin{tikzpicture}[scale=.8]
    \begin{scope}[shift={(-2.2,0)}]
      \node(F) at (0,1.8) {\scriptsize future};
      \node(P) at (0,.3) {\scriptsize past};
      \draw[->](P)--(F);
      \begin{scope}[shift={(0,-2.1)}]
        \node(F) at (0,1.6) {\scriptsize future};
        \node(P) at (0,.3) {\scriptsize past};
        \draw[->](P)--(F);
      \end{scope}
    \end{scope}
    \begin{scope}[shift={(0,0)}]
      \draw[->] (0,0) -- (-1.3,1.3) node [below=0ex] {\scriptsize $\uplus$};
      \draw[->] (0,0) -- (1,1) node [below=0ex] {\scriptsize $\uminus$};
      \draw (-1.2,1.2) -- (-.3,2.1) -- (.9,.9);
      \draw (-.8,.8) .. controls +(45:.3) and +(160:.3) .. (-.2,.65)
      .. controls +(-20:.3) and +(135:.4) .. (.7,.7);
      \begin{scope}[shift={(0,-1.3)}]
        \draw[densely dashed,->] (.4,-.4) -- (-.65,.65) node [above=-.5ex] {\scriptsize $\rplus$};
        \draw[densely dashed,->] (-.4,-.4) -- (.65,.65) node [above=-.5ex] {\scriptsize $\rminus$};
        \draw (0,-.7) -- (-.9,.2) -- (0,1.1) -- (.9,.2) -- cycle;
        \draw (-.7,0) -- (.7,0);
      \end{scope}
    \end{scope}
    \begin{scope}[shift={(2.4,0)}]
      \draw[->] (0,0) -- (-1.3,1.3);
      \draw[->] (0,0) -- (1,1);
      \draw (-1.2,1.2) -- (-.3,2.1) -- (.9,.9);
      \draw (-1.2,1.2) .. controls +(0:.6) and +(-160:.4) .. (0,.8)
      .. controls +(20:.3) and +(135:.4) .. (.7,.7);
      \begin{scope}[shift={(0,-1.3)}]
        \draw[densely dashed,->] (-.4,.4) -- (.45,-.45) node [below=-.5ex] {\scriptsize $\rplus$};
        \draw[densely dashed,->] (.6,.6) -- (-.45,-.45) node [below=-.5ex] {\scriptsize $\rminus$};
        \draw (0,-.7) -- (-.7,0) -- (.2,.9) -- (.9,.2) -- cycle;
        \draw (-.7,0) -- (.7,0);
      \end{scope}
    \end{scope}
    \begin{scope}[shift={(4.8,0)}]
      \draw[->] (0,0) -- (-1.3,1.3);
      \draw[->] (0,0) -- (1,1);
      \draw (-1.2,1.2) -- (-.3,2.1) -- (.9,.9);
      \draw (.7,.7) .. controls +(135:1.5) and +(-45:1.3) .. (-1.1,1.3);
      \begin{scope}[shift={(0,-1.3)}]
        \draw[densely dashed,->] (.4,-.4) -- (-.3,.3) node [above=-.5ex] {\scriptsize $\rplus$};
        \draw[densely dashed,->] (-.4,-.4) -- (.65,.65) node [above=-.5ex] {\scriptsize $\rminus$};
        \draw (0,-.7) -- (-.5,-.2) -- (.4,.7) -- (.9,.2) -- cycle;
        \draw (-.3,0) -- (.7,0);
      \end{scope}
    \end{scope}
    \begin{scope}[shift={(7.8,0)}]
      \draw[->] (0,0) -- (-1.3,1.3);
      \draw[->] (0,0) -- (1,1);
      \draw (-1.2,1.2) -- (-.3,2.1) -- (.9,.9);
      \draw (.2,.2) .. controls +(135:1.5) and +(-45:1.3) .. (-.4,2);
      \begin{scope}[shift={(0,-1.3)}]
        \draw[densely dashed,->] (-.4,.4) -- (.45,-.45) node [below=-.5ex] {\scriptsize $\rplus$};
        \draw[densely dashed,->] (-.6,-.6) -- (.7,.7) node [above=-.5ex] {\scriptsize $\rminus$};
        \draw (-.2,-.9) -- (-.9,-.2) -- (.2,.9) -- (.9,.2) -- cycle;
        \draw (0,-.7) -- (0,.7);
      \end{scope}
    \end{scope}
    \begin{scope}[shift={(10.2,0)}]
      \draw[->] (0,0) -- (-1.3,1.3);
      \draw[->] (0,0) -- (1,1);
      \draw (-1.2,1.2) -- (-.3,2.1) -- (.9,.9);
      \draw (0,0) .. controls +(60:1.3) and +(-115:1) .. (-.3,2.1);
      \begin{scope}[shift={(0,-1.3)}]
        \draw[densely dashed,->] (.5,-.5) -- (-.5,.5) node [above=-.5ex] {\scriptsize $\rplus$};
        \draw[densely dashed,->] (.4,.4) -- (-.6,-.6) node [below=-.5ex] {\scriptsize $\rminus$};
        \draw (0,-.9) -- (-.8,-.1) -- (0,.7) -- (.8,-.1) -- cycle;
        \draw (0,-.9) -- (0,.7);
      \end{scope}
    \end{scope}
    \begin{scope}[shift={(12.6,0)}]
      \draw[->] (0,0) -- (-1.3,1.3);
      \draw[->] (0,0) -- (1,1);
      \draw (-1.2,1.2) -- (-.3,2.1) -- (.9,.9);
      \draw (-1,1) .. controls +(45:.2) and +(-45:.1) .. (-.9,1.5);
      \begin{scope}[shift={(0,-1.3)}]
        \draw[densely dashed,->] (-.3,.3) -- (.45,-.45) node [below=-.5ex] {\scriptsize $\rplus$};
        \draw[densely dashed,->] (-.2,-.2) -- (.7,.7) node [above=-.5ex] {\scriptsize $\rminus$};
        \draw (.2,-.5) -- (-.4,.1) -- (.3,.8) -- (.9,.2) -- cycle;
        \draw (0,-.3) -- (0,.5);
      \end{scope}
    \end{scope}
    \begin{scope}[shift={(15.6,0)}]
      \draw[->] (0,0) -- (-1.3,1.3);
      \draw[->] (0,0) -- (1,1);
      \draw (-1.2,1.2) -- (-.3,2.1) -- (.9,.9);
      \draw (-.7,.7) -- (.2,1.6);
      \begin{scope}[shift={(0,-1.3)}]
        \draw[->] (.5,-.5) -- (-.5,.5) node [above=-.5ex] {\scriptsize $\rplus$};
        \draw (-.05,-.05) -- (.05,.05);
      \end{scope}
    \end{scope}
  \end{tikzpicture}
  \caption{\label{fig:single-diamonds}{\bf Singular hypersurfaces
      $\rplus+\rminus=0$ in monotonicity diamonds.}  Three examples with spacelike hypersurfaces that end in
    three different ways on the west side; three with timelike hypersurfaces; one with null
    hypersurface.  The first row is in coordinates $\uplus,\uminus$; the second in local
    coordinates $\rplus,\rminus$ where the dashed lines are the axes $\rplus=0$ and
    $\rminus=0$.  We illustrate various possible monotonicities for $\rplus, \rminus$ by
    orienting their axes appropriately.  All diagrams are oriented such that the top is in
    the future of the bottom.  The last diagram of the second row is degenerate: in this
    diamond, $\rminus$ is a constant and the singularity is simply at $\rplus=-\rminus$.}
\end{figure}
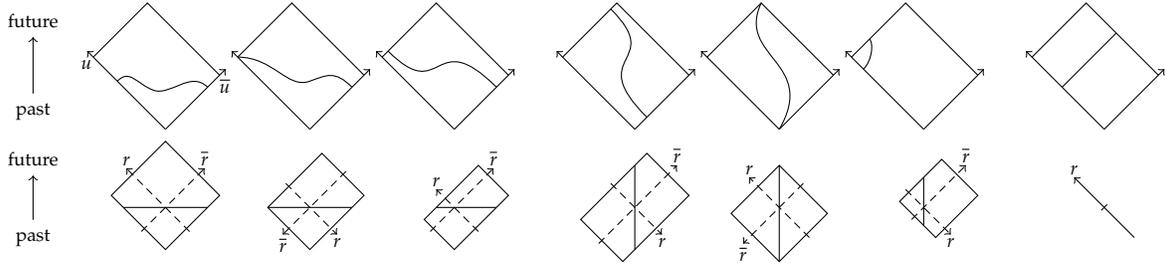

\paragraph{Four types of diamonds.} 

We classify monotonicity diamonds in terms of the relative monotonicity properties of the coordinates $\rplus, \rminus$.  These control the timelike, spacelike, or null nature of the gradient
\[
\nabla A = (A_{\uplus},A_{\uminus}) = \frac{1}{2} \bigl(  \rplus'(\uplus), \rminus'(\uminus) \bigr) ,
\]
whose norm squared is $g(\nabla A,\nabla A) = - \frac{1}{2} \Omega \rplus'(\uplus) \rminus'(\uminus)$.
This in turn determines the spacelike, timelike, or null nature of level sets of~$A$, hence of the singularity hypersurfaces $A=0$.  Different behaviors may arise at the intersection of such a hypersurface and the boundary of a diamond, as we now describe.  Several examples of diamonds are illustrated in Figure~\ref{fig:single-diamonds}, and a global picture of a spacetime partitioned into diamonds is given in Figure~\ref{fig:1.2}.

\bei

\item {\bf Diamond with timelike~$\nabla A$.} When the functions $\rplus$ and $\rminus$ have the same monotonicity in $\Delta$, the singular locus $\rplus + \rminus = 0$ in $\Delta$ is empty or is a spacelike hypersurface.  In the second case, the closure of this hypersurface in $\Cl(\Delta)$ is null on the boundary, unless it reaches a ``corner'' of the diamond in which case the closure can also be spacelike at the boundary.

\item {\bf Diamond with spacelike~$\nabla A$.} When the functions $\rplus$ and $\rminus$ have opposite monotonicities in $\Delta$, the singular locus in $\Delta$ is empty or is a timelike hypersurface whose closure is null on the boundary, unless it reaches a ``corner'' of the diamond in which case the closure can also be timelike at the boundary.

\item {\bf Diamond with null~$\nabla A$.} When one of the functions $\rplus$ and $\rminus$ is constant in $\Delta$ while the other one is monotonic, then the singular locus is empty or is a null hypersurface whose closure is null at the boundary.
This case only occurs for fine-tuned data: for generic incoming data, including cases with vacuum on some intervals, the functions~$\rplus,\rminus$ are never constant on an interval after the two plane waves begin interacting.

\item {\bf Diamond with constant~$A$.} When both functions $\rplus$ and $\rminus$ are constant in $\Delta$, either $\rplus+\rminus\neq 0$ and there is no singular locus, or $\rplus + \rminus = 0$ and the area function $A$ vanishes identically. In the latter case, the metric is fully degenerate within such a monotonicity diamond, and the value of the coefficients $\psi, \omega, \phi$ is (physically and mathematically) irrelevant\footnote{It might be possible to still define the values of the functions $\psi,\omega, \phi$ in such a diamond, by solving suitably chosen evolution equations but we will not pursue this here.}. 
As in the case of null~$\nabla A$ this phenomenon is excluded for generic incoming data, including cases with vacuum regions.

\eei

In our existence Theorem~\ref{theorem-plane}, we exclude the last two types of diamonds by assuming that neither $\rplus$ nor $\rminus$ are constant on intervals (besides the initial interval).  We also exclude the non-generic situation in which $\rplus+\rminus=0$ exactly at a corner of a monotonicity diamond, which can only occur if local maxima of $\rplus$ and~$\rminus$ are exactly opposite (local minima are always $-1$ hence cannot be opposite).  Without this latter restriction the singularity locus $\{A=0\}$ would include isolated values of $(\uplus,\uminus)$, which would not fit in our notion of cyclic spacetime.

\begin{figure}[t]\centering
  \includegraphics[height=7cm,trim=0 29 0 29]{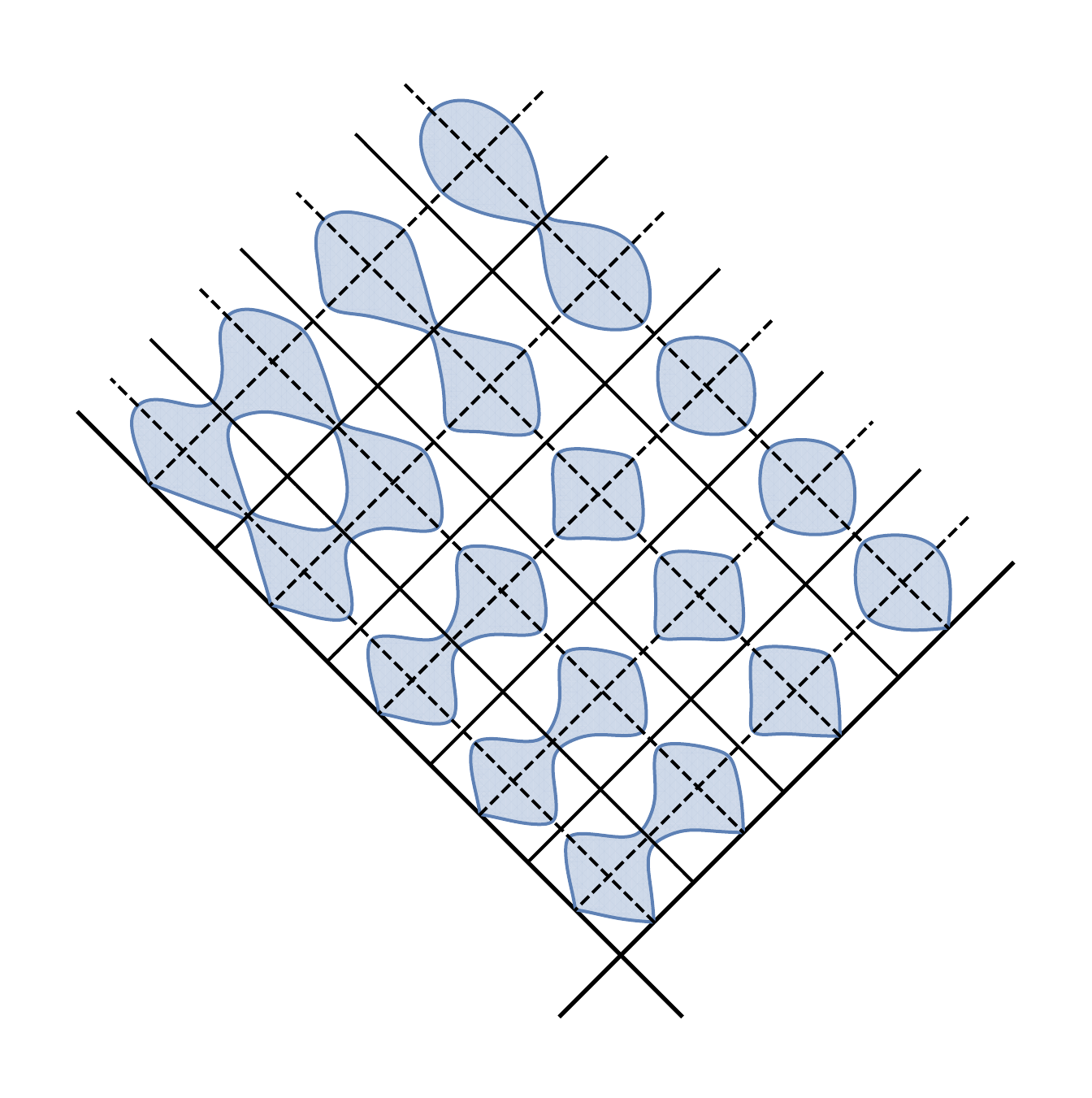}

  \caption{\label{fig:1.2}
    {\bf Global geometry of a colliding spacetime.} 
    The dashed black lines are loci where $\Fplus(\uplus)=0$ or $\Fminus(\uminus)=0$ and
    the solid black lines are loci where $\Fplus'(\uplus)=0$ or $\Fminus'(\uminus)=0$.
    For generic initial data these are null hypersurfaces at constant $\uplus$ or constant $\uminus$,
    and they decompose spacetime into diamond-shaped regions.
    Each such diamond is cut through by at most one spacelike or timelike singularity hypersurface, generically with null tangents at both ends (except along $\Ncalplus_0$ and~$\Ncalminus_0$ where tangents are not null).
    These singularity hypersurfaces (curved blue lines) interpolate from spacelike to timelike parts, and are generically smooth away from the axes $\uplus=0$ and $\uminus=0$.
    For generic initial data with compact support or sufficient decay at $\uplus\to+\infty$ and at $\uminus\to+\infty$, there are finitely many islands.}
\end{figure}

\subsection{Evolution equations in diamonds}

\paragraph{The essential evolution equation.}

According to the Einstein equations and the Bianchi identities, the matter field $\phi$ satisfies the wave equation 
$
\Box_g \phi = 0, 
$
which we can express (cf.~the equation \eqref{cphi-eq-scalarfield}) within each diamond, as 
\bel{phi-waveeq}
\phi_{\uplus\uminus} 
= - {1 \over 2A} \Big( A_\uplus \, \phi_\uminus + A_\uminus \, \phi_\uplus \Big)
= - {1 \over 2(\rplus + \rminus)} \Big( \rplus' \, \phi_\uminus + \rminus' \, \phi_\uplus \Big), 
\ee
in which $\rplus' = {d \over d\uplus} \rplus(\uplus)$, etc. Within a diamond, let us parametrize the hypersurface as $(\uplus(w), \uminus(w), x, y)$ in which $x,y \in \RR^2$ and the parameter $w$ varies in an interval. By definition, we have $\rplus(\uplus(w)) + \rminus(\uminus(w)) = 0$ and, therefore, by differentiation,  
$
\uplus'(w)\,\rplus'(\uplus(w)) + \uminus'(w) \, \rminus'(\uminus(w)) = 0. 
$
If at least one of the two components $\uplus'(w)$ and $\uminus'(w)$ is not vanishing, the tangent vector 
 to the curve is $(\uplus'(w), \uminus'(w))$ and we can take the normal vector (which needs not be unit)
\[
n(w) \coloneqq - \nabla A(w) = -(A_{\uplus},A_{\uminus}) = \frac{-1}{2} \bigl(  \rplus'(\uplus(w)), \rminus'(\uminus(w)) \bigr) ,
\]
oriented toward negative values of~$A$.
This defines our choice of orientation of the singularity hypersurface (relative to the orientation of the spacetime) for the purpose of defining a 
cyclic spacetime \cite{LLV-1a}.

\paragraph{Three types of evolution equations.}

Let us express the wave equation~\eqref{phi-waveeq} in each type of diamonds distinguished above.

\bei 

\item {\bf Diamond with timelike~$\nabla A$.} When $\rplus'$ and $\rminus'$ have the same sign $\epsilon=\pm 1$, it is natural to introduce the following time and space variables, where the sign is chosen to make $t$ increase toward the future:
\[
t = \epsilon (\rplus + \rminus)/2,
\qquad
z = \epsilon (- \rplus + \rminus)/2.
\]
The wave equation for the matter field becomes
$
\phi_{tt} + {1 \over t} \phi_t - \phi_{zz}  = 0 .
$
The singularity at $t=0$ is located in the future of the initial data: one should advance toward the singularity hypersurface then apply a suitable jump condition.

\item {\bf Diamond with spacelike~$\nabla A$.}
When $\rplus'$ and $\rminus'$ have opposite signs $-\epsilon$ and $\epsilon$, it is more natural to set
\[
t = \epsilon (- \rplus + \rminus)/2,
\qquad
z = \epsilon ( \rplus + \rminus)/2,
\]
and the wave equation for the matter field becomes
$
\phi_{tt} - \phi_{zz} - {1 \over z} \phi_z = 0 .
$
Here, $z=0$ corresponds to the singularity and it is necessary to solve simultaneously for $z<0$ and $z>0$ and impose a jump condition at the interface $z=0$. 

\item {\bf Diamond with null or vanishing~$\nabla A$.} When one or both of $\rplus'$ and $\rminus'$ vanishes identically we can no longer use $(\rplus,\rminus)$ as local coordinates. Instead, we rewrite the equation as
\[
\bigl((\rplus+\rminus)^{1/2} \phi\bigr)_{\uplus\uminus} = {- \rplus'\rminus' \over 4 (\rplus+\rminus)^{3/2}} \phi = 0.
\]
This wave equation for $(\rplus+\rminus)^{1/2} \phi$ is easily solved starting from initial values of this function on the past boundary of the diamond, which consists of two light-like hypersurfaces.

In fact, consider an open interval $\Iplus_j$ constructed above on which {\sl $\rplus$~is a constant}.
Since $\rplus=2\Fplus^2-1$ and $\Fplus$ is given by the initial data $\Eplus$ through $\Fplus''=-\Eplus\Fplus$ (with $\Fplus\neq 0$ whenever $\Fplus'=0$), we deduce that $\Eplus=\psiplus_\uplus^2+\frac{1}{2}\phiplus_\uplus^2$ vanishes on~$\Iplus_j$, namely there is no incoming wave initially in this interval.  The wave equation propagates this initial data $((\rplus+\rminus)^{1/2}\phi)_\uplus=0$ from $\uminus=0$ to all~$\uminus$ and we conclude that $\phi_\uplus(\uplus,\uminus)=0$ for all $\uplus\in\Iplus_j$ and all $\uminus\in\RR$.
For the same reasons $\psi_\uplus$ also vanishes.
The evolution equation~\eqref{eq-a2} of~$\omega$ reads $\omega_{\uplus\uminus}=0$ in this region, and, combining with our normalization $\omega=0$ in the initial data, we learn that $\omega_{\uplus}=0$.
Altogether, $\phi_\uplus=\psi_\uplus=\omega_\uplus=A_\uplus=0$ for $\uplus\in\Iplus_j$ so that the metric in this region is that of a single right-moving plane-symmetric gravitational wave.
Likewise, on intervals $\Iminus_j$ with $\rminus'(\uminus)=0$, the metric is that of a left-moving plane wave.

\eei


\section{Formulation in coordinates and junction conditions}
\label{section---8}

\subsection{Einstein's field equations in plane-symmetry}
\label{sec:81}

\paragraph{Metric in global null coordinates.}

Before we discuss the actual construction within each monotonicity diamond, we need to write down the Einstein equations in coordinates as well as the notion of scattering maps, and provide a couple of technical observations.  
In the coordinates $(\uplus, \uminus, x,y)$ introduced in Section~\ref{section--31} with $(\uplus, \uminus)  \in \RR^2$ and $(x,y) \in \Tbb^2$ (or $\RR^2$), the metric $g$ takes the form  
\bel{equa:metricg}
g  = -2 \Omega \, d\uplus d\uminus + X \, dx^2 + Y \, dy^2,
\ee
where the coefficients $\Omega,X,Y$ depend upon the null variables $(\uplus,\uminus)$, only.
We will allow for the coefficients $\Omega,X,Y$ to vanish or blow up along certain (spacelike, null, or timelike) hypersurfaces.
However, for the sake of clarity in the presentation, we normalize their sign to be positive away from these singularity hypersurfaces, 
that is, $\Omega, X, Y>0$ (except at singularity hypersurfaces).
Clearly, the metric \eqref{equa:metricg} has Lorentzian signature $(-,+,+,+)$ and can also be written in terms of time and space variables $\tau \coloneqq (\uplus + \uminus)/2$ and $\zeta \coloneqq (- \uplus + \uminus)/2$ as
\be
g  = 2 \Omega \, ( - d\tau^2 + d\zeta^2) + X \, dx^2 + Y \, dy^2. 
\ee

Einstein's field equations for the metric \eqref{equa:metricg} are most conveniently expressed in terms of the metric coefficients $\omega$, $a$, $\psi$ defined by
\be
\Omega = e^{2\omega}, \qquad X = e^{2(a+\psi)}, \qquad  Y = e^{2(a-\psi)}
\ee
away from singularity hypersurfaces, and such that $\omega$, $a$, $\psi$ may blow up to $\pm \infty$ at these singularity hypersurfaces in order to accommodate coefficients $\Omega,A,X,Y$ that may vanish or blow up. 
 Furthermore, from $a$ we define an auxiliary metric coefficient denoted by $A \in \RR$ and satisfying
\be
A^2 = e^{4a} = X Y \quad \text{(when $A \neq 0$).} 
\ee
Here, $\omega$ and $a$ are precisely the conformal coefficient and area function defined earlier in Section~\ref{section--31}.
For the class of singularity scattering maps we use to traverse singularities, the coefficient~$A$ will turn out to be a smooth function on spacetime, provided one chooses the sign of~$A$ to change across each singularity hypersurface. In the following, we will use the metric and matter variables $\omega, a,\psi, \phi$ or $\omega, A,\psi, \phi$ depending on convenience.

\paragraph{Ricci curvature in global null coordinates.}

In order to express the field equations for a sufficiently regular metric (i.e.\ away from any singularity), we compute the components of the Ricci tensor of the metric \eqref{equa:metricg} as follows: 
\bel{equa:Ric1}
\aligned
R_{\uplus\uplus} 
& = - 2 \, (\psi_\uplus)^2 + {2 \over A} \, \omega_\uplus A_\uplus + {1 \over 2 A^2} (A_\uplus)^2 - {1 \over A} A_{\uplus \uplus},
\qquad 
& R_{xx} 
& = e^{-2\omega +2\psi} \big( A_\uminus \psi_\uplus + \psi_\uminus A_\uplus + 2 \, A \psi_{\uplus \uminus} + A_{\uplus \uminus} \big),
\\
R_{\uplus \uminus} 
& = -2 \psi_\uplus \psi_\uminus - 2\omega_{\uplus \uminus} + {1 \over 2 A^2} A_\uminus A_\uplus - {1 \over A} A_{\uminus \uplus},
\qquad
& R_{yy} 
& = - e^{-2\omega-2\psi} \big( A_\uminus \psi_\uplus + \psi_\uminus A_\uplus + 2 \, A \psi_{\uplus \uminus} - A_{\uplus \uminus} \big), 
\\
R_{\uminus \uminus} 
& = - 2 (\psi_\uminus)^2 + {2 \over A} \omega_\uminus A_\uminus + {1 \over 2 A^2} (A_\uminus)^2 - {1 \over A} A_{\uminus \uminus}, 
\endaligned
\ee
while by virtue of the Einstein equations
we find 
\[
R_{\uplus\uplus} = 8 \pi \, (\phi_\uplus)^2,
\qquad
R_{\uminus\uminus} =  8 \pi \, (\phi_\uminus )^2,
\qquad
R_{\uplus\uminus} =  8 \pi \, \phi_\uplus \,  \phi_\uminus, 
\qquad
R_{xx} = R_{yy} = 0. 
\]
Therefore, the metric and matter variables $\omega, A,\psi, \phi$ satisfy the following set of equations: 
\[ 
\aligned
{2 A_\uplus  \over A} \, \omega_\uplus
& =  8 \pi \,  (\phi_\uplus)^2 +  2 \, (\psi_\uplus)^2 - {(A_\uplus)^2 \over 2 A^2}  + {1 \over A} A_{\uplus \uplus},
\qquad
& {2 A_\uminus \over A} \omega_\uminus
=  8 \pi \,  (\phi_\uminus)^2  + 2 (\psi_\uminus)^2 - {(A_\uminus)^2 \over 2 A^2} + {1 \over A} A_{\uminus \uminus},
\\
2\omega_{\uplus \uminus}
& = -2 \psi_\uplus \psi_\uminus -  8 \pi \,  \phi_\uplus \,  \phi_\uminus + {1 \over 2 A^2} A_\uminus A_\uplus - {1 \over A} A_{\uminus \uplus},
\\
0 &
= e^{-2\omega +2\psi} \bigl( A_\uminus \psi_\uplus + \psi_\uminus A_\uplus + 2 \, A \psi_{\uplus \uminus} + A_{\uplus \uminus} \bigr),
\qquad
& 0  
=  - e^{-2\omega-2\psi} \bigl( A_\uminus \psi_\uplus + \psi_\uminus A_\uplus + 2 \, A \psi_{\uplus \uminus} - A_{\uplus \uminus} \bigr).
\endaligned
\]
By linearly combining the last two equations, we deduce that $A_{\uplus\uminus} = 0$, so that $\psi$ satisfies
$
\psi_{\uplus\uminus} + {1 \over 2A} \big( A_\uplus \psi_\uminus + A_\uminus \psi_\uplus \big) = 0.
$
Finally,
the Bianchi identity provides us with a wave equation also for the scalar field, namely 
$
\phi_{\uplus\uminus} + {1 \over 2A} \big( A_\uplus \phi_\uminus + A_\uminus \phi_\uplus \big) = 0. 
$

\paragraph{Evolution system for the metric and matter field.}

After re-ordering and using the notation $\omega,a,\psi, \phi$, we conclude that the field equations  
are equivalent to a coupled system of second-order partial differential equations of semi-linear type,
\begin{subequations}
\label{equa-EinEq}
\begin{align}
\label{eqxi}
(e^{2a})_{\uplus\uminus} &=0,
\\
\label{cphi-eq-deux}
\psi_{\uplus\uminus}  &= - a_\uplus \, \psi_\uminus - a_\uminus \, \psi_\uplus,
\\
\label{cphi-eq-scalarfield}
\phi_{\uplus\uminus}  &= - a_\uplus \, \phi_\uminus - a_\uminus \, \phi_\uplus,
\\
\label{eq-a1-baru}
2 \omega_\uplus \, a_\uplus
&= a_{\uplus\uplus}  + (a_\uplus)^2 +  (\psi_\uplus)^2 + 4 \pi (\phi_\uplus)^2,
\\
\label{eq-a1-ubar}
2 \omega_\uminus \, a_\uminus  &= a_{\uminus\uminus} + (a_\uminus)^2 + (\psi_\uminus)^2 + 4 \pi (\phi_\uminus)^2,
\\
\label{eq-a2}
\omega_{\uplus\uminus}
&= - \psi_\uplus \, \psi_\uminus - 4 \pi \phi_\uplus\,\phi_\uminus + a_\uplus \, a_\uminus.
\end{align}
\end{subequations}
These equations are valid {\sl away from singularities} and enjoy the following properties. 

\bei 

\item  The evolution equation \eqref{eqxi} satisfied by $e^{2a}$ is the standard wave equation and is
solved explicitly from the incoming radiation data (see \eqref{equa:Bexpli}). This coefficient is the area (or area density in the non-compact case) of the symmetry orbits parametrized by $(x,y)$, and we refer in the following to $A$, or equivalently $a$, as the {\bf areal coefficient.} 

\item Then, it is convenient to refer to \eqref{cphi-eq-deux}--\eqref{cphi-eq-scalarfield} as the {\bf essential Einstein equations}, since they are decoupled from the remaining equations and can also be solved independently. 
They are linear wave equations (in curved space) and have {\sl singular coefficients} that blow up when $A$ approaches zero.

\item The two equations \eqref{eq-a1-baru} and \eqref{eq-a1-ubar} are differential equations along the characteristic directions, and can be solved for~$\omega$ once $\psi$ and $\phi$ are known. 

\item Finally, the remaining equation \eqref{eq-a2} is a direct consequence of the other equations.

\eei 

While the essential equations appear to be linear in nature, they do depend on the areal coefficient~$A$ which approaches zero and changes sign on certain hypersurfaces, as we studied in Section~\ref{section---7}.
The dependence on~$A$ is {\sl truly nonlinear,} and the collision problem shows that the Einstein equations in plane symmetry do allow for curvature singularities along these singularity hypersurfaces.
In addition, junction conditions across singularity hypersurfaces may be nonlinear and may couple $\phi$ and~$\psi$, as we discuss momentarily.

\begin{remark}[Invariance under null coordinate transformation]
\label{rem:invariance}
In the expression \eqref{equa:metricg} of the metric, each null coordinate can be rescaled by a composition with an arbitrary function while keeping the general form of the metric.
The field equations are also invariant under such a change of coordinates. More precisely, setting $\uplus = f(\uplus)$ and $\uminus = g(\uminus)$, we should take into account that the metric coefficient $\omega$ transforms non-trivially, as follows: $e^{2\omega} \mapsto e^{2 \widetilde{\omega}} = e^{2\omega} \frac{df}{ d\uplus} \frac{dg}{ d\uminus}$ while
$A_{\uplus\uplus} - 2 \omega_\uplus A_\uplus = e^{2\omega} (e^{-2\omega} A_{\uplus})_{\uplus}$.
Such a rescaling could be used to adapt the null coordinates to the problem under consideration; however our geometric construction of the coordinates already selects a preferred choice (up to linear rescaling) by requiring that ${\del\over\del\uplus}$ and ${\del\over\del\uminus}$ be geodesic fields on the $\uminus=0$ and $\uplus=0$ initial hypersurfaces, respectively.
\end{remark}

\subsection{Fuchsian expansions near singularities}
\label{sec:82}

\paragraph{From Fuchsian data to singularity scattering data.}
Scattering maps were defined and classified in~\cite{LLV-1a} by working in the ADM formalism in a Gaussian foliation, that is, a foliation by normal proper time or distance.
Our aim is to reformulate them in plane symmetry, in the null coordinate system $(\rplus,\rminus)$ that we used to describe one diamond of the colliding gravitational wave spacetime.
As a first step, we begin by expanding our variables $\omega,a,\phi,\psi$ on either side of the singularity $\rplus+\rminus=0$ in terms of functions $\phi_\sharp,\phi_\flat,\psi_\sharp,\psi_\flat,\upsilon_\sharp,\upsilon_\flat$ of a coordinate $w=-\rplus+\rminus$ along the singularity.
Then we set up a Gaussian foliation and determine the singularity scattering data $(g,K,\phi_0,\phi_1)=\SscrADMfromF(\phi_\sharp,\phi_\flat,\psi_\sharp,\psi_\flat,\upsilon_\sharp,\upsilon_\flat)$ describing the same asymptotic spacetime in the ADM formalism; the map is given in~\eqref{sharpflat-to-01}.
We then construct, in~\eqref{01-to-sharpflat} below the inverse which is a map $\SscrFfromADM$ that reconstructs Fuchsian data from plane-symmetric singularity scattering data.
In the next section we are going to compose these maps together and explicitly rewrite singularity scattering maps as relating Fuchsian data on both sides of the singularity hypersurface.

\paragraph{Fuchsian data.}

In coordinates $(\rplus,\rminus)$, the metric is
\be
g^{(4)} = 4\epsilon e^{2\upsilon}d\rplus d\rminus+e^{2(a+\psi)}dx^2+e^{2(a-\psi)}dy^2 , \qquad
\upsilon\coloneqq \omega-\frac{1}{2}\log|2\rplus'\rminus'| ,
\ee
where the sign $\epsilon=\sgn(-\rplus'\rminus')$ depends on the diamond.
We will need the $\rplus+\rminus\to 0^{\pm}$ expansion
derived later for~$\phi$, and its analogue for~$\psi$, along the singularity hypersurface where $A=(\rplus+\rminus)/2$ vanishes:
\bse\label{asympt}
\be
\phi = \phi_\sharp^\pm(-\rplus+\rminus) \log |\rplus+\rminus| + \phi_\flat^\pm(-\rplus+\rminus) + o(1) , \qquad
\psi = \psi_\sharp^\pm(-\rplus+\rminus) \log |\rplus+\rminus| + \psi_\flat^\pm(-\rplus+\rminus) + o(1) .
\ee
As long as it is clear from context, we shall suppress from the notation the subscript $\pm$ pertaining to the side of the singularity, and the dependence of coefficients on $(-\rplus+\rminus)$.

The evolution equation giving~$\omega_\rplus$ (or equivalently~$\upsilon_\rplus$) can be written compactly and expanded as
\[
\upsilon_\rplus
= (\rplus+\rminus) \bigl(4\pi\phi_\rplus^2+\psi_\rplus^2\bigr) - {1\over 4(\rplus+\rminus)}
= \biggl(\biggl(4\pi\phi_\sharp^2+\psi_\sharp^2-\frac{1}{4} \biggr)\log|\rplus+\rminus|\biggr)_\rplus
  - 8\pi\phi_\flat'\phi_\sharp - 2\psi_\flat'\psi_\sharp
  + o(1) .
\]
Combining this with the analogous equation for $\upsilon_\rminus$ leads to the expansion
\bel{alphasharpflat}
\gathered
\upsilon = \upsilon_\sharp(-\rplus+\rminus) \log |\rplus+\rminus| + \upsilon_\flat(-\rplus+\rminus) + o(1) , \\
\upsilon_\sharp = 4\pi\phi_\sharp^2+\psi_\sharp^2-\frac{1}{4} , \qquad
\upsilon_\flat' = 8\pi\phi_\flat'\phi_\sharp + 2\psi_\flat'\psi_\sharp ,
\endgathered
\ee
where we only provide an expression for the derivative of $\upsilon_\flat$ with respect to its argument, that is, the coordinate $-\rplus+\rminus$ parallel to the singularity. This only determines $\upsilon_\flat$ up to a constant, which can be computed from the initial data but which we will not need for now. Finally, we have 
\be
a=\frac{1}{2}\log|A|=a_\sharp\log|\rplus+\rminus|+a_\flat, \qquad a_\sharp=\frac{1}{2}, \qquad a_\flat=-\frac{1}{2}\log 2.
\ee
\ese

\begin{remark}\label{rem:ups-useless}
  Since \eqref{alphasharpflat} only determines $\upsilon_\flat$ up to a constant, the functions $(\phi_\sharp,\phi_\flat,\psi_\sharp,\psi_\flat)$ are not quite enough to specify Fuchsian data.  For this reason we shall work with extended data $(\phi_\sharp,\phi_\flat,\psi_\sharp,\psi_\flat,\upsilon_\sharp,\upsilon_\flat)$, constrained by~\eqref{alphasharpflat}.  Below, we identify these two relations as Einstein's Hamiltonian and momentum constraints.
\end{remark}

\paragraph{Gaussian coordinates.}

We now relate the coordinates $(\rplus,\rminus,x,y)$ to a Gaussian coordinate system $(s,w,x,y)$ with the same $(x,y)$.  On the singularity we set $w=-\rplus+\rminus$.  Then we extend these coordinates away from the singularity by keeping $(w,x,y)$ constant along geodesics normal to the hypersurface, and denoting by~$s$ the proper time or distance away from the singularity.
We choose $\sgn s=-\sgn(\rplus+\rminus)=-\sgn(A)$ so that in the first diamond $s$ has the same orientation as physical time.
Given the plane symmetry, these geodesics have a constant value of $x,y$ but variable values of $\rplus,\rminus$, and the latter are solutions of the differential equations
\[
\rplus_{ss} = -2\upsilon_\rplus \rplus_s^2 , \qquad
\rminus_{ss} = -2\upsilon_\rminus \rminus_s^2 ,
\]
where $s$-derivatives are taken at constant~$w$.  We find the change of coordinates
\bse\label{coordinates-Gauss-null}
\be
-\rplus+\rminus = w + o\bigl(|s|^{1/(1+\upsilon_\sharp(w))}\bigr) , \qquad
\rplus+\rminus = - (\sgn s)\bigl((1+\upsilon_\sharp(w))e^{-\upsilon_\flat(w)}|s|\bigr)^{1/(1+\upsilon_\sharp(w))} + o\bigl(|s|^{1/(1+\upsilon_\sharp(w))}\bigr),
\ee
and the inverse change of coordinates
\be
w=-\rplus+\rminus+o\bigl(|\rplus+\rminus|\bigr), \qquad
s= \frac{-\sgn(\rplus+\rminus)\,e^{\upsilon_\flat(-\rplus+\rminus)}}{1+\upsilon_\sharp(-\rplus+\rminus)}|\rplus+\rminus|^{1+\upsilon_\sharp(-\rplus+\rminus)} + o\bigl(|\rplus+\rminus|^{1+\upsilon_\sharp(-\rplus+\rminus)} \bigr),
\ee
\ese
where we used the proper time/distance as the affine parameter~$s$ of the geodesic.
Then we compute $4e^{2\upsilon}d\rplus d\rminus$ in terms of the coordinates $s,w$.
First, $\upsilon = \upsilon_\sharp\log|\rplus+\rminus|+\upsilon_\flat+o(1) = \frac{1}{1+\upsilon_\sharp}\bigl(\upsilon_\flat+\upsilon_\sharp\log\bigl((1+\upsilon_\sharp)|s|\bigr)\bigr)+o(1)$, from which one checks $\del_s(\rplus+\rminus)=e^{-\upsilon}(1+o(1))$, $\del_s(-\rplus+\rminus)=e^{-\upsilon}o(1)$, $\del_w(\rplus+\rminus)=O(|s|\log|s|)=o(1)$, and $\del_w(-\rplus+\rminus)=1+o(1)$.
This yields
\[
\aligned
4e^{2\upsilon}d\rplus d\rminus
& = e^{2\upsilon}\bigl( d(\rplus+\rminus)^2 - d(-\rplus+\rminus)^2\bigr)
= e^{2\upsilon}\bigl( e^{-\upsilon}(1+o(1))ds+o(1)dw\bigr)^2
- e^{2\upsilon}\bigl( e^{-\upsilon}o(1)ds+(1+o(1))dw\bigr)^2 \\
& = ds^2-e^{2\upsilon}dw^2+o(1)ds^2+o(e^\upsilon)dwds+o(e^{2\upsilon})dw^2 .
\endaligned
\]
Altogether, the metric is asymptotically $g^{(4)} \simeq \pm ds^2 + g_*(s)$ as $s\to 0$ (on one particular side of the singularity), with
\bel{gst-from-Fuchs}
g_*(s) = \mp \bigl((1+\upsilon_\sharp)e^{\upsilon_\flat/\upsilon_\sharp}|s|\bigr)^{2\upsilon_\sharp/(1+\upsilon_\sharp)} dw^2
+ \frac{1}{2}e^{2\psi_\flat}\bigl((1+\upsilon_\sharp)e^{-\upsilon_\flat}|s|\bigr)^{(1+ 2\psi_\sharp)/(1+\upsilon_\sharp)} dx^2
+ \frac{1}{2}e^{-2\psi_\flat}\bigl((1+\upsilon_\sharp)e^{-\upsilon_\flat}|s|\bigr)^{(1-2\psi_\sharp)/(1+\upsilon_\sharp)} dy^2 .
\ee

\paragraph{Singularity scattering data from Fuchsian data.}

Starting from a Fuchsian data set $(\phi_\sharp,\phi_\flat,\psi_\sharp,\psi_\flat,\upsilon_\sharp,\upsilon_\flat)$, which we recall obeys~\eqref{alphasharpflat}, namely $\upsilon_\sharp=4\pi\phi_\sharp^2+\psi_\sharp^2-1/4$ and $\upsilon_\flat'=8\pi\phi_\flat'\phi_\sharp+2\psi_\flat'\psi_\sharp$, we have thus obtained the corresponding singularity data set $(g,K,\phi_0,\phi_1)$, defined as the limit~\eqref{singularity-data-set}:
\bel{sharpflat-to-01}
\aligned
g & = g_*\bigl|_{s=1} , \qquad &
K & = \diag\biggl(\frac{\upsilon_\sharp}{1+\upsilon_\sharp},\frac{1/2 + \psi_\sharp}{1+\upsilon_\sharp},\frac{1/2 - \psi_\sharp}{1+\upsilon_\sharp}\biggr) \quad \text{in the basis $\del_w,\del_x,\del_y$,}\\
\phi_0 & = \frac{\phi_\sharp}{1+\upsilon_\sharp} , &
\phi_1 & = \frac{\phi_\sharp}{1+\upsilon_\sharp} \bigl(\log(1+\upsilon_\sharp)-\upsilon_\flat\bigr) + \phi_\flat ,
\endaligned
\ee
with $g_*$ given above in~\eqref{gst-from-Fuchs}.
Observe that the denominators are non-zero by construction since $\upsilon_\sharp\geq -1/4$.
The Kasner exponents (eigenvalues of~$K$) sum to $\Tr K=1$, as they should, and one also checks that the Hamiltonian constraint is satisfied as well since
$1 - \Tr K^2 - 8\pi\phi_0^2 = 2(\upsilon_\sharp+1/4-4\pi\phi_\sharp^2-\psi_\sharp^2) / (1+\upsilon_\sharp)^2 = 0$.
Strikingly, the momentum constraint $\nabla_a K^a_b=8\pi\phi_0\del_b\phi_1$ is also obeyed, as we work out in the proof of Lemma~\ref{lem:FfromADM}, later on.  This relies on numerous cancellations based on~\eqref{alphasharpflat} and on the precise form of the metric.
As a side comment we note the surprisingly simple relation $2(1-k_1)\abs{g}^{1/2}=1$,
which holds in the $w,x,y$ coordinates defined above along the singularity.
We summarize our observations in a lemma.

\begin{lemma}[Fuchsian-to-ADM map on the singularity]
\label{lem:ADMfromF}
  Consider a plane-symmetric solution of the Einstein-scalar field system with the expansion~\eqref{asympt} in the canonical null coordinates $(\rplus,\rminus)$ on one side of the singularity $\rplus+\rminus=0$, in terms of functions $(\phi_\sharp,\phi_\flat,\psi_\sharp,\psi_\flat,\upsilon_\sharp,\upsilon_\flat)$ obeying~\eqref{alphasharpflat}.
  In Gaussian coordinates, the solution has the asymptotic form $\phi\simeq\phist$, $g^{(4)}\simeq\epsilon ds^2+\gst$ with $\epsilon=\sgn(-\rplus'\rminus')$ and the asymptotic profile $(\gst,\phist)$ defined by
  \bel{equa:time-asymptotic-profile}
  \gst(s) = \abs{s}^{2\Kmoi}\gmoi , \qquad \Kst(s) = \frac{-1}{s} \Kmoi , \qquad
  \phist(s) = \phi_{0-}\log|s| + \phi_{1-} ,
\ee
 in which the singularity data set is $(g,K,\phi_0,\phi_1)=\SscrADMfromF(\phi_\sharp,\phi_\flat,\psi_\sharp,\psi_\flat,\upsilon_\sharp,\upsilon_\flat)$ with explicit expressions given in~\eqref{sharpflat-to-01}.
  In the coordinates $w,x,y$ defined above, $2(1-k_1)\abs{g}^{1/2}=1$ where $k_1$~is the eigenvalue of~$K$ transverse to symmetry orbits.
\end{lemma}

\paragraph{Plane-symmetric singularity data sets.}

We now construct the inverse map, denoted by~$\SscrFfromADM$, which can of course only be defined on plane-symmetric singularity data sets.
In a coordinate system $(w,x,y)$ adapted to the symmetry, a plane-symmetric singularity data set only depends on the first coordinate and is such that $g$~and $K=\diag(k_1,k_2,k_3)$ are diagonal in the basis corresponding to $w,x,y$.
We restrict ourselves to data sets for which the symmetry orbits are spacelike, namely $g_{22},g_{33}>0$, because the scattering maps of interest to us preserve these signs.
We now show that generic plane-symmetric singularity data sets with $g_{22},g_{33}>0$, specifically those that never take the value $k_1=1$ (hence $k_2=k_3=\phi_0=0$), must take the form~\eqref{sharpflat-to-01} up to reparametrization of~$w$.

Our first step is to note that changing the coordinate~$w$ to any monotonic function $\widetilde{w}=\widetilde{w}(w)$ rescales~$g_{11}$ by $(\del_w\widetilde{w})^2$ and leaves all other components unaffected.
We gauge fix this freedom by enforcing a property obeyed by the metric in~\eqref{sharpflat-to-01}:
\bel{detg-plane}
2(1-k_1)\abs{g}^{1/2} = 1 .
\ee
This is only possible provided there are no points with $k_1=1$ (and $k_2=k_3=\phi_0=0$), or equivalently points where $K$ is orthogonal to the plane-symmetry orbits.
We are then ready to state the following lemma.

\begin{lemma}[ADM-to-Fuchsian map on the singularity]
  \label{lem:FfromADM}
  Consider the set of plane-symmetric singularity data sets $(g,K,\phi_0,\phi_1)$ such that $K$ is nowhere orthogonal to the symmetry orbits, and expressed in coordinates such that~\eqref{detg-plane} is obeyed.
  The map $\SscrADMfromF$ of Lemma~\ref{lem:ADMfromF} takes values in this set, and has an inverse $\SscrFfromADM\colon(g,K,\phi_0,\phi_1)\to(\phi_\sharp,\phi_\flat,\psi_\sharp,\psi_\flat,\upsilon_\sharp,\upsilon_\flat)$ explicitly given as
\bel{01-to-sharpflat}
\aligned
\phi_\sharp & = \frac{\phi_0}{1-k_1} , \qquad
& \phi_\flat & = \phi_1 - \frac{\phi_0}{1-k_1} \log(4g_{22}g_{33})^{1/2} ,
\\
\psi_\sharp & = \frac{k_2-k_3}{2(1-k_1)} , \qquad
& \psi_\flat & = \frac{k_3\log(2g_{22})^{1/2}-k_2\log(2g_{33})^{1/2}}{1-k_1} ,
\\
\upsilon_\sharp & = \frac{k_1}{1-k_1} , \qquad
& \upsilon_\flat & = - \frac{1}{1-k_1}\log(4g_{22}g_{33})^{1/2} - \log(1-k_1) .
\endaligned
\ee
\end{lemma}

\begin{proof}
We readily plug these formulas~\eqref{01-to-sharpflat} into~\eqref{sharpflat-to-01} and simplify them using only $\Tr K=1$.
The singularity data set $(g,K,\phi_0,\phi_1)$ is then exactly reproduced, except for~$g_{11}$, for which one must additionally use~\eqref{detg-plane}.
This does not conclude the proof, as there remains to establish that for any plane-symmetric data set $(g,K,\phi_0,\phi_1)$ the parameters $(\phi_\sharp,\phi_\flat,\psi_\sharp,\psi_\flat,\upsilon_\sharp,\upsilon_\flat)$ given in~\eqref{01-to-sharpflat} are valid Fuchsian data in the sense that $\upsilon_\sharp$ and~$\upsilon_\flat$ obey the relations~\eqref{alphasharpflat}.

We first explain $\upsilon_\sharp=-1/4+4\pi\phi_\sharp^2+\psi_\sharp^2$.  We show that the data sets~\eqref{sharpflat-to-01} admit the most general values of $(\phi_0,k_1,k_2,k_3)$ except for the value $(0,1,0,0)$, which is only obtained in the infinite $\phi_\sharp$ or~$\psi_\sharp$ limit.
The Hamiltonian constraint $8\pi\phi_0^2+k_1^2+k_2^2+k_3^2=1$ defines a $3$-sphere and we consider the stereographic projection with respect to its pole at $k_1=1$, and the inverse projection for which we introduce a notation~$\upsilon_\sharp$:
\[
(\phi_\sharp,\beta,\gamma) = \frac{(\phi_0,k_2,k_3)}{1-k_1} \in\RR^3, \qquad
(\phi_0,k_1,k_2,k_3) = \frac{(\phi_\sharp,\upsilon_\sharp,\beta,\gamma)}{1+\upsilon} , \qquad
1+\upsilon_\sharp = \frac{1}{2}(1+8\pi\phi_\sharp^2+\beta^2+\gamma^2) = \frac{1}{1-k_1}.
\]
The trace condition $\Tr K=1$ translates to $\upsilon_\sharp+\beta+\gamma=1+\upsilon_\sharp$, whose solutions we parametrize as $\beta=1/2+\psi_\sharp$, $\gamma=1/2-\psi_\sharp$.
The resulting parametrization of $(\phi_0,K)$ coincides with the one in \eqref{sharpflat-to-01} and~\eqref{01-to-sharpflat}, and the relation between $\upsilon_\sharp$ and $\phi_\sharp,\beta,\gamma$ reproduces $\upsilon_\sharp=-1/4+4\pi\phi_\sharp^2+\psi_\sharp^2$ given in~\eqref{alphasharpflat}.

To check that $\upsilon_\flat'=8\pi\phi_\sharp\phi_\flat'+2\psi_\sharp \psi_\flat'$ we use the momentum constraint, which in plane symmetry with diagonal $g$ and~$K$ reduces to the following (primes denote $\del_1=\del_w$):
\[
0 = k_1' + (k_1-k_2) \bigl(\log|g_{22}|^{1/2}\bigr)' + (k_1-k_3) \bigl(\log|g_{33}|^{1/2}\bigr)' - 8\pi\phi_0 \phi_1' .
\]
Using~\eqref{sharpflat-to-01} (as we discussed above, these relations hold), we convert to $\phi_\sharp,\phi_\flat,\psi_\sharp,\psi_\flat,\upsilon_\sharp,\upsilon_\flat$ progressively, converting $\phi_0,k_1,k_2,k_3$ only in a second step to keep expressions manageable.  This yields
\begin{align*}
  0 & = k_1' + (k_1-k_2)\, \bigl(\psi_\flat+k_2(\log(1+\upsilon_\sharp)-\upsilon_\flat)\bigr)' \\
  & \phantom{{}= k_1'} + (k_1-k_3)\, \bigl(-\psi_\flat+k_3(\log(1+\upsilon_\sharp)-\upsilon_\flat)\bigr)' - 8\pi\phi_0\bigl(\phi_\flat+\phi_0(\log(1+\upsilon_\sharp)-\upsilon_\flat)\bigr)'
  \\
    & = \frac{\upsilon_\sharp'}{(1+\upsilon_\sharp)^2} - \frac{1}{1+\upsilon_\sharp} \bigl(\log(1+\upsilon_\sharp)-\upsilon_\flat\bigr)' - \frac{2\psi_\sharp \psi_\flat'+8\pi\phi_\sharp\phi_\flat'}{1+\upsilon_\sharp}
      = \frac{\upsilon_\flat'-2\psi_\sharp \psi_\flat'-8\pi\phi_\sharp\phi_\flat'}{1+\upsilon_\sharp} ,
\end{align*}
where we collected separately the terms involving $\log(1+\upsilon_\sharp)-\upsilon_\flat$ and its derivative: the former has coefficient $(k_1-k_2)k_2'+(k_1-k_3)k_3'-8\pi\phi_0\phi_0'=0$ while the latter has coefficient $(k_1-k_2)k_2+(k_1-k_3)k_3-8\pi\phi_0^2=k_1-1=-1/(1+\upsilon_\sharp)$.
We learn that the momentum constraint is equivalent to $\upsilon_\flat'=8\pi\phi_\sharp\phi_\flat'+2\psi_\sharp \psi_\flat'$, precisely as we wanted.

We conclude that the parametrization~\eqref{sharpflat-to-01} gives the most general plane symmetric singularity data set with $g_{22},g_{33}>0$ and $k_1\neq 1$ everywhere, up to a suitable reparametrization of~$w$.
\end{proof}

\subsection{Scattering maps in plane symmetry}
\label{sec:83}

\paragraph{Scattering maps for Fuchsian data.}

The singularity scattering maps $\Sbf\colon(\gmoi,\Kmoi,\phi_{0-},\phi_{1-})\mapsto(\gpoi,\Kpoi,\phi_{0+},\phi_{1+})$ studied in~\cite{LLV-1a}
can finally be translated into maps relating Fuchsian data on the two sides:
\be
\SscrFfromADM\circ\Sbf\circ\SscrADMfromF\colon
(\phi_\sharp^-,\phi_\flat^-,\psi_\sharp^-,\psi_\flat^-,\upsilon_\sharp^-,\upsilon_\flat^-)\mapsto(\phi_\sharp^+,\phi_\flat^+,\psi_\sharp^+,\psi_\flat^+,\upsilon_\sharp^+,\upsilon_\flat^+) .
\ee
In general, the scattering map~$\Sbf$ does not preserve the condition~\eqref{detg-plane} that we used to gauge-fix reparametrizations of the coordinate~$w$ along the singularity.
We only defined~$\SscrFfromADM$ when this condition is obeyed, so the composition $\SscrFfromADM\circ\Sbf\circ\SscrADMfromF$ implicitly includes a coordinate change $w_-\mapsto w_+(w_-)$.
The change of coordinates is characterized by
\bel{wpoi-wmoi}
\frac{dw_+}{dw_-} = 2(1-k_{1+})\abs{g_{+,\Sbf}}^{1/2} = \frac{2(1-k_{1+})\abs{g_{+,\Sbf}}^{1/2}}{2(1-k_{1-})\abs{\gmoi}^{1/2}},
\ee
where we used $2(1-k_{1-})\abs{\gmoi}^{1/2}=1$ to write an expression that is easier to evaluate for concrete scattering maps.
Here, $g_{+,\Sbf}$~denotes the metric obtained by applying~$\Sbf$ and before performing the change of coordinates.
In the new coordinate~$w_+$, the metric~$\gpoi$ obeys by construction $2(1-k_{1+})\abs{\gpoi}^{1/2}=1$.

Throughout this paper we concentrate on so-called ultralocal scattering maps~$\Sbf$, as introduced in~\cite{LLV-1a}, namely junction conditions such that data $(\gpoi,\Kpoi,\phi_{0+},\phi_{1+})$ at a point depends on data $(\gmoi,\Kmoi,\phi_{0-},\phi_{1-})$ at the same point along the singularity, and not on any spatial derivatives.
The resulting $\SscrFfromADM\circ\Sbf\circ\SscrADMfromF$ are then ultralocal in the same sense provided one takes into account the change of variables: $(\phi_\sharp^+,\phi_\flat^+,\psi_\sharp^+,\psi_\flat^+,\upsilon_\sharp^+,\upsilon_\flat^+)(w_+(w_-))$ only depends on the value of $(\phi_\sharp^-,\phi_\flat^-,\psi_\sharp^-,\psi_\flat^-,\upsilon_\sharp^-,\upsilon_\flat^-)(w_-)$ and not on its derivatives.
Explicit expressions (in terms of Fuchsian data) for the general maps $\Siso$ and~$\Sani$ classified in 
Theorem 5.4 of \cite{LLV-1a}
are easy to write but unwieldy and unenlightening, so we refrain from writing them in general.

\paragraph{Three characterizations of momentum-preserving maps.}

Since our aim is merely to illustrate the use of scattering maps to construct a spacetime globally beyond singularities, rather than being fully general, we concentrate on the class of {\sl momentum-preserving ultralocal maps}.
Among ultralocal maps this class can be characterized in three ways.
\bei
\item {\bf Maps such that $\Kpoi=\Kmoi$.}  We took this to be the definition of momentum-preservation, and we classified such maps in Theorem 5.4 of \cite{LLV-1a}.
In the notations there, they are $\Sani_{\Phi,c,\eta}$ with $\eta=+1$ and $\Phi=a(\phi_0,\phi_1+f(\theta,\phi_0))$ with $a=\pm 1$ for some suitably regular function $f=f(\theta,\phi_0)$.
We give them explicitly in~\eqref{junction-for-planes}.

\item {\bf Maps that are shift-covariant, quiescence-preserving, and whose inverse also is.}
A scattering map is \emph{shift-covariant} if it respects the symmetry of the wave equation under constant shifts of~$\phi$, in the sense that such a constant shift on one side is mapped to a constant shift on the other side.
By Theorem~5.4 of~\cite{LLV-1a}, 
shift-covariant maps have $\phi_{0+}/\rpoi=\eta a^{-1}\phi_{0-}/\rmoi$ and $\phi_{1+}=a(\phi_{1-}+f)$ for some sign~$\eta$, some non-zero $a\in\RR$ and function $f=f(\thetamoi,\phi_{0-})$.
A scattering map is \emph{quiescence-preserving} if it maps a positive-definite~$K_-$ to a positive-definite~$K_+$ (this is a natural condition when taking into account gradient instabilities).
In the present case, this imposes $|a|\leq 1$ and $\eta=+1$.
The inverse of such a scattering map has the inverse value $a\to 1/a$, and imposing that this inverse is quiescence-preserving requires $|1/a|\leq 1$, hence $a=\pm 1$, which leads to $\phi_{0+}=\phi_{0-}$ so $\rpoi=\rmoi$ so $(K_+-\delta/3)=\eta(\rpoi/\rmoi)(K_--\delta/3)=(K_--\delta/3)$, which is the statement of momentum-preservation.
The converse is easy to check.

\item {\bf Invertible maps such that $dw_+/dw_-$ is constant.}
Among the maps $\Siso$ and $\Sani$ listed in Theorem~5.4 of~\cite{LLV-1a}, invertibility rules out~$\Siso$.
For~$\Sani$, we rely on~(5.12g) in \cite{LLV-1a}, that is, 
\bel{Sani-gpoi}
\gpoi = c^2 \biggl(\frac{\rmoi}{\rpoi}\biggr)^{2/3}
\exp\biggl(16\pi\epsilon\xi\cos\Theta_-
-16\pi\epsilon\bigl( \del_{\thetamoi} \xi + \frac{\phi_{0+}}{\rpoi} \del_{\thetamoi} \phi_{1+} \bigr)
\sin\Theta_-\biggr) \gmoi ,
\ee
where $\epsilon=\pm 1$ and $c>0$, while $\phi_{0+}$ and $\phi_{1+}$ are functions of $(\thetamoi,\phi_{0-},\phi_{1-})$ specified by the singularity scattering map.
We compute the change of coordinates~\eqref{wpoi-wmoi} by working out
that the volume factor scales as $\abs{g_{+,\Sbf}}^{1/2}/\abs{\gmoi}^{1/2}=c^3\rmoi/\rpoi$, then writing Kasner exponents $k_{1\pm}=\frac{1}{3}+\frac{2}{3}r_{\pm}\cos\theta_{\pm}$ with $\cos\theta_+=\epsilon\cos\theta_-$.  We thus want the following to be some constant~$C$:
\bel{wpoi-wmoi-Sani}
\frac{dw_+}{dw_-} = c^3 \frac{1/\rpoi-\epsilon\cos\theta_-}{1/\rmoi-\cos\theta_-} = C .
\ee
For each value of~$\theta_-$ this expresses $\rpoi$ as a function of~$\rmoi$.
Invertibility of the scattering map requires this function to be a bijection from $[0,1]$ to itself.
Since $\rmoi=0$ gives $\rpoi=0$ the bijection must be non-decreasing and map $\rmoi=1$ to $\rpoi=1$.
We learn $c^3(1-\epsilon\cos\theta_-) = C(1-\cos\theta_-)$ for all~$\theta_-$, hence $C=c^3$ and $\epsilon=+1$.
Plugging this back into~\eqref{wpoi-wmoi-Sani} yields $\rpoi=\rmoi$ hence $\Kpoicirc=\epsilon(\rpoi/\rmoi)\Kmoicirc=\Kmoicirc$, which is the statement of momentum-preservation.
The converse is easy to check.
\eei

The third characterization of momentum-preserving maps may seem like an ad hoc condition, so let us justify why it is natural.
We impose invertibility, which is a very sensible requirement in the timelike case: it should be possible to express the junction condition both as a scattering map $\Sbf\colon(\gmoi,\Kmoi,\phi_{0-},\phi_{1-})\mapsto(\gpoi,\Kpoi,\phi_{0+},\phi_{1+})$ and as a map~$\Sbf^{-1}$ in the other direction.
In contrast, in the spacelike case, microscopic physics might be not be time-reversible and may involve dissipation phenomena, so that the scattering maps would not need to be invertible if we only had spacelike singularity hypersurfaces.
The condition that $dw_+/dw_-$ is constant, namely that $w_+$ depends linearly on~$w_-$, is motivated by the following observation.

\paragraph{Wave-equation for the areal coefficient.}

The Einstein equation~\eqref{eqxi} states that the area coefficient $A=\pm e^{2a}$ obeys a wave equation
\bel{Awave-everywhere}
A_{\uplus\uminus}=0 ,
\ee
valid away from singularity hypersurfaces $A=0$.
We now explain that for momentum-preserving scattering maps with $c=1$ (see below for $c\neq 1$), the wave equation is obeyed everywhere.
On each regularity domain the wave equation implies that $2A=\rplus(\uplus)+\rminus(\uminus)$ for some functions $\rplus,\rminus$, and these functions may be discontinuous at singularity hypersurfaces.
Our only task is thus to show that they are continuous.

For momentum-preserving scattering maps with $c=1$, the same coordinate $w=w_+=w_-$ can be used on both sides of the singularity.
Converting from the Gaussian coordinates adapted to the ADM formalism to null-coordinates adapted to the global evolution problem through~\eqref{coordinates-Gauss-null}, we learn that $-\rplus+\rminus$ is continuous across the singularity.
Since $\rplus+\rminus\to 0$ at the singularity on both sides of the singularity, we learn that $\pm\rplus+\rminus\to 0$ are both continuous at the singularity, hence $\rplus$ and~$\rminus$ also are.
Altogether, there are globally defined functions $\rplus,\rminus$ such that $2A=\rplus(\uplus)+\rminus(\uminus)$, namely the wave equation~\eqref{Awave-everywhere} holds everywhere, as we announced and used in Section~\ref{section---7} when decomposing spacetime into regularity domains $\{A>0\}$ and $\{A<0\}$ separated by singularity hypersurfaces $\{A=0\}$.

The constant~$c$ in the maps~\eqref{junction-for-planes} of interest to us scales the metric by $c^2$ on one side of the singularity.
Thanks to the fact that (in our collision problem) space-time is naturally bi-partitioned along singularity hypersurfaces according to the sign of~$A$, as depicted for instance in Figure~\ref{fig:1.2}, we can rescale the metric by $c^{-2}$ in all regions with $A<0$ to eliminate the parameter from scattering maps.
This reduces all momentum-preserving ultralocal scattering maps to the case $c=1$.
Likewise, one can normalize $a=+1$ (by flipping the sign of~$\phi$ in regions where $A<0$, if $a=-1$).

\paragraph{Momentum-preserving maps for Fuchsian data.}

As we just discussed, momentum-preserving ultralocal scattering maps can be normalized to set $a=+1$ and $c=1$.
We work out
\bse\label{plane-scattering-choice}
\be
\aligned
\phi_\sharp^+ & = \phi_\sharp^- , &
\psi_\sharp^+ & = \psi_\sharp^- , &
\upsilon_\sharp^+ & = \upsilon_\sharp^- , \\
\phi_\flat^+ & = \phi_\flat^- + \phi_\sharp^- \beta_1 + f , \quad &
\psi_\flat^+ & = \psi_\flat^- + \psi_\sharp^- \beta_1 + \beta_2, \quad &
\upsilon_\flat^+ & = \upsilon_\flat^- + (1 + \upsilon_\sharp^-) \beta_1 ,
\endaligned
\ee
where $\xi,\beta_1,\beta_2$ are functions of $(\thetamoi,\phi_{0-})$, or equivalently $(\phi_\sharp^-,\upsilon_\sharp^-)$, given by
\be
\aligned
\beta_1 & = -2\log c + 8\pi\cos\thetamoi\,\xi - 8\pi\sin\thetamoi\,\del_{\thetamoi}\bigl(\frac{\phi_{0-}}{r(\phi_{0-})} f + \xi\bigr) , 
\qquad
&& \xi = - \int_{-1/\sqrt{12\pi}}^{\phi_{0-}} \del_y f(\thetamoi,y)\,\frac{ydy}{r(y)}, 
\\
\beta_2 & = -4\pi\sqrt{3}\bigl(\sin\thetamoi\,\xi + \cos\thetamoi\,\del_{\thetamoi}\bigl(\frac{\phi_{0-}}{r(\phi_{0-})} f + \xi\bigr)\bigr) ,
\endaligned
\ee
and we recall for completeness how $\phi_{0-},\theta_-$ are related to Fuchsian data:
\be
\phi_{0-} = \frac{\phi_\sharp^-}{1+\upsilon_\sharp^-} , \qquad
r_- = r(\phi_{0-}) = \sqrt{1-12\pi\phi_{0-}^2} , \qquad
\cos\theta_- = \frac{\upsilon_\sharp^--1/2}{r_-(1+\upsilon_\sharp^-)} , \qquad
\sin\theta_- = \frac{-\sqrt{3}\,\psi_\sharp^-}{r_-(1+\upsilon_\sharp^-)} .
\ee
\ese
The junction condition has a triangular structure, in which $(\phi_\sharp,\psi_\sharp,\upsilon_\sharp)$ variables are the same on both sides, and the subleading variables $(\phi_\flat,\psi_\flat,\upsilon_\flat)$ simply jump by a nonlinear function of the leading ones.
In addition, the $\upsilon$ variables {\sl do not appear} in the expressions for $\phi$ and~$\psi$ variables, so we learn that the restricted data $(\phi_\sharp^+,\phi_\flat^+,\psi_\sharp^+,\psi_\flat^+)$ can be determined in terms of the corresponding restricted data on the other side.  This should be contrasted with Remark~\ref{rem:ups-useless}, above.
While this simplification of the problem is not crucial, it is actually rather useful in practice, as it lets us concentrate on the two main variables $\psi,\phi$ before solving for the last metric coefficient~$\upsilon$ (or equivalently~$\omega$).

\paragraph{Causality of momentum-preserving scattering maps.}

When solving the initial value problem with timelike singularity hypersurface in the next section, we discover that the jump in $(\phi_\flat,\psi_\flat)$ prescribed by~\eqref{plane-scattering-choice} can be determined from the initial data.
We find that the initial value problem is well posed if and only if the map~\eqref{causal-condition} below, which controls the jump of $(\phi_\flat,\psi_\flat)$, is bijective so that $(\phi_\sharp^-,\psi_\sharp^-)$ can be determined from this jump.

\begin{definition}[Causality for momentum-preserving ultralocal scattering maps]\label{def:causal}
  A momentum-preserving ultralocal scattering map $\Sbf=\Sani_{(\phi_0,\phi_1+f),c,+}$ (with $c=1$) determined by a periodic function $f=f(\theta,\phi_0)$ is {\bf causal} if the map
  \bel{causal-condition}
  (\phi_\sharp,\psi_\sharp) \mapsto (\phi_\sharp \beta_1 + f , \psi_\sharp \beta_1 + \beta_2)
  \ee
  is bijective, where $\beta_1,\beta_2$ are constructed from~$f$ through~\eqref{plane-scattering-choice}.
\end{definition}

\begin{example}
  Any affine map $f(\theta,\phi_0)=b\phi_0+e$ for $b\neq 0$ and $e\in\RR$ gives a causal scattering map.
  Indeed, one computes $\phi_\sharp\beta_1+f=(2b/3)\phi_\sharp+e$ and $\psi_\sharp \beta_1 + \beta_2=(2b/3)\psi_\sharp$, so that the map~\eqref{causal-condition} is then simply a bijective rescaling by $2b/3\neq 0$.
\end{example}


\section{Building a cyclic spacetime one diamond at a time}
\label{section---9} 

\subsection{Global solution of the plane collision problem}

\paragraph{Main statement.}

We are now in a position to provide a proof of Theorem~\ref{theo:first} which we will first restate in a more detailed form. 
We have described the incoming gravitational data and our choice of null coordinates in Section~\ref{section----6}, and presented the equations in Section~\ref{section---8}: the gravitational field equations away from singularities in Section~\ref{sec:81} and the scattering maps across singularities in Sections~\ref{sec:82} and \ref{sec:83}.
Hence, we can now summarize our formulation of the characteristic initial value problem:
we seek metric coefficients $\omega,A,\psi$ and a matter field~$\phi$
satisfying the field equations and junction conditions within the region $\Mup=\{ 0<\uplus; \, 0<\uminus\}$ when the following data are prescribed on $\Ncalplus_0$, $\Ncalminus_0$, and $\Pcal_0$: 
\bel{101-b}
\aligned
\omega & = 0, \quad &
\psi & = \psiplus_0, \quad &
\phi & = \phiplus_0 \qquad 
&& \text{ on the hypersurface } \Ncalplus_0,
\\
\omega & = 0, \quad &
\psi & = \psiminus_0, \quad &
\phi & = \phiminus_0 \qquad 
&& \text{ on the hypersurface } \Ncalminus_0,
\\ 
A & = 1 \quad &
A_\uplus &= 0, \quad &
A_\uminus & = 0 
\qquad 
&& \text{ on the two-plane } \Pcal_0.
\endaligned
\ee
For the junction conditions, we rely on a momentum-preserving scattering map~$\Sbf$, as described earlier in~\eqref{plane-scattering-choice}.
In the course of our proof below, we encounter a causality condition on~$\Sbf$ without which the evolution problem with a timelike singularity hypersurface would be ill-posed.
Our objective is to establish that the initial data set $(\psiplus_0, \phiplus_0, \psiminus_0, \phiminus_0)$ uniquely determines the unknown metric and matter field and, therefore, the global spacetime geometry in~$\Mup$.

For the sake of simplicity, we henceforth concentrate on data satisfying a non-degeneracy condition,
and throughout our discussion we work with functions that are $C^\infty$ (that is, smooth) away from the singularity hypersurfaces.

\begin{definition}[Generic initial data]
  \label{rem:non-degeneracy}
  An initial data set is said to be {\bf generic} (or non-degenerate) if the functions $\rplus$ and $\rminus$ are never constant on an interval other than the initial one $(-\infty,\uplus_\star^-]$ and $(-\infty,\uminus_\star^-]$, respectively, and if none of the local maxima of~$\rplus$ and $\rminus$ are exactly opposite.
\end{definition}

As noted in Proposition~\ref{prop:AreaFuncSm}, this genericity assumption does not preclude compactly supported initial data or intervals with no incoming radiation.
The two genericity conditions ensure respectively that the singular locus has no null hypersurface component and that it has no codimension~$2$ component, as explained in Section~\ref{sec:diamonds}.
This is needed in order for the constructed spacetime to be a cyclic spacetime in the sense of \cite{LLV-1a}.

\begin{theorem}[Global spacetime geometry for the plane gravitational collision problem]
\label{theorem-plane}
Let $\Sbf$ be a momentum-preserving ultralocal scattering map, that is, $\Sbf=\Sani_{\Phi,c,+}$ with $\Phi(\theta,\phi_0,\phi_1)=(\phi_0,\phi_1+f(\theta,\phi_0))$,
which additionally is causal in the sense of Definition~\ref{def:causal}.
Let $\phiplus_0,\psiplus_0,\phiminus_0,\psiminus_0\colon[0,+\infty)\to\RR$ be smooth data for the matter field and modular parameter along $\Ncalplus_0,\Ncalminus_0$, with $\phiplus_0(0)=\phiminus_0(0)$ and $\psiplus_0(0)=\psiminus_0(0)$, and assume that these data are generic in the sense of Definition~\ref{rem:non-degeneracy}.
\bei 
  
\item Then, the characteristic initial value problem associated with the plane-symmetric initial data~\eqref{101-b} induced on the two wave fronts $\Ncalplus_0,\Ncalminus_0$ admits a global Cauchy development $(\Mup, g, \phi)$ that is a cyclic spacetime based on the scattering map~$\Sbf$ in the sense of  \cite{LLV-1a}.

\item By the definition of a cyclic spacetime, the Einstein field equations are satisfied away from a collection of singularity hypersurfaces, while the junction condition prescribed by~$\Sbf$ holds across each (spacelike or timelike) singularity hypersurface, aside from a $2$-dimensional exceptional locus.
  The curvature of $(\Mup, g)$ generically blows up as one approaches any singularity hypersurface. 

\eei 
\end{theorem}

\paragraph{Construction of the solution.}

Let us first summarize steps we have already achieved to obtain the geometry of~$\Mup$.

\bei
\item {\bf Initial data for the conformal coefficient $\omega$.}
We reiterate Remark~\ref{rem:invariance}: an initial data set with non-zero $\omega$ can be brought to the form~\eqref{101-b} with $\omega=0$ on $\Ncalplus_0,\Ncalminus_0$ by choosing coordinates $\uplus,\uminus$ to be affine parameters along the initial hypersurfaces.
We assume for definiteness that data for the two incoming waves are prescribed for an infinite range $\uplus\in[0,+\infty)$ of affine parameter (and likewise for~$\uminus$).
If data are only prescribed until some finite value $\uplus=\uplus_{\max}$, for example, our Cauchy development should simply be stopped at that value of $\uplus$ for all~$\uminus$ because data are missing to go further.

\item {\bf Initial data for the area function~$A$.}
As explained in Section~\ref{sec:64}, the prescribed incoming radiation data~\eqref{101-b} on $\Ncalplus_0$ determines a function $\Fplus(\uplus)=|A(\uplus,0)|^{1/2}$ on~$\Ncalplus_0$ by solving the Raychaudhuri equation~\eqref{equa:Ffunctions} $\Fplus_{\uplus\uplus}=-\Eplus\Fplus$ with $\Eplus=\psiplus_\uplus^2 + 4 \pi \phiplus_\uplus^2$, and likewise the data $\phiminus_0,\psiminus_0$ determines $\Fminus=|A|^{1/2}$ on~$\Ncalminus_0$.
The genericity assumption of Definition~\ref{rem:non-degeneracy} states that $\Fplus$ must not be constant on any interval other than the initial segment $(-\infty,\uplus_\star^-)$ before the start of the incoming wave, and likewise for~$\Fminus$.

\item {\bf Areal function~$A$ everywhere.}
These values of~$A$ along $\Ncalplus_0$ and~$\Ncalminus_0$ provide initial data for the wave equation $A_{\uplus\uminus}=0$, which is obeyed everywhere for our choice of scattering map as explained near~\eqref{Awave-everywhere}.  From its global solution $A(\uplus,\uminus) = \Fplus(\uplus)^2 + \Fminus(\uminus)^2 - 1$ one finds a collection of singularity hypersurfaces $\{A=0\}$ studied in Section~\ref{sec:singularities}.  The genericity assumption ensures that $\{A=0\}$ consists of spacelike and timelike hypersurfaces joined at a collection of points in the $(\uplus,\uminus)$ plane, with no null hypersurface.

\item {\bf Decomposition into monotonicity diamonds.}
In Section~\ref{sec:diamonds} we split~$\Mup$ along constant-$\uplus$ or constant-$\uminus$ null rays along which $A_\uplus=2\Fplus(\uplus)\Fplus'(\uplus)$ or $A_\uminus=2\Fminus(\uminus)\Fminus'(\uminus)$ vanish, respectively.
These rays partition the interaction domain $\Mup$ into monotonicity diamonds, which by definition are maximal characteristic domains within which the area coefficients $A_\uplus$ and $A_\uminus$ keep a constant sign.
Under our non-degeneracy assumption this sign is never zero: the sign of $A_\uplus$ is alternatively $\pm 1$ in successive intervals~$\Iplus_i$ of~$\uplus$ while that of $A_\uminus$ alternates in successive intervals~$\Iminus_j$ of~$\uminus$.
Thus, the gradient $\nabla A$ is alternatively timelike and spacelike in neighboring diamonds $\Delta_{ij}=\Iplus_i\times\Iminus_j$, in a checkerboard pattern.
\eei
To construct the metric and matter fields $\phi,\psi,\omega$ in the whole domain~$\Mup$ it is thus sufficient to solve the characteristic initial value problem in each diamond~$\Delta_{ij}$ successively, using values along future boundaries of~$\Delta_{ij}$ as initial data for the diamonds $\Delta_{i+1,j}$ and~$\Delta_{i,j+1}$.
By induction this constructs the spacetime geometry for all values of~$\uplus,\uminus$.

\paragraph{The initial value problem in each diamond.}

Throughout this section we work in a single diamond, hence we can use the local coordinates $(\rplus,\rminus)$, in which the (same) singular wave equation obeyed by the matter field~$\phi$ and modular parameter~$\psi$ takes a canonical form
\bel{eq:3-1-two}
\psi_{\rplus\rminus}  + {\psi_\rminus +  \psi_\rplus \over 2 (\rplus + \rminus)} = 0 .
\ee
This equation does not involve~$\omega$, and the junction condition for $\psi,\phi$ imposed by momentum-preserving scattering maps~\eqref{plane-scattering-choice} also does not involve~$\omega$ (nor its shifted version~$\upsilon$).
We can thus begin with these essential metric and matter fields $\psi,\phi$, whose evolution is {\sl decoupled except} at the singularity hypersurfaces where our scattering map, in general, does introduce some 
non-trivial coupling.
Once $\psi,\phi$ are known, the function $\omega=\omega(\uplus, \uminus)$ is easily obtained.
We summarize in Section~\ref{ssec:wrapup} how our construction yields a cyclic spacetime.

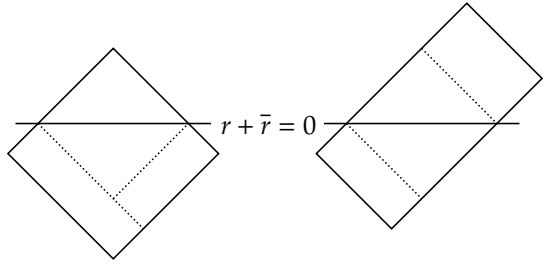
\begin{figure}
  \centering
  \begin{tikzpicture}[semithick]
    \begin{scope}
      \draw[densely dotted] (-1,0) -- (.4,-1.4);
      \draw[densely dotted] (0,-1) -- (1,0);
      \draw (-1.4,-.4) -- (0,1) -- (1.4,-.4) -- (0,-1.8) -- cycle;
      \draw (-1.3,0) -- (1.3,0) node [right] {$\rplus+\rminus=0$};
    \end{scope}
    \begin{scope}[shift={(4.1,0)}]
      \draw[densely dotted] (-1,0) -- (0,-1);
      \draw[densely dotted] (0,1) -- (1,0);
      \draw (-1.4,-.4) -- (.6,1.6) -- (1.6,.6) -- (-.4,-1.4) -- cycle;
      \draw (-1.3,0) -- (1.3,0);
    \end{scope}
  \end{tikzpicture}
  \caption{\label{fig:subdiamonds}{\bf Decomposition of monotonicity diamonds into regular and symmetric diamonds} (along dotted lines). Any diamond~$\Delta$ can be decomposed into smaller diamonds that are either completely on one side of the line $\rplus+\rminus=0$ or are symmetric around it.  Indeed, if the singularity $\rplus+\rminus=0$ cuts~$\Delta$ along the interval $(\rplus_1,\rplus_2)$ of values of~$\rplus$ we can consider the symmetric diamond $(\rplus_1,\rplus_2)\times(-\rplus_2,-\rplus_1)$ contained in~$\Delta$, and adjoin to it regular diamonds to obtain~$\Delta$.}
\end{figure}

As explained and depicted in Figure~\ref{fig:subdiamonds}, each $\Delta_{ij}$ can be cut further into smaller diamonds so that any singularity hypersurface passes through corners of these subdivisions.
Let us denote by $\Delta=(\rplus_1,\rplus_2)\times(\rminus_1,\rminus_2)$ one such smaller diamond, with bounds ordered as $\rplus_1<\rplus_2$ and $\rminus_1<\rminus_2$ (the upper bounds may be infinite).
Either
\bei
\item $\Delta$~is a diamond without singularity, lying entirely on one side of the line $\rplus+\rminus=0$ (so $\rplus_2+\rminus_2\leq 0$ or $0\leq\rplus_1+\rminus_1$), or
\item $\Delta$~is a symmetric diamond, in the sense that the line $\rplus+\rminus=0$ joins two of its vertices (specifically $\rplus_1+\rminus_2=\rplus_2+\rminus_1=0$).
\eei
The whole problem reduces to solving the characteristic initial value problem in these two types of diamonds.

The physical time orientation with respect to $\rplus,\rminus$ depends on the diamond, so we must allow data to be prescribed on any two neighboring sides of~$\Delta$: $(\rplus_1,\rplus_2)\times\{\rminus_i\}$ and
$\{\rplus_j\}\times(\rminus_1,\rminus_2)$ for some $i=1,2$ and $j=1,2$.
For example, the case $i=j=2$, where data are prescribed on the $\rminus=\rminus_2$ and $\rplus=\rplus_2$ sides, is relevant for diamonds such as $\Delta_{1,1}$ where $\rplus',\rminus'<0$ so $\nabla A$ is timelike.
We begin in Section~\ref{sec:9Abel} with an explicit formula for $\psi(\rplus,\rminus)$ (and for~$\phi$) away from singularities, based on the inverse Abel transform of characteristic data.
It is used for each type of diamond.

\bei
\item {\bf Diamond without singularity.}
  Both for timelike and for spacelike $\nabla A$ the explicit formula of Section~\ref{sec:9Abel} solves the characteristic initial value problem under consideration.
\item {\bf Symmetric diamond with a spacelike singularity.}
  In Section~\ref{sec:9sing1} we use the same Abel representation formula to solve in the triangle before the singularity.  We expand the explicit formula along the singularity, apply the singularity scattering map~$\Sbf$ to obtain Fuchsian data on the other side, and provide an explicit formula for $\phi$ (and~$\psi$) after the singularity in terms of this Fuchsian data.
\item {\bf Symmetric diamond with a timelike singularity.}
  In Section~\ref{sec:9sing2} we tackle the hardest case: initial data are prescribed on the past boundary of the diamond, which lies on both sides of the singularity.  We apply the same expansions as before to express Fuchsian data on each side of the singularity in terms of {\sl initial and final data on all boundaries} of the diamond.  Then we write down the relations that the singularity scattering map~$\Sbf$ imposes between these two sets of Fuchsian data.  This translates to equations on the initial and final data, which can be solved explicitly for the final data provided $\Sbf$~is causal in the sense of Definition~\ref{def:causal}.
\eei

\subsection{Abel representation formula}
\label{sec:9Abel}

\paragraph{A preliminary step: Abel transform and its inverse.}
The explicit solutions we find for the wave equation~\eqref{eq:3-1-two} obeyed by $\phi,\psi$ are based on the Abel transform and its inverse, which we introduce now.
In the definition below the restriction on~$\sigma$ ensures that the integrand has at most inverse square root singularities hence is integrable; we also use the standard convention that $\int_\rho^r\coloneqq -\int_r^\rho$ if $\rho>r$.
The interval~$I$ may be infinite, for instance $r_2$~may be~$+\infty$.
The formula below is derived in Appendix~\ref{app:part2}.

\begin{lemmadefinition}
\label{lemmadef:Abel}
Fix an interval $I=(r_1,r_2)$, one of its end points $\rho=r_1$~or~$r_2$, and a parameter $\sigma\in(-\infty,-r_2]\cup[-r_1,+\infty)$.
The {\bf Abel transform} $\Abf_{\rho\sigma}$ of a function $F\colon I\to\RR$ is the function $f\colon I\to\RR$ defined for all $r \in I$ by 
\[
f(r) = \Abf_{\rho\sigma}[F](r)
= \int_{\rho}^{r} {F(s) \over \sqrt{|r-s|\,|\sigma+s|}} \, ds.
\]
The Abel transform can be inverted explicitly, namely for $r\in I$  
\[
F(r)
= \Abf_{\rho\sigma}^{-1}[f](r)
= {\sgn(r-\rho) \over \pi} \sqrt{|\sigma+r|} \, {d \over dr}
    \int_\rho^r \frac{f(s)}{\sqrt{|r-s|}} \, ds .
\]
\end{lemmadefinition}

Note that the Abel transform $f=\Abf_{\rho\sigma}[F]$ of a bounded function~$F$ has a limit $f(r)\to 0$ as $r\to\rho$.
Conversely, if a function~$f$ has a finite non-zero limit $f(\rho)\neq 0$, then its inverse Abel transform~$F$ has an inverse square-root singularity as $r\to\rho$.  This can be seen in the inverse Abel transform of a constant, which is explicitly $\Abf_{\rho\sigma}^{-1}[1](r) = \frac{1}{\pi} \sgn(r-\rho) \sqrt{|\sigma+r|/|r-\rho|}$.
To avoid such singularities below when applying this lemma, we systematically shift the function~$f$ by its value at~$\rho$ and treat the constant part separately as an overall shift of the solution.

\paragraph{Abel representation formula away from singularities.}
Consider a diamond $\Delta=(\rplus_1,\rplus_2)\times(\rminus_1,\rminus_2)$ from the decomposition explained earlier.
The matter field~$\phi$ and the modular parameter~$\psi$ obey the same wave equation~\eqref{eq:3-1-two}, which is nothing but the classical {\bf Euler-Poisson-Darboux equation} with exponent $1/2$, and is singular along the line $\rplus+\rminus=0$.
We consider solutions that are regular away from the hypersurface, but can become singular as one approaches the hypersurface.
We introduce here the key formula that parametrizes solutions to~\eqref{eq:3-1-two} in a connected component~$D$ of $\Delta\setminus\{\rplus+\rminus=0\}$ in terms of data prescribed on two null boundaries $\rminus=\rminus_i$ (for $i=1$~or~$2$) and $\rplus=\rplus_j$ (for $j=1$~or~$2$).

In a diamond with a singularity (which by assumption passes through two corners of~$\Delta$), our explicit formula only applies when data are prescribed on two boundaries that lie on the same side of the singularity, namely provided $\rplus_j+\rminus_i\neq 0$.  The solution is then defined in the domain
\[
D = \Delta\cap \bigl\{(\rplus,\rminus)\bigm|\sgn(\rplus+\rminus)=\sgn(\rplus_j+\rminus_i)\bigr\} .
\]
For a spacelike singularity the formulas provide the solution in a triangle before the singularity in terms of data prescribed on past boundaries of the diamond.
For a timelike singularity the formulas are not directly applicable to the evolution problem since the past boundary of~$\Delta$ lies on both sides of the singularity, but we use them as an intermediate step.
In a diamond without singularity, the four cases of $(i,j)$ are relevant for the evolution problem depending on signs of $\rplus',\rminus'$, and in all four cases $D=\Delta$.

We solve in~$D$ the following Goursat problem with prescribed boundary data $\psiplusG,\psiminusG$ (that must take equal value $\psiplusG(\rplus_j)=\psiminusG(\rminus_i)=\psinot$ at the common corner):
\bel{eq:ivp}
\Lcal  \psi \coloneqq
\psi_{\rplus\rminus}  + {\psi_\rminus +  \psi_\rplus \over 2 (\rplus + \rminus)} 
= 0 \qquad  \text{for all } (\rplus, \rminus) \in D,
\qquad \quad  
\psi |_{\rminus = \rminus_i} = \psiplusG,
\qquad
\psi |_{\rplus = \rplus_j} = \psiminusG .
\ee

\begin{lemma}[Abel representation formula for the Goursat problem]
\label{lemma-Abel}
The solution to the characteristic initial value problem~\eqref{eq:ivp} within~$D$
admits the representation formula
\bel{gen-form}
\psi(\rplus, \rminus)
= \psinot + \int_{\rplus_j}^{\rplus} {\Psiplus(s) \over \sqrt{|\rplus - s|\,|\rminus +s|}} \, ds
+ \int_{\rminus_i}^\rminus {\Psiminus(s)\over \sqrt{|\rplus + s| \, |\rminus - s|}} \, ds,
\qquad
(\rplus, \rminus)  \in D ,
\ee
in which $\Psiplus = \Abf_{\rplus_j\rminus_i}^{-1}[\psiplusG-\psinot] \colon [\rplus_1,\rplus_2] \to \RR$ and $\Psiminus = \Abf_{\rminus_i\rplus_j}^{-1}[\psiminusG-\psinot] \colon [\rminus_1,\rminus_2] \to \RR$ are inverse Abel transforms of the Goursat data $\rplus\mapsto\psiplusG(\rplus)-\psinot$ and $\rminus\mapsto\psiminusG(\rminus)-\psinot$.
\end{lemma} 

\begin{proof}
  The condition that $\rplus+\rminus$ keeps a constant sign throughout~$D$ ensures that the integrands in~\eqref{gen-form} have at most inverse square root singularities hence remains integrable.
  It is easy to check using Lemma~\ref{lemmadef:Abel} that the proposed solution~\eqref{gen-form} takes the correct initial values, as we now show.
  Because $\psiminusG-\psinot$ vanishes at~$\rminus_i$, its inverse Abel transform~$\Psiminus$ remains bounded at~$\rminus_i$, so that the second integral in~\eqref{gen-form} vanishes at~$\rminus_i$.
  Then we compute the first integral:
  \[
  \psi(\rplus, \rminus_i) - \psinot
  = \int_{\rplus_j}^{\rplus} {\Psiplus(s) \over \sqrt{|\rplus - s|\,|\rminus_i +s|}} \, ds
  = \Abf_{\rplus_j\rminus_i}[\Psiplus](\rplus)
  = \Abf_{\rplus_j\rminus_i}\bigl[\Abf_{\rplus_j\rminus_i}^{-1}[\psiplusG-\psinot]\bigr](\rplus) = \psiplusG(\rplus) - \psinot ,
  \]
  and likewise $\psi(\rplus_j, \rminus) = \psinot + \Abf_{\rminus_i\rplus_j}[\Psiminus](\rminus) = \psiminusG(\rminus)$.
  For the proof that $\psi$ given in~\eqref{gen-form} solves the equation~\eqref{eq:3-1-two}, we refer to Appendix~\ref{app:part2}.
\end{proof}

\paragraph{Diamond without singularity.}

For a diamond~$\Delta$ that is entirely on one side of $\rplus+\rminus=0$, the Abel representation formula of Lemma~\ref{lemma-Abel} yields $\psi(\rplus,\rminus)$, and likewise $\phi(\rplus,\rminus)$, explicitly in terms of available initial data.

\subsection{Diamond with a spacelike singularity hypersurface}
\label{sec:9sing1}

We solve here the initial value problem in a diamond that is split into two triangles by a spacelike singularity.
This is in particular the case for the first region of interaction, between the wave-front intersection~$\Pcal_0$ and the first singularity.
Within the first triangle of the first diamond, the geometry resulting from colliding gravitational waves (vacuum Einstein equations) was first solved by Szekeres~\cite{Szekeres70} and further analyzed by Yurtsever~\cite{Yurtsever88} in terms of a (generalized) Kasner behavior near the spacelike singular hypersurface.
The subject was taken up again \cite{FKV,BV1} (after adding the dilaton of string theory) in the context of pre-Big Bang cosmology as an example of inhomogeneous initial conditions that naturally lead to dilaton-driven inflation in the string frame.
Reference~\cite{FKV} dealt directly with the problem in the Einstein frame and then converted the asymptotic solutions to the string frame.
In~\cite{BV1} the problem was studied directly in the string frame, with generic  $(d-1)$-dimensional planar symmetry and in the presence of other massless fields appearing in string theory.
In both papers only half of the Fuchsian data (i.e.\ the coefficients of the singular terms) were computed.
Because of the presence of $(d-1)$ abelian isometries, the equations of motion are endowed with a large global symmetry allowing to construct pair of duality (and time-reversal)-related solutions which, in the spirit of the pre-Big Bang scenario, should be joined together at the space like singular hypersurface.

We have studied this case in detail and found that, by limiting the matching to solutions related by the exact $(\ZZ/2\ZZ)^{d-1}$ symmetry it is {\it not} possible to satisfy our classification of consistent singularity maps. In particular, the shear coefficients are not uniformly rescaled through the map. However, one can argue that as one approaches the singularity
dependence upon the $z = \frac{1}{2}(\rminus - \rplus)$ coordinate becomes subleading and, in the limit, can be totally neglected. As a result the $(\ZZ/2\ZZ)^{d-1}$ symmetry gets enhanced to a full $(\ZZ/2\ZZ)^d$ group, allowing also for the reversal of the string-frame Kasner exponent in the $z$~direction. We have checked that, in analogy with the homogeneous case discussed in \cite{LLV-1a},
the scattering map corresponding to the reversal of all $d$~Kasner exponents does fall in the general classification scheme of this paper. This shows, once more, the predictive power of our classification of consistent singularity maps.

We treat for definiteness the case where physical time flows from positive to negative $\rplus,\rminus$, which happens in diamonds where $\rplus',\rminus'<0$, such as the first diamond.  (The other case $\rplus',\rminus'<0$ is mapped to it by the symmetry $(\rplus,\rminus)\to(-\rplus,-\rminus)$.)
Our convention is summarized as follows: 
\bel{diamond-spacelike}
\Delta = (\rplus_1,\rplus_2)\times(\rminus_1,\rminus_2) , \qquad
\rplus_1 + \rminus_2 = \rplus_2 + \rminus_1 = 0 , \qquad\qquad
\mathtikz{
  \draw (1.3,0) -- (0.3,1) node [above=-3pt] {\scriptsize $(\rplus_1,\rminus_1)$} -- (-.7,0) -- (0.3,-1) node [below=-3pt] {\scriptsize $(\rplus_2,\rminus_2)$} -- cycle;
  \draw (-1,0) -- (1.6,0) node [right,yshift=1pt] {$\rplus+\rminus=0$.};
  \draw[->,densely dashed] (-.5,.5) -- (.85,-.85) node [below right=-3pt]{$\rplus$};
  \draw[->,densely dashed] (.8,.8) -- (-.55,-.55) node [below left=-4pt]{$\rminus$};
}
\ee
Initial data $\psiplusG^-,\psiminusG^-$ are prescribed on the past null boundaries $\rminus=\rminus_2$ and $\rplus=\rplus_2$, respectively.
(The $-$~superscripts indicate $-\sgn A$, in accordance with our orientation convention for singularity hypersurfaces.)
The Abel representation formula of Lemma~\ref{lemma-Abel} yields $\phi,\psi$ in the bottom triangle in terms of these initial data.
To continue, we state here a two-term expansion in terms of Fuchsian data $(\phi_\sharp^-,\phi_\flat^-,\psi_\sharp^-,\psi_\flat^-)$ near the singularity, whose derivation is postponed to Appendix~\ref{app:part2}.

\begin{lemma}[Fuchsian data from Goursat data]
\label{prop:evol-char}
Under the conditions of Lemma~\ref{lemma-Abel} with data $\psiplusG^-$ and~$\psiminusG^-$ prescribed along the boundaries $\rminus=\rminus_2$ and $\rplus=\rplus_2$, respectively, with $\psiplusG^-(\rplus_2)=\psiminusG^-(\rminus_2)=\psinot$, one has
\[
\psi(\rplus, \rminus) = \psi_\sharp^-(-\rplus+\rminus) \log|\rplus+\rminus| + \psi_\flat^-(-\rplus+\rminus) + o(1) \qquad \text{as $\rplus+\rminus\to 0^+$,}
\]
with $\psi_\sharp^-,\psi_\flat^-\colon(-2\rplus_2,2\rminus_2)\to\RR$ given in terms of $\Psiplus^-=\Abf_{\rplus_2\rminus_2}^{-1}[\psiplusG^--\psinot]$ and $\Psiminus^-=\Abf_{\rminus_2\rplus_2}^{-1}[\psiminusG^--\psinot]$ as
\[
\aligned
\psi_\sharp^-(2z) & = \Psiplus^-(-z) + \Psiminus^-(z) ,
\\
\psi_\flat^-(2z)
& = \psinot - \Psiplus^-(-z) \log\bigl(4(\rplus_2+z)\bigr)
+ \int_{-z}^{\rplus_2} \frac{\Psiplus^-(-z)-\Psiplus^-(s)}{z+s} ds
- \Psiminus^-(z) \log\bigl(4(\rminus_2-z)\bigr)
+ \int_{z}^{\rminus_2} \frac{\Psiminus^-(z)-\Psiminus^-(s)}{s-z} ds .
\endaligned
\]
\end{lemma}

Next, we apply the momentum-preserving scattering map~$\Sbf$ of~\eqref{plane-scattering-choice} (with the normalization $a=+1$ and $c=1$) to the data $(\phi_\sharp^-,\phi_\flat^-,\psi_\sharp^-,\psi_\flat^-)$ coming from Lemma~\ref{prop:evol-char}.
This gives new Fuchsian data $(\phi_\sharp^+,\phi_\flat^+,\psi_\sharp^+,\psi_\flat^+)$ that we use as initial data to solve in the second triangle of the diamond~\eqref{diamond-spacelike}.
We seek a function $\psi(\rplus,\rminus)$ on the domain $\Delta\cap\{\rplus+\rminus<0\}$ solving the Fuchsian problem
\bel{eq:fuchsvp}
\aligned
\psi_{\rplus\rminus} & + {\psi_\rminus +  \psi_\rplus \over 2 (\rplus + \rminus)} = 0
& & \text{for all } (\rplus, \rminus) \in \Delta\cap\{\rplus+\rminus<0\},
\\
\psi(\rplus,\rminus) & = \psi_\sharp^+(-\rplus+\rminus) \log|\rplus+\rminus| + \psi_\flat^+(-\rplus+\rminus) + o(1)
& & \text{as $\rplus+\rminus\to 0^-$.}
\endaligned
\ee

\begin{lemma}[Poisson representation formula for the Fuchsian problem]
\label{lemma-Poisson}
The solution to the singular initial value problem~\eqref{eq:fuchsvp} within~$\Delta\cap\{\rplus+\rminus<0\}$ admits the representation formula
\bel{Poisson-rep}
\psi(\rplus,\rminus)
= {1 \over \pi} \int_{-1}^1 \biggl( \psi_\sharp^+(-\rplus+\lambda\rplus+\rminus+\lambda\rminus) \, \log\bigl(4(1-\lambda^2)\abs{\rplus+\rminus}\bigr) + \psi_\flat^+(-\rplus+\lambda\rplus+\rminus+\lambda\rminus) \biggr) \frac{d\lambda}{\sqrt{1-\lambda^2}} .
\ee
\end{lemma}

This lemma, established in Appendix~\ref{app:part2}, completes our construction of $\psi,\phi$ throughout the diamond~$\Delta$.
The explicit expression manifestly has finite limits on the future boundary of~$\Delta$, namely along the sides $\rminus=\rminus_1$ and $\rplus=\rplus_1$.
These values are then to be used as initial data for the next diamonds.

\subsection{Diamond with a timelike singularity hypersurface}
\label{sec:9sing2}

\paragraph{Stationary singular interface.}

We now reach the most difficult type of diamond, in which the singularity hypersurface is timelike.
Dealing with these diamonds is significantly more involved since each side of the singularity can no longer be handled independently from the other.
Instead, the singularity hypersurface at $A=0$ behaves as a (singular) interface with a genuine coupling between the two sides, and we will find that stronger conditions on the scattering map are required.

As we noted already, the momentum-preserving scattering maps~\eqref{plane-scattering-choice} that we selected respect the wave equation $A_{\uplus\uminus}=0$, so that its solution $A=(\rplus+\rminus)/2$ is known globally and the location of the singularity hypersurface $A=0$ is known a priori.
In the local null coordinates $\rplus,\rminus$ adapted to the problem in the given diamond, the interface is stationary, placed at a fixed position~$0$ in the spatial coordinate $\rplus+\rminus$.

Our approach involves many of the same formulas as in the previous section, but interpreted in a completely different way.
Just as in the case of a spacelike singularity, the singular hyperbolic equation~\eqref{eq:3-1-two} is obeyed independently by the matter field~$\phi$ and modular parameter~$\psi$ away from the singularity, while these fields may be mixed by the junction condition.

As explained and depicted in Figure~\ref{fig:subdiamonds} we can restrict our attention to a {\sl symmetric diamond}, in which the singularity passes through the two corners $(\rplus_1,\rminus_2)$ and $(\rplus_2,\rminus_1)$.  Without loss of generality (up to exchanging $\rplus\leftrightarrow\rminus$) we assume $\rplus'<0<\rminus'$ so that physical time flows toward negative~$\rplus$ and positive~$\rminus$.
Our conventions and some further notations are summarized as follows, with physical time flowing from the bottom to the top of the diagram:
\[
\Delta = (\rplus_1,\rplus_2)\times(\rminus_1,\rminus_2) , \qquad
\rplus_1 + \rminus_2 = \rplus_2 + \rminus_1 = 0 , \qquad\qquad
\mathtikz[scale=1.1]{\scriptsize
  \draw (0,1) node [above left=-2pt] {$(\rplus_1,\rminus_2)$}
  -- (1,0) node [midway,above right=-3pt] {$\Psiplus^-$}
  node [right=-2pt] {$(\rplus_2,\rminus_2)$.}
  -- (0,-1) node [midway,below right=-3pt] {$\Psiminus^-$}
  node [below left=-2pt] {$(\rplus_2,\rminus_1)$}
  -- (-1,0) node [midway,below left=-3pt] {$\Psiplus^+$}
  node [left=-2pt] {$(\rplus_1,\rminus_1)$}
  -- cycle node [midway,above left=-3pt] {$\Psiminus^+$};
  \draw (0,-1.3) -- (0,1.4) node [right=-2pt] {$\rplus+\rminus=0$};
  \draw[->,densely dashed] (-.35,-.85) -- (1,.4) node [above right=-2pt]{$\rminus$};
  \draw[->,densely dashed] (-.85,.35) -- (.4,-1) node [below right=-2pt]{$\rplus$};
}
\]

\paragraph{Solution on both sides of the singularity.}

Our strategy begins with expressing the solution $\psi(\rplus,\rminus)$ on both sides of the singularity, using the Abel representation formula~\eqref{gen-form}, in terms of both the prescribed initial data {\sl and the ``final data'' along future boundaries}.
As in the spacelike case we reserve the superscripts $\pm$ to denoting the $\pm(\rplus+\rminus)<0$ sides of the singularity.
In the region $\rplus+\rminus<0$ the formula involves the initial data $\psiplusG^+$ along $\rminus=\rminus_1$ and the final data $\psiminusG^+$ along $\rplus=\rplus_1$, while
in the region $\rplus+\rminus>0$ it involves the initial data $\psiminusG^-$ along $\rplus=\rplus_2$ and the final data $\psiplusG^-$ along $\rminus=\rminus_2$:
\bel{psitimelikebothsides}
\aligned
\psi(\rplus, \rminus)
& = \int_{\rplus_1}^{-\rminus} {- \Psiminus^-(-s)\over \sqrt{|\rplus - s| \, |\rminus + s|}} \, ds
+ \int_{\rminus_1}^{-\rplus} {- \Psiplus^-(-s) \over \sqrt{|\rplus + s|\,|\rminus - s|}} \, ds
\quad \text{for } \rplus+\rminus>0,
\\
\psi(\rplus, \rminus)
& = \int_{\rplus_1}^{\rplus\phantom{-}} {\Psiplus^+(s) \over \sqrt{|\rplus - s|\,|\rminus +s|}} \, ds
+ \int_{\rminus_1}^{\rminus\phantom{-}} {\Psiminus^+(s)\over \sqrt{|\rplus + s| \, |\rminus - s|}} \, ds
\quad \text{for } \rplus+\rminus<0,
\endaligned
\qquad
\aligned
\Psiplus^- & = \Abf_{\rplus_2\rminus_2}^{-1}[\psiplusG^-] , \\
\Psiminus^- & = \Abf_{\rminus_2\rplus_2}^{-1}[\psiminusG^-] , \\
\Psiplus^+ & = \Abf_{\rplus_1\rminus_1}^{-1}[\psiplusG^+] , \\
\Psiminus^+ & = \Abf_{\rminus_1\rplus_1}^{-1}[\psiminusG^+] ,
\endaligned
\ee
where we changed variables $s\to -s$ and used $\rplus_1+\rminus_2=\rplus_2+\rminus_1=0$ to rewrite the expressions for $\rplus+\rminus>0$.
We emphasize that $\psiplusG^+,\psiminusG^-$ are initial data (hence $\Psiplus^+,\Psiminus^-$ are known), while $\psiminusG^+,\psiplusG^-$ (hence $\Psiminus^+,\Psiplus^-$) are unknown at this stage.

To connect the two solutions along the singularity hypersurface, we expand~\eqref{psitimelikebothsides} in terms of Fuchsian data $(\psi_\sharp^{\pm},\psi_\flat^{\pm})$ on both sides, using Lemma~\ref{prop:evol-char} and its analogue on the other side.
We get
\bel{timelike-fuchs}
\aligned
\psi(\rplus, \rminus) & = \psi_\sharp^{\pm}(-\rplus+\rminus) \log|\rplus+\rminus| + \psi_\flat^{\pm}(-\rplus+\rminus) + o(1) \qquad \text{as $\rplus+\rminus\to 0^{\mp}$,}
\\
\psi_\sharp^{\pm}(2z) & = \mp \Psiplus^{\pm}(-z) \mp \Psiminus^{\pm}(z) ,
\\
\psi_\flat^-(2z) & = - \Psiminus^-(z) \log\bigl(4(-z-\rplus_1)\bigr)
+ \int_{\rplus_1}^{-z} \frac{\Psiminus^-(z)-\Psiminus^-(-s)}{-z-s} ds
- \Psiplus^-(-z) \log\bigl(4(z-\rminus_1)\bigr)
+ \int_{\rminus_1}^z \frac{\Psiplus^-(-z)-\Psiplus^-(-s)}{z-s} ds ,
\\
\psi_\flat^+(2z) & = \Psiplus^+(-z) \log\bigl(4(-z-\rplus_1)\bigr)
+ \int_{\rplus_1}^{-z} \frac{\Psiplus^+(s)-\Psiplus^+(-z)}{-z-s} ds
+ \Psiminus^+(z) \log\bigl(4(z-\rminus_1)\bigr)
+ \int_{\rminus_1}^{z} \frac{\Psiminus^+(s)-\Psiminus^+(z)}{z-s} ds .
\endaligned
\ee

\paragraph{The continuous junction condition is not causal.}
The formulas above apply to the matter field upon changing all~$\psi$ to~$\phi$ and $\Psi$ to~$\Phi$.
As a warm-up let us consider the simplest momentum-preserving scattering map, obtained by taking $f=0$ in~\eqref{plane-scattering-choice}.
It leads to the junction condition $(\phi_\sharp^+,\phi_\flat^+,\psi_\sharp^+,\psi_\flat^+) = (\phi_\sharp^-,\phi_\flat^-,\psi_\sharp^-,\psi_\flat^-)$, thus to solutions $\phi,\psi$ that have the same expansion on both sides of the singularity, simply changing $\rplus+\rminus$ to its opposite while keeping $-\rplus+\rminus$ fixed.  This is implemented by the symmetry $(\rplus,\rminus)\to(-\rminus,-\rplus)$ of the wave equation~\eqref{eq:3-1-two} obeyed by $\phi,\psi$.
A solution with such a symmetric expansion must thus itself obey the symmetry.

In particular this would lead to $\psiplusG^+(s)=\psiminusG^-(-s)$ and $\psiplusG^-(s)=\psiminusG^+(-s)$: initial data $\psiplusG^+,\psiminusG^-$ would be identified {\sl to each other}, which means that only certain classes of initial data are compatible with the junction condition.
This is completely unsuitable for our application to describing a timelike interface which appears dynamically in the collison of plane-symmetric gravitational waves.
In fact, this junction condition describes a $\ZZ/2\ZZ$ orbifold of spacetime.

Another point of view is instructive.
Imposing the junction conditions $\psi_\sharp^+=\psi_\sharp^-$ and $\psi_\flat^+ = \psi_\flat^-$ on~\eqref{timelike-fuchs} leads to equality $\Psiplus^+(s)=-\Psiminus^-(-s)$ and $\Psiplus^-(s)=-\Psiminus^+(-s)$ of inverse Abel transforms up to a sign.
(This can be argued for instance by noting that the latter equalities indeed lead to equal $\psi_\sharp^{\pm}$ and $\psi_\flat^{\pm}$, then using that $\Psiplus^{\pm},\Psiminus^{\pm}$ are uniquely determined by Fuchsian data.)
The relations between inverse Abel transforms then translate to relations $\psiplusG^+(s)=\psiminusG^-(-s)$ and $\psiplusG^-(s)=\psiminusG^+(-s)$ between values along boundaries of~$\Delta$.
The fact that different calculations match serves as a consistency check on various signs in formulas above.

\paragraph{Causal junction condition.}

As we just saw,
the momentum-preserving scattering map~$\Sbf$ given in~\eqref{plane-scattering-choice} is not causal for $f=0$.
By imposing that the evolution problem of interest has a solution for arbitrary initial data $\psiplusG^+,\psiminusG^-$, we now uncover the causality condition on~$f$, which we stated above as Definition~\ref{def:causal}.
The junction condition reads
\bel{junction-sec9}
\phi_\sharp^+ = \phi_\sharp^- , \qquad
\psi_\sharp^+ = \psi_\sharp^- , \qquad
\phi_\flat^+ = \phi_\flat^- + F_1(\phi_\sharp^-,\psi_\sharp^-) , \qquad
\psi_\flat^+ = \psi_\flat^- + F_2(\phi_\sharp^-,\psi_\sharp^-)
\ee
for a function $F=(F_1,F_2)\colon\RR^2\to\RR^2$ determined from~$f$ through~\eqref{plane-scattering-choice}.
As we will see momentarily, the causality condition is that $F$~is a bijection of~$\RR^2$.

Injecting~\eqref{timelike-fuchs} into $\psi_\sharp^+ = \psi_\sharp^-$ and rearranging terms yields
\[
\Psiplus^-(-z) + \Psiminus^+(z) = - \Psiplus^+(-z) - \Psiminus^-(z) ,
\]
which expresses a combination of (inverse Abel transforms of) outgoing data $\Psiplus^-,\Psiminus^+$ in terms of the initial data $\Psiplus^+,\Psiminus^-$.
The jump $\psi_\flat^+-\psi_\flat^-$ computed using~\eqref{timelike-fuchs} involves the two sums $\Psiplus^-(-z) + \Psiminus^+(z)$ and $\Psiplus^+(-z) + \Psiminus^-(z)$, which we just saw are opposite and determined solely from the initial data.
By changing $s\to-s$ and integrating by parts the jump simplifies to
\[
\psi_\flat^+(2z) - \psi_\flat^-(2z)
= - \del_z \int_{\rminus_1}^{-\rplus_1} \bigl(\Psiplus^+(-s)+\Psiminus^-(s)\bigr) \log|z-s|\,ds
\]
for all $z\in(\rminus_1,-\rplus_1)$, in which the derivative~$\del_z$ cannot be placed inside the integral as it would make the integrand divergent.
Together with the analogous expression for $\phi_\flat^+(2z)-\phi_\flat^-(2z)$, this expresses the left-hand side of the last junction condition in~\eqref{junction-sec9} explicitly in terms of the initial data $\psiplusG^+,\psiminusG^-$.

The junction condition~\eqref{junction-sec9} expresses the jump (that we just computed) as a function $F\colon\RR^2\to\RR^2$ of the leading Fuchsian data $(\phi_\sharp,\psi_\sharp)$.  Inverting the relation gives
\[
\bigl(\phi_\sharp^-(2z),\psi_\sharp^-(2z)\bigr)
=F^{-1}\bigl(\phi_\flat^+(2z) - \phi_\flat^-(2z),\psi_\flat^+(2z) - \psi_\flat^-(2z)\bigr)
\]
in terms of the inverse function $F^{-1}\colon\RR^2\to\RR^2$, provided $F$ is bijective.
If $F$ were not surjective, then only certain values of the jump would be allowed, rather than arbitrary $(\phi_\flat^+-\phi_\flat^-,\psi_\flat^+-\psi_\flat^-)(2z)\in\RR^2$, so that the existence of the singularity would put a restriction on the initial data; such a restriction would violate causality.
If $F$ were not injective, then there could be several allowed values of $(\phi_\sharp,\psi_\sharp)$ for a given initial data set, which would lead to a lack of predictive power.
This motivates our definition of a causal momentum-preserving scattering map in Definition~\ref{def:causal} as one for which $F$ is bijective.

Once $(\phi_\sharp,\psi_\sharp)$ are determined from the $(\phi_\flat,\psi_\flat)$ jumps, hence from the initial data, one easily uses~\eqref{timelike-fuchs} to deduce the outgoing data $\Psiplus^-,\Psiminus^+$, and finally the full solution~\eqref{psitimelikebothsides} on both sides of the singularity.
Altogether, we are done constructing the matter field~$\phi$ and the modular parameter~$\psi$ that solves the characteristic initial value problem in a diamond with a timelike singularity.

\subsection{Completing the construction}
\label{ssec:wrapup}

\paragraph{Construction of the conformal factor.}

At this stage, we know the solution coefficients $a,\phi,\psi$ for all $(\uplus,\uminus)\in\RR^2$.
To construct the conformal factor~$\omega$ in each regularity domain (in which $A\neq 0$) we integrate the differential equations~\eqref{eq-a1-baru} for~$\omega_\uplus$ and~\eqref{eq-a1-ubar} for~$\omega_\uminus$ along characteristics.
This determines the coefficient $\omega$ within any region where $A\neq 0$, up to a ``locally constant'' function.
Moreover, the compatibility condition $(\omega_\uplus)_\uminus=(\omega_\uminus)_\uplus$ and the wave equation~\eqref{eq-a2} for~$\omega_{\uplus\uminus}$ are then obeyed, thanks to the wave equations satisfied by $\phi,\psi$ as we now establish it.
By symmetry it suffices to check that the $\uminus$ derivative of \eqref{eq-a1-baru} gives~\eqref{eq-a2}:
\[
(\omega_\uplus)_\uminus
\overset{\eqref{eq-a1-baru}}{=} \biggl(\frac{A_{\uplus\uplus}}{2A_\uplus}\biggr)_\uminus - \biggl(\frac{A_\uplus}{4A}\biggr)_\uminus + \biggl(\frac{A}{A_\uplus}\bigl(\psi_\uplus^2 + 4 \pi \phi_\uplus^2\bigr)\biggr)_\uminus
\overset{\text{\eqref{eqxi}--\eqref{cphi-eq-scalarfield}}}{=}
\frac{A_\uplus A_\uminus}{4A^2} - \psi_\uplus\psi_\uminus - 4\pi\phi_\uplus\phi_\uminus .
\]
To fix the relative constants in disconnected components of the regularity domain $\{A\neq 0\}$, we rely on our junction conditions~\eqref{plane-scattering-choice} across each singularity hypersurface.
These junction conditions relate the expansion of $\upsilon=\omega-\frac{1}{2}\log|2\rplus'\rminus'|$ on both sides of each singularity hypersurface.
Thanks to our construction of the singularity scattering maps, the junction at each singularity hypersurface is fully consistent with the linear combinations of \eqref{eq-a1-baru} and~\eqref{eq-a1-ubar} that express the parallel derivative of~$\omega$ on both sides of the hypersurface (namely Einstein's constraint equations).

\paragraph{Cyclic spacetime.}

We now prove that our construction yields a cyclic spacetime in the sense of  \cite{LLV-1a}.
The underlying $4$-manifold is $\Mcal^4\simeq\RR^2\times\Tbb^2$ (or $\RR^4$), endowed with a global coordinate system $(\uplus,\uminus,x,y)$, a scalar field~$\phi$ and a Lorentzian metric
\[
g^{(4)} = -2\Omega d\uplus \,d\uminus + e^{2(a+\psi)} dx^2 + e^{2(a-\psi)} dy^2 ,
\]
where $\Omega,a,\psi,\phi$ are constructed in the previous sections.
The singular locus is $\Ncal^3=\Lscr=\{A=0\}$ defined in~\eqref{Lscr}.
By our genericity assumption of Definition~\ref{rem:non-degeneracy} this locus is the union of a collection of hypersurfaces, oriented by the choice of normal vector $-\nabla A$.
The exceptional locus $\Pcal^2\subset\Ncal^3$ is
\[
\Pcal^2 = \bigl\{(\uplus,\uminus,x,y)\in\Mcal^4\bigm|A(\uplus,\uminus)=0,\,A_\uplus(\uplus,\uminus) A_\uminus(\uplus,\uminus)=0\bigr\} ,
\]
which is $2$-dimensional by our genericity assumption. Importantly, all of the conditions for $(\Mcal^4,\Ncal^3,\Pcal^2,g^{(4)},\phi)$ to be a cyclic spacetime are obeyed by construction, as we now explain.

\bei 

\item {\bf Einstein equations.} 
The Einstein-scalar field equations~\eqref{equa-EinEq} hold outside the singular locus $\Ncal^3=\{A=0\}$, since this is how we constructed the metric components $A,\psi,\omega$ and scalar field~$\phi$ away from singularities.

\item {\bf Local foliations.}
Consider a point $(\uplus,\uminus,x,y)\in\Ncal^3\setminus\Pcal^2$, namely such that $A=0$, $A_\uplus\neq 0$, and $A_\uminus\neq 0$.  The latter two conditions state that the point belongs to the interior of a monotonicity diamond, and the first condition states that it lies on the spacelike or timelike singularity hypersurface $A=0$ within the diamond.  Upon constructing $\psi,\phi,\omega$ in Sections~\ref{sec:9sing1} and~\ref{sec:9sing2} we found their logarithmic singularities near a singularity hypersurface, from which we had determined in Section~\ref{sec:82} the behaviour of geodesics normal to the singularity hypersurface.
For a sufficiently small neighborhood~$\Hcal_0$ of $(\uplus,\uminus,x,y)$ inside $\Ncal^3 \setminus\Pcal^2$, the normal geodesics originating from~$\Hcal_0$ are well-defined for a sufficiently small interval $(s_{-1},s_1)\ni 0$ of proper time or distance coordinate~$s$.
The union of these geodesic segments is manifestly a neighborhood of $(\uplus,\uminus,x,y)$ foliated by constant-$s$ hypersurfaces~$\Hcal_s$ and the metric is $g^{(4)}=\pm ds^2+g(s)$ by construction.  We choose the sign of $s$ to be that of $-A$, which corresponds to choosing the hypersurface's orientation such that $-\nabla A$ is a positively-oriented normal vector.

\item {\bf Singularity behavior.}
We also determined in Section~\ref{sec:82} the asymptotic profiles on both sides of the singularity, hence the singularity data sets $(\gmoi, \Kmoi, \phi_{0-}, \phi_{1-})$ and $(\gpoi, \Kpoi, \phi_{0+}, \phi_{1+})$, expressed in~\eqref{sharpflat-to-01} in terms of Fuchsian data describing the expansions of $\phi,\psi,\omega$ near the singularity (in this equation the subscripts $\pm$ were dropped to lighten notations).

\item {\bf Scattering conditions.}
In Section~\ref{sec:83} we related Fuchsian data describing the expansions of $\phi,\psi,\omega$ on both sides of the singularity by translating the junction condition
\bel{junction-last}
(\gpoi, \Kpoi, \phi_{0+}, \phi_{1+}) = \Sbf(\gmoi, \Kmoi, \phi_{0-}, \phi_{1-})
\ee
to Fuchsian data.  We then used the resulting relation between Fuchsian data to determine $\phi,\psi,\omega$ beyond singularity hypersurfaces.
Thus, the junction condition~\eqref{junction-last} holds by construction on each singularity hypersurface.

\eei

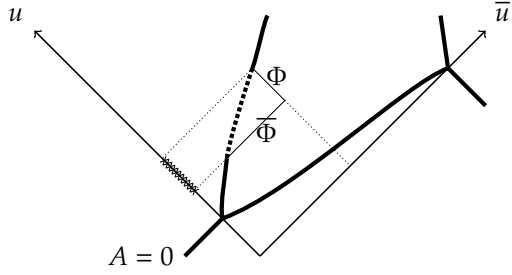
\begin{figure}\centering
  \begin{tikzpicture}
    \draw[semithick,->] (0,0) -- (-3,3) node [above left] {$\uplus$};
    \draw[semithick,->] (0,0) -- (3,3) node [above right] {$\uminus$};
    \draw[ultra thick] (-.5,.5) -- (-1,0) node [left] {$A=0$};
    \draw[ultra thick] (2.5,2.5) -- +(.5,-.5);
    \draw[ultra thick] (.1,3.2) .. controls +(-120:.1) and +(70:.2) .. (-.1,2.5);
    \draw[ultra thick,densely dotted] (-.1,2.5) .. controls +(-110:.2) and +(85:.2) .. (-.44,1.3);
    \draw[ultra thick] (-.44,1.3)
    .. controls +(-95:.2) and +(100:.2) .. (-.5,.5)
    .. controls +(20:1) and +(-160:.6) .. (2.5,2.5) -- (2.4,3.2);
    \draw(-.1,2.5)--(.33,2.07) node[midway,xshift=3pt,yshift=3pt]{$\Phiplus$}--(-.44,1.3) node[midway,xshift=4pt]{$\Phiminus$};
    \draw[densely dotted] (-.1,2.5) -- (-1.3,1.3);
    \draw[densely dotted] (.33,2.07) -- (1.2,1.2);
    \draw[densely dotted] (-.44,1.3) -- (-.87,.87);
    \foreach\x in{.87,.91,...,1.3} {
      \draw(-\x-.02,\x-.06)--+(.04,.12);
      \draw(-\x-.06,\x-.02)--+(.12,.04);
    }
  \end{tikzpicture}
  \caption{\label{fig:generic-curvature-sing}{\bf Non-generic plane-symmetric spacetime} in which Fuchsian data $\phi_\sharp$ vanishes on some interval (thick dotted segment) within the $A=0$ locus (thick lines).  By Lemma~\ref{prop:evol-char} the data $\phi_\sharp=\Phiplus+\Phiminus$ where $\Phiplus,\Phiminus$ are computed from values of~$\phi$ along certain null segments (thin lines).  In this configuration, $\Phiplus$ depends on more initial data than $\Phiminus$; perturbations of the initial data in the corresponding interval (hatched line) only affect~$\Phiplus$, hence do not leave $\phi_\sharp=0$ on the given interval.}
\end{figure}

\paragraph{Curvature is generically singular.}

The last statement of Theorem~\ref{theorem-plane} that remains to be proven is that curvature generically blows up as one approaches each singularity hypersurface $2A=\rplus+\rminus=0$.
Specifically, let us consider the (four-dimensional) Ricci scalar~$R^{(4)}$.
The Einstein-scalar field equations express it as the norm of $\nabla\phi$.  We write the expression in $\rplus,\rminus$ coordinates and use the expansions~\eqref{asympt} and~\eqref{alphasharpflat} to get
\[
R^{(4)} = \mp 2\pi e^{-2\upsilon} \phi_\rplus \phi_\rminus
= \mp 2\pi e^{-2\upsilon_\flat} \phi_\sharp^2 \,|\rplus+\rminus|^{-2(1+\upsilon_\sharp)} (1+o(1)) , \quad \rplus+\rminus\to 0 ,
\]
where the sign depends on the monotonicity diamond, more precisely on the spacelike or timelike nature of the singularity hypersurface.
Since $1+\upsilon_\sharp=3/4+\psi_\sharp^2+4\pi\phi_\sharp^2>0$, the Ricci scalar blows up, except perhaps at points of the singularity where $\phi_\sharp=0$.
We recall that the junction condition states that $(\phi_\sharp,\psi_\sharp,\upsilon_\sharp)$ are the same on both sides of the singularity so we can freely switch from one side of the singularity to the other in that regard.

Let us show that points with $\phi_\sharp=0$ are generically isolated, namely that $\phi_\sharp$ generically does not vanish on a whole interval.
For this, we assume that $\phi_\sharp$ vanishes on an interval and show that this situation is unstable.
Consider the symmetric diamond~$\Delta$ whose diagonal is the interval on which $\phi_\sharp$ vanishes identically, and apply Lemma~\ref{prop:evol-char} to this diamond: it expresses the Fuchsian data schematically as $\phi_\sharp=\Phiplus+\Phiminus$, as a sum of contributions from two null boundaries of~$\Delta$ lying on the same side of the singularity hypersurface.
If the hypersurface is timelike these boundaries are a past and a future boundary, while if the hypersurface is spacelike we can choose to use the two past boundaries, as depicted in Figure~\ref{fig:generic-curvature-sing}.
In either case, the two null boundaries have different causal pasts, so that $\Phiplus$ and $\Phiminus$ depend on different subsets of the initial data.
Perturbations of the initial data in suitable intervals thus affect only one of $\Phiplus$ and~$\Phiminus$, so that $\phi_\sharp=\Phiplus+\Phiminus$ cannot remain identically vanishing under such generic perturbations.
This ends the proof that generically $R^{(4)}$ blows up at generic points of singularity hypersurfaces, and thereby concludes the proof of Theorem~\ref{theorem-plane}.


\small

\paragraph*{Acknowledgments.} 
 
The authors gratefully acknowledge support from the Simons Center for Geometry and Physics, Stony Brook University. 
The second author (PLF)  
 is  grateful to the Institut Mittag-Leffler where he attended the Program: ``General relativity, geometry and analysis: beyond the first 100 years after Einstein''.


\small

\appendix 

\section{Derivation of representation formulas}
\label{app:part2} 

\paragraph{Proof of Lemma~\ref{lemmadef:Abel}.}

Let us first show that $\Abf_{\rho\sigma}\bigl[\Abf_{\rho\sigma}^{-1}[f]\bigr](r)=f(r)$ for any smooth function $f\colon(r_1,r_2)\to\RR$ and any $r\in(r_1,r_2)$.
We compute the left-hand side as follows:
\[
\aligned
\Abf_{\rho\sigma}\bigl[\Abf_{\rho\sigma}^{-1}[f]\bigr](r)
& = \frac{1}{\pi} \int_\rho^r \frac{\sgn(s-\rho)}{\sqrt{|r-s|}}  \, \frac{d}{ds} \int_\rho^s \frac{\,f(t)}{\sqrt{|s-t|}} \, dt \, ds & & (\text{a factor of~$\sqrt{|\sigma+s|}$ cancels out}),
\\
& = \frac{\sgn(r-\rho)}{\pi} \int_\rho^r \frac{1}{\sqrt{|r-s|}}  \, \frac{d}{ds} \int_{\rho-s}^0 \frac{f(s+t)}{\sqrt{|t|}} \, dt \, ds & & (t\to s+t),
\endaligned
\]
thus
\[
\aligned
& \Abf_{\rho\sigma}\bigl[\Abf_{\rho\sigma}^{-1}[f]\bigr](r)
= \frac{\sgn(r-\rho)}{\pi} \int_\rho^r \frac{1}{\sqrt{|r-s|}}  \, \biggl(
\frac{f(\rho)}{\sqrt{\rho-s|}} + \int_{\rho-s}^0 \frac{f'(s+t)}{\sqrt{|t|}} \, dt
\biggr) \, ds & & (\text{compute derivative}),
\\
& \qquad = f(\rho) \frac{\sgn(r-\rho)}{\pi} \int_\rho^r \frac{1}{\sqrt{|r-s||\rho-s|}} \, ds
+ \int_\rho^r f'(t) \frac{\sgn(r-\rho)}{\pi} \int_t^r \frac{1}{\sqrt{|r-s||s-t|}} \, ds dt
& & (t\to -s+t \text{ and swap integrals}),
\\
& \qquad = f(\rho) + \int_\rho^r f'(t) dt = f(r),
\endaligned
\]
where to get the last line we used that $\int_t^r ds/\sqrt{|r-s||s-t|}=\pi\sgn(r-t)=\pi\sgn(r-\rho)$.
Conversely let us show that $\Abf_{\rho\sigma}^{-1}\bigl[\Abf_{\rho\sigma}[F]\bigr]=F$:
\[
\aligned
\Abf_{\rho\sigma}^{-1}\bigl[\Abf_{\rho\sigma}[F]\bigr]
& = \sgn(r-\rho) \frac{1}{\pi} \sqrt{|\sigma+r|} \, \frac{d}{dr}
\int_\rho^r \int_\rho^s \frac{F(t)}{\sqrt{|r-s||s-t||\sigma+t|}} \, dt \, ds ,
\\
& = \sgn(r-\rho) \frac{1}{\pi} \sqrt{|\sigma+r|} \, \frac{d}{dr}
\int_\rho^r \frac{F(t)}{\sqrt{|\sigma+t|}} \int_t^r \frac{1}{\sqrt{|r-s||s-t|}} \, ds \, dt
& & \text{(swap integrals),}
\\
& = \sqrt{|\sigma+r|} \, \frac{d}{dr} \int_\rho^r \frac{F(t)}{\sqrt{|\sigma+t|}} \, dt = F(r)
& & \text{(compute the $s$ integral).}
\endaligned
\]

\paragraph{Proof of Lemma~\ref{lemma-Abel}.}

We recall that we work in a characteristic domain~$D$ that lies entirely on one side of the line $\rplus+\rminus=0$.  Let us denote $\epsilon=\sgn(\rplus+\rminus)$.
We showed in the main text that $\psi$ given in~\eqref{gen-form} assumes the prescribed values on suitable boundaries.  Our task now is to show that it obeys the singular wave equation~\eqref{eq:3-1-two}, which we write as $2|\rplus+\rminus|\psi_{\rplus\rminus}+\epsilon(\psi_\rminus+\psi_\rplus) = 0$.
The wave equation is linear so we simply need to show each of the two terms in~\eqref{gen-form} obey it.
In addition, these two terms are mapped to each other under the symmetry $\rplus\leftrightarrow\rminus$ of the wave equation, so we simply need to treat the $\Psiplus$ term (equivalently we set $\psinot=0$ and $\Psiminus=0$).
The $\rminus$ derivative is easy:
\[
\psi(\rplus, \rminus) = \int_{\rplus_j}^{\rplus} \frac{\Psiplus(s)}{|\rplus-s|^{1/2}\,|\rminus+s|^{1/2}} \, ds , \qquad
\psi_\rminus(\rplus, \rminus)
= - \frac{\epsilon}{2} \int_{\rplus_j}^{\rplus} \frac{\Psiplus(s)}{|\rplus-s|^{1/2}\,|\rminus+s|^{3/2}} \, ds .
\]
For the $\rplus$ derivative we change variables as $s=\rplus+t$, so that the variable integration bound is not the one which lies at a singularity of the integrand, then after taking the derivative we change back to~$s$.  (Intuitively, we take the $\del_\rplus+\del_s$ derivative, or the $\rplus$ derivative ``at fixed $s-\rplus$''.)
This yields
\[
\psi_\rplus(\rplus, \rminus)
= \frac{\Psiplus(\rplus_j)}{|\rplus_j-\rplus|^{1/2}\,|\rminus+\rplus_j|^{1/2}}
+ \int_{\rplus_j}^{\rplus} \frac{1}{|s-\rplus|^{1/2}}\del_s\biggl(\frac{\Psiplus(s)}{|\rminus+s|^{1/2}}\biggr) \, ds
\]
and its $\rminus$ derivative
\[
\psi_{\rplus\rminus}(\rplus, \rminus)
= - \frac{\epsilon}{2} \frac{\Psiplus(\rplus_j)}{|\rplus_j-\rplus|^{1/2}\,|\rminus+\rplus_j|^{3/2}}
- \frac{\epsilon}{2} \int_{\rplus_j}^{\rplus} \frac{1}{|s-\rplus|^{1/2}} \del_s \biggl(\frac{\Psiplus(s)}{|\rminus+s|^{3/2}}\biggr)\, ds .
\]
We are ready to compute.  First we collect together terms without integrals, and terms with $\del_s$ derivatives, then, in a second line, we use $\epsilon(|\rminus+s|-|\rplus+\rminus|)=s-\rplus$ and integrate by parts:
\begin{align*}
& 2|\rplus+\rminus|\psi_{\rplus\rminus}+\epsilon(\psi_\rplus+\psi_\rminus)
= \frac{\epsilon(|\rminus+\rplus_j|-|\rplus+\rminus|)\Psiplus(\rplus_j)}{|\rplus_j-\rplus|^{1/2}\,|\rminus+\rplus_j|^{3/2}}
+ \int_{\rplus_j}^{\rplus} \frac{1}{|s-\rplus|^{1/2}} \del_s \biggl(\frac{\epsilon(|\rminus+s| - |\rplus+\rminus| )\Psiplus(s)}{|\rminus+s|^{3/2}}\biggr)\, ds
- \frac{1}{2} \int_{\rplus_j}^{\rplus} \frac{\Psiplus(s)}{|\rplus-s|^{1/2}\,|\rminus+s|^{3/2}} \, ds ,
\\
& \qquad = \frac{(\rplus_j-\rplus)\Psiplus(\rplus_j)}{|\rplus_j-\rplus|^{1/2}\,|\rminus+\rplus_j|^{3/2}}
+ \biggl[\frac{(s-\rplus)\Psiplus(s)}{|s-\rplus|^{1/2}\,|\rminus+s|^{3/2}}\biggr]_{\rplus_j}^{\rplus}
- \int_{\rplus_j}^{\rplus} \frac{(s-\rplus)\Psiplus(s)}{|\rminus+s|^{3/2}} \del_s \biggl(\frac{1}{|s-\rplus|^{1/2}}\biggr)\, ds
- \frac{1}{2} \int_{\rplus_j}^{\rplus} \frac{\Psiplus(s)}{|\rplus-s|^{1/2}\,|\rminus+s|^{3/2}} \, ds = 0 .
\end{align*}
Our final step is to notice that the first two terms cancel, as do the last two terms.  This concludes our proof that $\psi$ given in~\eqref{gen-form} obeys the wave equation.

\paragraph{Proof of Lemma~\ref{prop:evol-char}.}

We consider the solution $\psi$ given in~\eqref{gen-form} in the case of a symmetric diamond ($\rplus_1+\rminus_2=\rplus_2+\rminus_1=0$) where data are prescribed along the $\rminus=\rminus_2$ and $\rplus=\rplus_2$ boundaries.  We wish to expand the solution near the singularity $\rplus+\rminus\to 0^+$.
For this we will switch to coordinates $t=(\rplus+\rminus)/2>0$ and $z=(-\rplus+\rminus)/2\in(-\rplus_2,\rminus_2)$ as needed.
As in the proof of Lemma~\ref{lemma-Abel} above, we treat the term $\Psiplus^-$ only, setting $\Psiminus^-=0$ as the corresponding term is entirely analogous to~$\Psiplus^-$.
We thus wish to expand (as $t\to 0^-$)
\[
\psi(\rplus,\rminus)
= - \int_{\rplus}^{\rplus_2} \frac{\Psiplus^-(s)\,ds}{\sqrt{(s-\rplus)(\rminus+s)}}
= - \int_{t-z}^{\rplus_2} \frac{\Psiplus^-(s)\,ds}{\sqrt{(z+s)^2-t^2}}.
\]
We split $\Psiplus^-(s)$ into $\Psiplus^-(-z)$ and $\Psiplus^-(s)-\Psiplus^-(-z)$ and treat the two parts separately.  First, by the change of coordinates $s=-z+t\sigma$ we obtain
\[
- \int_{t-z}^{\rplus_2} \frac{\Psiplus^-(-z)\,ds}{\sqrt{(z+s)^2-t^2}}
= - \Psiplus^-(-z) \int_1^{(\rplus_2+z)/t} \frac{d\sigma}{\sqrt{\sigma^2-1}}
= - \Psiplus^-(-z) \operatorname{arccosh}\biggl(\frac{|\rplus_2+z|}{|t|}\biggr)
= \Psiplus^-(-z) \log\biggl(\frac{|t|}{2|\rplus_2+z|}\biggr) + O(t^2) .
\]
The second part has a manifestly finite limit as $t\to 0$:
\[
\int_{t-z}^{\rplus_2} \frac{\Psiplus^-(-z)-\Psiplus^-(s)}{\sqrt{(z+s)^2-t^2}} ds
= \int_{-z}^{\rplus_2} \frac{\Psiplus^-(-z)-\Psiplus^-(s)}{z+s} ds
+ o(1) .
\]
This yields the desired terms in the expansion (from which the first correction term can be computed to be $O(t^2\log|t|)$): 
\[
\psi(\rplus,\rminus)
= \Psiplus^-(-z) \log|\rplus+\rminus|
- \Psiplus^-(-z) \log(4|\rplus_2+z|)
+ \int_{-z}^{\rplus_2} \frac{\Psiplus^-(-z)-\Psiplus^-(s)}{z+s} ds + o(1) .
\]

\paragraph{Proof of Lemma~\ref{lemma-Poisson}.}

For this lemma we work in the region $\rplus+\rminus>0$.
The claim in the lemma is two-fold: first, that $\psi(\rplus,\rminus)$ given in~\eqref{Poisson-rep} is a solution of the wave equation; second, that it has the expected asymptotics at the singularity.
By linearity we can treat the $\phi_\flat^+$ and $\phi_\sharp^+$ terms in turn.
First, the $\phi_\flat^+$ term (i.e.\ take $\phi_\sharp^+=0$)
\[
\psi(\rplus,\rminus)
= {1 \over \pi} \int_{-1}^1 \psi_\flat^+(-\rplus+\lambda\rplus+\rminus+\lambda\rminus) \frac{d\lambda}{\sqrt{1-\lambda^2}} .
\]
Showing it obeys the wave equation is straightforward:
\[
2 (\rplus+\rminus)\psi_{\rplus\rminus} + \psi_\rplus + \psi_\rminus
= {1 \over \pi} \int_{-1}^1 \biggl(
2 (\rplus+\rminus) \sqrt{1-\lambda^2}\,\psi_\flat^{+\prime\prime}
+ \frac{-1+\lambda}{\sqrt{1-\lambda^2}}\,\psi_\flat^{+\prime}
+ \frac{1+\lambda}{\sqrt{1-\lambda^2}}\,\psi_\flat^{+\prime}
\biggr) \, d\lambda
= {2 \over \pi} \int_{-1}^1 \del_\lambda \bigl( \sqrt{1-\lambda^2}\,\psi_\flat^{+\prime}\bigr) \, d\lambda = 0 ,
\]
where derivatives of~$\psi_\flat^+$ are evaluated at $-\rplus+\lambda\rplus+\rminus+\lambda\rminus$.
Taking the $t=(\rplus+\rminus)/2\to 0$ limit at fixed $z=(-\rplus+\rminus)/2$ is even easier:
\[
\psi(-z,z)
= {1 \over \pi} \int_{-1}^1 \psi_\flat^+(2z) \frac{d\lambda}{\sqrt{1-\lambda^2}}
= \psi_\flat^+(2z) ,
\]
where the integral is evaluated by writing $\lambda=\sin\theta$ with $\theta\in[-\pi/2,\pi/2]$.
Second, the $\phi_\sharp^+$ term (i.e.\ take $\phi_\flat^+=0$)
\[
\psi(\rplus,\rminus)
= {1 \over \pi} \int_{-1}^1 \psi_\sharp^+(-\rplus+\lambda\rplus+\rminus+\lambda\rminus) \, \log\bigl(4(1-\lambda^2)(\rplus+\rminus)\bigr) \frac{d\lambda}{\sqrt{1-\lambda^2}} .
\]
Calculations are more tedious than above.  Again we omit the arguments of derivatives of~$\phi_\sharp^+$ for brevity.  First we compute derivatives:
\[
\aligned
\psi_{\rplus} & = {1 \over \pi} \int_{-1}^1 \biggl(
(-1+\lambda) \log\bigl(4(1-\lambda^2)(\rplus+\rminus)\bigr) \, \psi_\sharp^{+\prime}
+ \frac{1}{\rplus+\rminus}\,\psi_\sharp^+
\biggr)\, \frac{d\lambda}{\sqrt{1-\lambda^2}} ,
\\
\psi_{\rminus} & = {1 \over \pi} \int_{-1}^1 \biggl(
(1+\lambda) \log\bigl(4(1-\lambda^2)(\rplus+\rminus)\bigr) \, \psi_\sharp^{+\prime}
+ \frac{1}{\rplus+\rminus}\,\psi_\sharp^+
\biggr)\, \frac{d\lambda}{\sqrt{1-\lambda^2}} ,
\\
\psi_{\rplus\rminus} & = {1 \over \pi} \int_{-1}^1 \biggl(
- (1-\lambda^2) \log\bigl(4(1-\lambda^2)(\rplus+\rminus)\bigr) \, \psi_\sharp^{+\prime\prime}
+ (1+\lambda) \frac{1}{\rplus+\rminus} \, \psi_\sharp^{+\prime}
+ (-1+\lambda) \frac{1}{\rplus+\rminus} \, \psi_\sharp^{+\prime}
- \frac{1}{(\rplus+\rminus)^2}\,\psi_\sharp^+
\biggr) \frac{d\lambda}{\sqrt{1-\lambda^2}} .
\endaligned
\]
Then we combine them and collect terms as
\[
\aligned
2(\rplus+\rminus)\psi_{\rplus\rminus} + \psi_{\rplus} + \psi_{\rminus}
& = {2 \over \pi} \int_{-1}^1 \biggl(
\lambda \log\bigl(4(1-\lambda^2)(\rplus+\rminus)\bigr) \, \psi_\sharp^{+\prime}
+ 2\lambda \psi_\sharp^{+\prime}
- (\rplus+\rminus) (1-\lambda^2) \log\bigl(4(1-\lambda^2)(\rplus+\rminus)\bigr) \, \psi_\sharp^{+\prime\prime}
\biggr)\, \frac{d\lambda}{\sqrt{1-\lambda^2}} ,
\\
& = {2 \over \pi} \int_{-1}^1 \del_\lambda \biggl(
- \sqrt{1-\lambda^2} \log\bigl(4(1-\lambda^2)(\rplus+\rminus)\bigr) \, \psi_\sharp^{+\prime}
\biggr)\, d\lambda = 0 .
\endaligned
\]
To take the $t=(\rplus+\rminus)/2\to 0$ limit at fixed $z=(-\rplus+\rminus)/2$ we split the logarithm into $\log(\rplus+\rminus)+\log(4(1-\lambda^2))$.  The first piece leads to exactly the same calculations as for $\psi_\flat^+$, while the second piece turns out to give a vanishing contribution.
\[
\aligned
\psi(\rplus,\rminus)
& = \log(\rplus+\rminus) {1 \over \pi} \int_{-1}^1 \psi_\sharp^+(-\rplus+\lambda\rplus+\rminus+\lambda\rminus) \frac{d\lambda}{\sqrt{1-\lambda^2}}
+ {1 \over \pi} \int_{-1}^1 \psi_\sharp^+(-\rplus+\lambda\rplus+\rminus+\lambda\rminus) \, \log(4(1-\lambda^2)) \frac{d\lambda}{\sqrt{1-\lambda^2}} ,
\\
& = \psi_\sharp^+(2z) \log(\rplus+\rminus) {1 \over \pi} \int_{-1}^1 \frac{d\lambda}{\sqrt{1-\lambda^2}}
+ \psi_\sharp^+(2z) {1 \over \pi} \int_{-1}^1 \log(4(1-\lambda^2)) \frac{d\lambda}{\sqrt{1-\lambda^2}} + o(1)
= \psi_\sharp^+(2z) \log(\rplus+\rminus) + o(1) ,
\endaligned
\]
where the second integral is evaluated by changing $\lambda=\sin\theta$ and recognizing a famous integral $2\int_{-\pi/2}^{\pi/2} \log(2\cos\theta) d\theta=0$.
This concludes our proof of Lemma~\ref{lemma-Poisson}.

\end{document}